\documentclass[11pt]{article}

\usepackage{amsfonts,latexsym,amsthm,amssymb,amsmath,amscd,euscript,tikz,mathtools}
\usepackage{framed}
\usepackage[margin=1in]{geometry}
\usepackage{color}
\usepackage[colorlinks=true,citecolor=blue,linkcolor=blue]{hyperref}
\usepackage{enumitem}
\usepackage{siunitx}
\usepackage{textcomp}
\usepackage{physics}
\usepackage{mathrsfs}
\usepackage{dsfont}
\usepackage[ruled,vlined,linesnumbered]{algorithm2e}
\usepackage{titlesec}
\usepackage{tikz-cd}
\usepackage{thmtools}
\usepackage{thm-restate}
\usepackage{cleveref}
\usepackage{stmaryrd}
\usepackage{graphicx}
\usetikzlibrary{positioning,chains,fit,shapes,calc}

\allowdisplaybreaks[1]

\newtheorem{theorem}{Theorem}[section]
\newtheorem{proposition}[theorem]{Proposition}
\newtheorem{lemma}[theorem]{Lemma}
\newtheorem{claim}[theorem]{Claim}
\newtheorem{corollary}[theorem]{Corollary}

\newtheorem{definition}[theorem]{Definition}

\newtheorem{remark}[theorem]{Remark}

\theoremstyle{remark}

\newcommand{\cA}{\mathcal{A}}\newcommand{\cB}{\mathcal{B}}
\newcommand{\cC}{\mathcal{C}}
\newcommand{\cE}{\mathcal{E}}\newcommand{\cF}{\mathcal{F}}

\newcommand{\cQ}{\mathcal{Q}}

\newcommand{\bC}{\mathbb{C}}
\newcommand{\bF}{\mathbb{F}}

\newcommand{\bN}{\mathbb{N}}

\newcommand{\bR}{\mathbb{R}}

\newcommand{\bZ}{\mathbb{Z}}

\newcommand{\1}{\mathds{1}}

\newcommand{\poly}{\operatorname{poly}}

\newcommand{\Enc}{\operatorname{Enc}}

\newcommand{\evl}{\operatorname{ev}}

\newcommand{\filler}{\cdot}

\newcommand{\nc}{\newcommand}

\nc{\on}{\operatorname}
\nc{\Spec}{\on{Spec}}
\nc{\Aut}{\textit{Aut}}
\nc{\id}{\textit{id}}
\nc{\chr}{\on{char}}
\nc{\im}{\on{im}}
\nc{\Hom}{\on{Hom}}
\nc{\lcm}{\on{lcm}}
\nc{\dual}[1]{\prescript{t}{}{#1}}
\nc{\transpose}[1]{{#1}^{\intercal}}
\nc{\Sym}{\on{Sym}}
\nc{\End}{\on{End}}
\nc{\stab}{\on{stab}}
\nc{\Li}{\on{Li}}
\nc{\spn}{\on{span}}
\nc{\sgn}{\on{sgn}}
\nc{\supp}{\on{supp}}
\nc{\Unif}{\on{Unif}}

\makeatletter
\newcommand\footnoteref[1]{\protected@xdef\@thefnmark{\ref{#1}}\@footnotemark}
\makeatother

\title{Near-Asymptotically-Good Quantum Codes with Transversal CCZ Gates and Sublinear-Weight Parity-Checks\thanks{Research supported in part by a ONR grant N00014-24-1-2491 and a UC Noyce initiative award. L.~Golowich acknowledges support from a National Science Foundation Graduate Research Fellowship under Grant No.~DGE 2146752.}}

\author{Louis Golowich \\
  UC Berkeley \\
  \href{mailto:lgolowich@berkeley.edu}{\texttt{lgolowich@berkeley.edu}}
  \and
  Venkatesan Guruswami \\
  UC Berkeley \\
  \href{mailto:venkatg@berkeley.edu}{\texttt{venkatg@berkeley.edu}}
}

\begin{document}

\pagenumbering{gobble}

\maketitle
\thispagestyle{empty}

\begin{abstract}
  It is a major challenge to construct good quantum codes supporting fault-tolerant (e.g.~transversal) non-Clifford gates with low-weight parity-check measurements. In this paper, we construct the first known quantum codes with linear dimension and distance supporting transversal non-Clifford gates that have sublinear locality (i.e.~parity-check weight). Specifically, we construct codes with transversal $CCZ$ gates that have dimension and distance $\Theta(N)$ and locality $O(\sqrt{N})$, where $N$ denotes the block length. We furthermore design an efficient decoding algorithm for these codes. The alphabet size of these codes is $q=\Theta(\sqrt{N})$, but it can be reduced to a constant (e.g.~$q=2$) while incurring a polylogarithmic loss in other parameters. We also show how to decrease the locality to $O(N^{1/3})$, albeit with a larger alphabet size and slightly lower distance.

\smallskip
  We construct these codes as products of classical codes with appropriate algebraic structure. While our quantum codes are subsystem codes with non-commuting gauge operators, we show they nevertheless permit error correction from noisy syndrome measurements.

\smallskip
  As byproducts, we prove multiple technical results of independent interest. In particular, our efficient decoder can be viewed as a new multivariate generalization of Prony's method for reconstructing a function from partial access to its Fourier transform. Meanwhile, our distance analysis involves new connections to the classical study of maximally recoverable codes. Our results on product codes also resolve a conjecture of Bravyi \& Hastings (2014) in the large-alphabet regime, by providing a new construction of quantum codes with dimension and distance $\Theta(N)$ and locality $N^\epsilon$ for arbitrary $\epsilon>0$.

\end{abstract}

\newpage

\tableofcontents

\newpage

\pagenumbering{arabic}

  

\section{Introduction}


A central problem in quantum computing lies in the construction of quantum error-correcting codes that permit efficient fault-tolerant computation. In particular, \emph{non-Clifford gates} present one of the most challenging components of a quantum computation to render fault-tolerant. Indeed, circuits consisting solely of Clifford gates are efficiently classically simulable, and are more easily performed fault-tolerantly. Universal quantum computation can be obtained by adding a single non-Clifford gate, such as $CCZ$ or $T$, to the Clifford gate set.

Fault-tolerant non-Clifford gates are often obtained by applying the gate transversally (i.e.~in a depth-1 circuit) to a carefully constructed quantum code on which this physical circuit induces the desired logical gate. Yet such codes have proven difficult to construct, and there are no known codes supporting transversal non-Clifford gates that achieve even close to the asymptotically optimal parameters for dimension (message length), distance (error resilience), and locality (parity-check weight).

Indeed, while the recent line of work \cite{krishna_towards_2019,wills_constant-overhead_2024,golowich_asymptotically_2025,nguyen_good_2025} gave the first known asymptotically good $[[N,\Theta(N),\Theta(N)]]$\footnote{An $[[N,K,D]]_q$ code refers to a quantum code of length $N$, dimension $K$, and distance $D$ with alphabet size (i.e.,~local qudit dimension) $q$.} constructions, meaning that the dimension and distance scale linearly with the block length $N$, these codes have poor locality, as they require parity-check (i.e.,~stabilizer) measurements of linear weight. This high check weight induces a significant space-time overhead for performing error correction on these codes (see e.g.,~\cite{nguyen_quantum_2024}). Fault-tolerant error correction on these codes therefore also requires some sort of concatenation to protect against errors during the parity-check measurements, which require large circuits to perform.

Codes supporting transversal non-Clifford gates with constant or polylogarithmic locality have also been constructed \cite{bombin_exact_2007,bombin_topological_2007,bombin_gauge_2015,zhu_non-clifford_2023,scruby_quantum_2024,golowich_quantum_2024,zhu_topological_2025,breuckmann_cups_2024}. However, all such known constructions have distance at most\footnote{The recent manuscript \cite{zhu_transversal_2025} claimed a construction of qLDPC codes of distance $\Theta(N^{2/3})$ with transversal $CCZ$ gates. However, we have become aware of a gap in the presented distance analysis. We believe this issue can likely be resolved with a modified construction, but an such an update has not yet been published.} $O(\sqrt{N})$, and typically have poor dimension as well. Specifically, the best previously known parameters among codes supporting transversal non-Clifford gates with locality $\leq\sqrt{N}$ are given by the $[[N,\Theta(\sqrt{N}),\Theta(\sqrt{N})]]$ codes with constant locality of \cite{zhu_topological_2025}, along with the $[[N,N^{1-\epsilon},\tilde{\Theta}(N^{1/3})]]$ codes with polylogarithmic locality of \cite{golowich_quantum_2024}.

In this paper, we construct the first known quantum codes with transversal non-Clifford gates that have linear dimension and distance, and sublinear locality. Specifically, we construct such codes with locality $O(\sqrt{N})$ over an alphabet of size $O(\sqrt{N})$, and we design an efficient decoding algorithm for these codes. This alphabet size can be reduced to a constant, while incurring a polylogarithmic loss in other parameters. We furthermore show how to decrease the locality to the $O(N^{1/3})$, albeit with a larger alphabet size and slightly lower distance. Our constructions build on the homological product paradigm previously studied by \cite{bravyi_homological_2014,zeng_minimal_2020}. As a byproduct, we also resolve (in the large alphabet regime) an open question of \cite{bravyi_homological_2014} by constructing homological product codes of linear dimension and distance with locality $N^\epsilon$ for arbitrarily small $\epsilon>0$ (though without transversal gates).

We remark that our codes permit magic state distillation with sub-polylogarithmic overhead. Until recently, it was a major open question to obtain such codes, which have since been constructed by the works \cite{wills_constant-overhead_2024,golowich_asymptotically_2025,nguyen_good_2025,golowich_quantum_2024,nguyen_quantum_2024}. Formally, given $m$ copies of the magic state $CCZ\ket{+}^{\otimes 3}$, our codes can be applied in the protocol described in \cite{bravyi_magic-state_2012} to distill $m/\log^\gamma(1/\epsilon)$ copies of the magic state with noise rate $\epsilon$, where the yield parameter $\gamma:=\log(N/K)/\log(D)$ approaches $0$ as $N\rightarrow\infty$. The smaller locality of our codes compared to those of \cite{wills_constant-overhead_2024,golowich_asymptotically_2025,nguyen_good_2025}, which otherwise achieve similar parameters, may translate to reduced space-time overhead for magic state distillation; we leave the problem of performing a precise theoretical as well as practical space-time comparison for future work (see in particular \cite{nguyen_quantum_2024}, where space-time overhead is of central importance).

To obtain our codes described above, we prove multiple algebraic results of independent interest, which constitute much of the technical work in this paper. Specifically, our efficent decoder mentioned above is based on a polynomial-time algorithm we develop for decoding \emph{dual} tensor products of Reed-Solomon codes. This latter algorithm, which is highly intricate, provides a new multivariate generalization of the well-studied Prony's method (dating back to 1795 \cite{de_prony_essai_1795}) for recovering a function given limited samples from its Fourier transform. Meanwhile, to bound the distance of our codes with locality $O(N^{1/3})$, we construct the first known codes possessing algebraic structure that exhibit a \emph{higher-order product-expansion} property, and have good dual codes\footnote{The known applications of product-expanding codes in quantum error correction crucially require the dual codes to have good distance \cite{panteleev_asymptotically_2022,leverrier_quantum_2022-1,dinur_good_2023,dinur_expansion_2024}.}. This property crucially underlies the recent construction of nearly good quantum locally testable codes \cite{dinur_expansion_2024}, and was previously only known for unstructured random codes \cite{kalachev_maximally_2025} (whose techniques we build on).

\subsection{Main Results}
\label{sec:mainres}
We now state our results on codes with transversal non-Clifford gates more formally. Recall that for a finite field $\bF_q$ of characteristic $p$, the $q$-ary $CCZ$ gate acts on three $q$-dits as\footnote{\label{footnote:trdef} Here $\tr_{\bF_q/\bF_p}:\bF_q\rightarrow\bF_p$ denotes the $\bF_p$-linear field trace map. Specifically, $\tr_{\bF_q/\bF_p}(\alpha)$ is defined to be the trace of the $\bF_p$-linear operator given by multiplication by $\alpha\in\bF_q$, when viewing $\bF_q$ as a $\bF_p$-linear vector space. Also, we will sometimes need the slightly more general $CCZ^a$ gate for $a\in\bF_q$; see Definition~\ref{def:crz}.}
\begin{equation*}
  CCZ\ket{z_1}\ket{z_2}\ket{z_3} = e^{2\pi i\tr_{\bF_q/\bF_p}(z_1z_2z_3)/p}\ket{z_1}\ket{z_2}\ket{z_3} \hspace{1em} \forall z_1,z_2,z_3\in\bF_q.
\end{equation*}
We say that a $[[N,K,D]]_q$ code supports a \emph{transversal $CCZ$ gate} if physical $CCZ$ gates can be applied on disjoint triples of physical qudits across three code states in order to induce logical $CCZ$ gates on $K$ disjoint triples of encoded logical qudits.

Our first main result provides the first known codes supporting transversal non-Clifford gates with linear dimension and distance, and sublinear locality:

\begin{theorem}[Informal statement of Theorem~\ref{thm:transRS}, Corollary~\ref{cor:subRSdec}, and Lemma~\ref{lem:alphred}]
  \label{thm:prod2inf}
  There exists an infinite family of $[[N,\Theta(N),\Theta(N)]]_q$ quantum subsystem codes of locality $O(\sqrt{N})$ that support transversal $CCZ$ with alphabet size $q=O(\sqrt{N})$. These codes have a polynomial-time decoding algorithm against $\Theta(N)$ adversarial errors.

  Furthermore, the alphabet size $q$ can be reduced to a constant (e.g.~$q=2$) while preserving the other parameters up to polylogarithmic factors, yielding\footnote{Recall the notation $\tilde{\Theta}(N):=\Theta(N\cdot(\log N)^{-O(1)})$.} $[[N,\tilde{\Theta}(N),\tilde{\Theta}(N)]]_2$ codes of locality $\tilde{O}(\sqrt{N})$ that support transversal $CCZ$.
\end{theorem}

While our codes in Theorem~\ref{thm:prod2inf} are subsystem codes, meaning that the weight-$O(\sqrt{N})$ parity-check measurements do not commute, we show in Appendix~\ref{sec:ec} that these codes nevertheless support a sort of ``single-shot'' error correction that works even in the presence of measurement errors (to the extent possible for a code of locality $O(\sqrt{N})$). Hence the subsystem nature of these codes does not pose any fundamental disadvantage compared to non-subsystem codes.

Furthermore, magic state distillation typically uses concatenation to ensure error-free syndrome measurements for the code with a transversal non-Clifford gate. In this setting, the single-shot error correction described above is therefore not necessary, but may lead to further improvements in efficiency. Specifically, by correcting some of the errors arising during the syndrome measurement, the single-shot property may allow for a more lightweight base code in the concatenation.

We construct our codes in Theorem~\ref{thm:prod2inf} by taking the product of a pair of quantum codes, which are in turn built from classical Reed-Solomon codes. Our efficient decoder in Theorem~\ref{thm:prod2inf} relies on a novel decoder for dual tensor products of such Reed-Solomon codes, which has implications of independent interest, as described in Section~\ref{sec:keydecoding} below.

Our second main result shows how to obtain codes with even smaller locality using higher-order products:

\begin{theorem}[Informal statement of Theorem~\ref{thm:tripleprodparam}]
  \label{thm:prod3inf}
  For every $\epsilon>0$, there exists an infinite family of $[[N,\Theta(N),\Theta(N^{1-\epsilon})]]_q$ quantum subsystem codes with locality $\Theta(N^{1/3})$ that support transversal $CCZ$ with alphabet size $q=2^{\poly(N)}$.
\end{theorem}

We construct our codes in Theorem~\ref{thm:prod3inf} by taking the product of three quantum codes, which we in turn build from randomly punctured tensor products of classical Reed-Solomon codes. As described in Section~\ref{sec:keypetRS} below, we leverage the combined algebraic structure and local testability of tensor Reed-Solomon codes to ensure good distance in Theorem~\ref{thm:prod3inf}, for which our proof involves novel results of independent interest.


\subsection{Additional Results on Product Constructions}
In this section, we provide more details on the product constructions we use, and we describe additional results that we prove for such products.

As mentioned above, the products we consider build upon the homological product construction of \cite{bravyi_homological_2014}, along with the closely related subsystem products studied by \cite{zeng_minimal_2020}. \cite{bravyi_homological_2014} constructed asymptotically good quantum codes with locality $O(\sqrt{N})$ from the product of two uniformly random CSS codes. Their distance analysis relied on this randomness in a seemingly ad-hoc manner, and they were unable to extend their analysis to higher-order products of more than two codes, which would yield smaller locality.

Our first key insight is that the notion of \emph{product-expansion}, introduced by \cite{panteleev_asymptotically_2022,kalachev_two-sided_2023}, provides a pseudorandomness property that is sufficient for good distance of homological and subsystem products. As such, we are able to derandomize the \cite{bravyi_homological_2014} construction, leading to several improvements:
\begin{enumerate}
\item We prove Theorem~\ref{thm:prod2inf} and Theorem~\ref{thm:prod3inf} above by taking products of codes with specific algebraic structure needed for the transversal $CCZ$ gates.
\item Whereas \cite{bravyi_homological_2014} only showed good distance and dimension for their codes of locality $O(\sqrt{N})$, we show that these codes are also locally testable, meaning that the weight of a corruption can be estimated by measuring a random low-weight check (see Definition~\ref{def:MDS} and Definition~\ref{def:CSScode}).
\item We resolve a conjecture of \cite{bravyi_homological_2014} (in the large alphabet regime) by extending the construction to higher-order products of locality $N^\epsilon$ for arbitrarily small $\epsilon>0$.
\end{enumerate}
The latter two results above are stated formally below; note the these codes are in fact ordinary (non-subsystem) CSS codes, but do not support transversal $CCZ$ gates.

\begin{theorem}[Informal statement of Corollary~\ref{cor:ssperandom} and Corollary~\ref{cor:sspemanyrandom}]
  \label{thm:homprodinf}
  There exists an infinite family of
  asymptotically good $[[N,\Theta(N),\Theta(N)]]_2$ quantum homological product codes of locality $O(\sqrt{N})$ that are locally testable with constant soundness.

  Furthermore, for every $\epsilon>0$, there exists an infinite family of $[[N,\Theta(N),\Theta(N)]]_q$ quantum homological product codes of locality $O(N^\epsilon)$ and alphabet size $q=2^{\poly(N)}$.
\end{theorem}



The code parameters in Theorem~\ref{thm:homprodinf} are surpassed by the recent breakthrough line of work constructing asymptotically good quantum LDPC codes \cite{breuckmann_balanced_2021,panteleev_asymptotically_2022,leverrier_quantum_2022-1,dinur_good_2023} and nearly good quantum locally testable codes \cite{dinur_expansion_2024,kalachev_maximally_2025} of constant locality. All of these constructions rely on a specific operation called a ``lifted'' or ``balanced'' product \cite{panteleev_quantum_2022,breuckmann_balanced_2021,panteleev_asymptotically_2022}, which can only be applied to highly symmetric codes possessing a large group symmetry. In contrast, our codes in Theorem~\ref{thm:homprodinf} use the more basic and widely applicable homological product construction, which is in essence an ordinary tensor product. To the best of our knownledge, Theorem~\ref{thm:homprodinf} provides the first known quantum codes of linear dimension and distance and of locality $\leq N^\epsilon$ for arbitrarily small $\epsilon>0$ that do not rely on lifted/balanced products.

\subsection{Roadmap}
The remainder of this paper is organized as follows. Section~\ref{sec:techover} provides a technical overview of our results, and concludes with some additional remarks and directions for future work in Section~\ref{sec:discussion}. We then provide necessary preliminary notions and results that will be used throughout this paper in Section~\ref{sec:prelim}. In Section~\ref{sec:peprod}, we present our distance bounds on quantum codes obtained via products of product-expanding codes. We present the core classical algorithm behind our decoder for Theorem~\ref{thm:prod2inf}, which is self-contained and has independent implications, in Section~\ref{sec:dualtensordec}. We then apply this algorithm to construct the quantum decoder in Section~\ref{sec:qdec}. We present our general paradigm for performing transversal non-Clifford gates on product codes in Section~\ref{sec:transversal}, and also apply this paradigm to construct the transversal gates in Theorem~\ref{thm:prod2inf}. Section~\ref{sec:perp} and Section~\ref{sec:tripleprod} provide our product-expansion result (underlying the distance bound) and transversal gate construction, respectively, for Theorem~\ref{thm:prod3inf}. Appendix~\ref{sec:peproofs} and Appendix~\ref{sec:omitdec} provide proofs that were omitted from the main body. Appendix~\ref{sec:ec} provides background on performing error correction at the level of quantum code states, and shows that our subsystem product codes support error correction from noisy syndromes.

\section{Technical Overview}
\label{sec:techover}
In this section, we provide a technical overview of key aspects of of our work. Specifically, in Section~\ref{sec:codeoverview}, we outline our code constructions in Theorem~\ref{thm:prod2inf} and Theorem~\ref{thm:prod3inf}, and describe the main ingredients for bounding the distance and constructing transversal $CCZ$ gates. We then provide more details on our decoding algorithm for Theorem~\ref{thm:prod2inf} in Section~\ref{sec:keydecoding} below, and on our distance proof for Theorem~\ref{thm:prod3inf} in Section~\ref{sec:keypetRS} below. We conclude the overview with some additional remarks and directions for future work in Section~\ref{sec:discussion}.

\subsection{Overview of Code Constructions}
\label{sec:codeoverview}
In this section, we provide an overview of our code constructions in Theorem~\ref{thm:prod2inf} and Theorem~\ref{thm:prod3inf}. Specifically, in Section~\ref{sec:codeinf} we present some basic coding theory background, and describe the product paradigm we consider. We then describe the main idea behind our distance bounds and transvesal $CCZ$ gates in Section~\ref{sec:subproddisinf} and Section~\ref{sec:subprodccz}, respectively.

In the overview below, we focus on the subsystem product paradigm used in Theorem~\ref{thm:prod2inf} and Theorem~\ref{thm:prod3inf}; see Section~\ref{sec:sscc} for background on the closely related homological product paradigm, which gives the (non-subsystem) codes in Theorem~\ref{thm:homprodinf}. Also see Appendix~\ref{sec:ec} for a description of how to perform error correction on our subsystem products in the presence of syndrome errors.

\subsubsection{Quantum Codes and Products}
\label{sec:codeinf}
In this section, we describe quantum subsystem codes, and the product paradigm of \cite{zeng_minimal_2020} for constructing them. More details can be found in Section~\ref{sec:prelim}.

We begin with some basic code definitions. Recall that a length-$n$ classical (linear) code of alphabet size $q$ is a subspace $C\subseteq\bF_q^n$, and the dual code is defined as $C^\perp=\{a\in\bF_q^n:a\cdot b=0\ \forall b\in C\}$. A quantum CSS code is a pair $Q=(Q_X,Q_Z)$ of length-$n$ classical codes satisfying the orthogonality condition $Q_X^\perp\subseteq Q_Z$, and has dimension $k=\dim(Q_Z)-\dim(Q_X^\perp)$.

In this paper, by a quantum subsystem CSS code (or simply ``subsystem code''), we mean a pair $Q=(Q_X,Q_Z)$ of length-$n$ classical codes, with no orthogonality condition. Such a subsystem code has dimension
\begin{equation*}
  k:=\dim(Q_Z+Q_X^\perp)-\dim(Q_X^\perp)=\dim(Q_Z)-\dim(Q_Z\cap Q_X^\perp),
\end{equation*}
and distance
\begin{equation*}
  d := \min_{c\in((Q_X+Q_Z^\perp)\setminus Q_Z^\perp)\cup((Q_Z+Q_X^\perp)\setminus Q_X^\perp)}|c|.
\end{equation*}

Given non-subsystem CSS codes $Q^i=(Q^i_X,Q^i_Z)$ of length $n$ and dimension $k_i$ for $i\in[t]$, the subsystem product code $Q=(Q_X,Q_Z)=\bigotimes_{i\in[t]}Q^i$ is defined by
\begin{equation*}
  Q_X = \bigotimes_{i\in[t]}Q^i_X, \hspace{1em} Q_Z = \bigotimes_{i\in[t]}Q^i_Z.
\end{equation*}
The product code $Q$ has length $N=n^t$, dimension $K=\prod_{i\in[t]}k_i$, and locality $\leq t\cdot n$. The main focus of our work is to construct such codes with (close to) linear distance, which also support transversal non-Clifford gates. We address these two objectives in Section~\ref{sec:subproddisinf} and Section~\ref{sec:subprodccz} below, respectively.

\subsubsection{Good Distance from Product-Expansion}
\label{sec:subproddisinf}
Our key approach to obtaining product codes with good distance is to leverage the product-expansion property introduced by \cite{panteleev_asymptotically_2022,kalachev_two-sided_2023}.

To define product-expansion, we need the following notation: for classical codes $(C_i\subseteq\bF_q^{n_i})_{i\in[t]}$, let
\begin{equation*}
  C^{(i)}=\bF_q^{n_1}\otimes\cdots\otimes\bF_q^{n_{i-1}}\otimes C_i\otimes\bF_q^{n_{i+1}}\otimes\cdots\otimes\bF_q^{n_t}.
\end{equation*}
Also, for $c\in C^{(i)}$, let $|c|_i\in[0,\prod_{j\neq i}n_j]$ denote the number of nonzero direction-$i$ columns in $c$.

\begin{definition}[\cite{kalachev_two-sided_2023}]
  \label{def:peinf}
  The \textbf{product-expansion $\rho$} of a collection $(C_i\subseteq\bF_q^{n_i})_{i\in[t]}$ of classical codes is the largest real number $\rho\geq 0$ such that for every $c\in C^{(1)}+\cdots+C^{(t)}$, there exists a decomposition
  \begin{equation*}
    c=c_1+\cdots+c_t
  \end{equation*}
  with each $c_i\in C^{(i)}$ such that
  \begin{equation*}
    |c| \geq \rho\sum_{i\in[t]}n_i|c_i|_i.
  \end{equation*}
\end{definition}

When $t=1$, the product-expansion of a single code simply equals its relative distance. Hence product-expansion is a sort of high-dimensional generalization of code distance.

Below, we state our main result bounding subsystem product code distance via product-expansion. For simplicity in exposition, we let $n=n_1=\cdots=n_t$.

\begin{theorem}[Informal statement of Theorem~\ref{thm:subpe}]
  \label{thm:subpeinf}
  Fix some $t\in\bN$. For each $i\in[t]$, let $Q^i=(Q^i_X,Q^i_Z)$ be a (non-subsystem) CSS code of length $n$ and dimension $\geq k$. For every $i\in[t]$, assume that the collections of codes $({Q^1_Z}^\perp,\dots,{Q^{i-1}_Z}^\perp,Q^i_X)$ and $({Q^1_X}^\perp,\dots,{Q^{i-1}_X}^\perp,Q^i_Z)$ have product-expansion at least some $\rho>0$. Then the subsystem product code $Q=\bigotimes_{i\in[t]}Q^i$ has parameters
  \begin{equation*}
    [[N=n^t,\;K\geq k^t,\;D \geq (\rho n)^t]].
  \end{equation*}
\end{theorem}

The distance bound in Theorem~\ref{thm:prod2inf} follows by applying Theorem~\ref{thm:subpeinf} with the product-expansion bound for Reed-Solomon codes shown by \cite{polishchuk_nearly-linear_1994}. The distance bound in Theorem~\ref{thm:prod3inf} follows by applying Theorem~\ref{thm:subpeinf} with our new product-expansion result for randomly punctured tensor Reed-Solomon codes, which we describe in more detail in Section~\ref{sec:keypetRS} below. The distance and testability bounds in Theorem~\ref{thm:homprodinf} also follow from similar results as in Theorem~\ref{thm:subpeinf}, but with the subsystem product replaced by the closely related (non-subsystem) homological product; see Section~\ref{sec:peprod} for details.

\subsubsection{Transversal $CCZ$ Gates from Algebraic Structure}
\label{sec:subprodccz}
We now describe how we perform transversal $CCZ$ gates on subsystem product codes. We say that a code supports a transversal $CCZ$ gate on $\ell$ logical qudits if the following holds: an appropriate application of physical $CCZ$ gates to disjoint triples of physical qudits across three codes states induces logical $CCZ$ gates on $\ell$ disjoint triples of logical qudits. The disjointness of the physical gates ensures our physical circuit has depth $1$, so an error on a given qudit can propagate to at most two other qudits. Meanwhile, the disjointness of the induced logical $CCZ$ gates prevents unwanted entanglement among the logical qudits, which for instance allows us to perform magic state distillation with these codes to obtain true unentangled magic $CCZ\ket{+}^{\otimes 3}$ states.

We prove the general sufficient condition for when subsystem product codes support transversal $CCZ$ stated in Theorem~\ref{thm:transgeninf} below. Here, for vectors $a,b\in\bF_q^n$, we let $a*b=(a_ib_i)_{i\in[n]}\in\bF_q^n$ denote the component-wise product. We extend this notation to vector spaces by letting $A*B=\spn\{a*b:a\in A,\;b\in B\}$, and similarly to powers, e.g.~$A^{*3}=A*A*A$.

\begin{theorem}[Informal statement of Theorem~\ref{thm:transgen}]
  \label{thm:transgeninf}
  Let $(Q^i=(Q^i_X,Q^i_Z))_{i\in[t]}$ be a collection of (non-subsystem) CSS codes, and let $Q=(Q_X,Q_Z)=\bigotimes_{i\in[t]}Q^i$ denote the subsystem product. For each $i\in[t]$, fix a subspace $L_i\subseteq Q_Z^i$ with $L_i\cap{Q_X^i}^\perp=\{0\}$, and let
  \begin{equation*}
    L := \bigotimes_{i\in[t]}L_i, \hspace{2em} S := Q_Z\cap Q_X^\perp 
  \end{equation*}
  satisfy
  \begin{equation}
    \label{eq:multpropinf}
    L^{*3} \cap (S*(L+S)^{*2}) = \{0\}.
  \end{equation}
  Then Q supports a transversal $CCZ$ gate on on $\dim(L)$ logical qudits.
\end{theorem}

While we have stated Theorem~\ref{thm:transgeninf} for the $CCZ$ gate, it generalizes directly to the $C^{r-1}Z$ gate for every integer $r\geq 2$ 
(see Definition~\ref{def:crz}), which in particular lies in higher levels of the Clifford hierarchy than $CCZ$ for $r\geq 4$.

At the level of quantum states, the transversal $CCZ$ gate in Theorem~\ref{thm:transgeninf} should be applied to a physical code state after gauge-fixing all of the extraneous degrees of freedom in the subsystem code to equal $\ket{0}$, which can for instance be performed by applying an appropriate Pauli correction after a round of $X$-error correction; see Section~\ref{sec:transgates} and Appendix~\ref{sec:ec} for details.

The condition in~(\ref{eq:multpropinf}) can be viewed as a sort of \emph{multiplication property} of the underlying codes $Q^i_X,Q^i_Z$, as it describes how codewords of appropriate products of these codes behave under pointwise multiplication. Intuitively, the space $L$ corresponds to logical operators of the product code $Q$, while $S$ corresponds to stabilizers, and component-wise multiplication corresponds to the transversal $CCZ$ gate. In these terms, the condition~(\ref{eq:multpropinf}) essentially says that the action of the transversal $CCZ$ gate on the logical operators of a given state (corresponding to $L^{*3}$) is not affected by applying stabilizers to the state (corresponding to terms in $S*(L+S)^{*2}$), and hence the transversal gate induces the desired logical $CCZ$ action on states that lie in the code $Q$.

Related multiplication properties were used in the recent works \cite{krishna_towards_2019,wills_constant-overhead_2024,golowich_asymptotically_2025,nguyen_good_2025,golowich_quantum_2024,nguyen_quantum_2024} in order to also construct codes capable of low-overhead magic state distillation. Yet as these prior works (with the exception of \cite{golowich_quantum_2024}) did not attempt to obtain sublinear locality, their multiplication properties were generally attainable in a simpler manner. For instance, \cite[Theorem~3.1]{golowich_asymptotically_2025} obtains quantum codes with transversal $CCZ$ gates from any classical code $C$ for which $C^{*3}$ has good distance.

Such multiplication properties have typically been achieved using codes with algebraic structure, such as Reed-Solomon or Reed-Muller codes. Indeed, to prove Theorem~\ref{thm:prod2inf} we show that~(\ref{eq:multpropinf}) is satisfied when $Q^i_X,Q^i_Z$ are appropriate Reed-Solomon codes, and to prove Theorem~\ref{thm:prod3inf} we show that~(\ref{eq:multpropinf}) is satisfied when $Q^i_X,Q^i_Z$ are appropriate (duals of) randomly punctured tensor Reed-Solomon codes.

However, in both of these applications (and especially for Theorem~\ref{thm:prod3inf}), showing that~(\ref{eq:multpropinf}) holds requires a rather delicate argument, as we are restricted to parameter regimes where the algebraic codes are known to be product-expanding. Specifically, pairs of Reed-Solomon codes are only product-expanding when the rates sum to $<1$ \cite{polishchuk_nearly-linear_1994,kalachev_two-sided_2023}. Meanwhile, while our randomly punctured tensor Reed-Solomon codes have no such rate restriction, it is more difficult to reason about component-wise products of their dual codes, which are also involved in the condition~(\ref{eq:multpropinf}). Nevertheless, we overcome these issues to prove Theorem~\ref{thm:prod2inf} and Theorem~\ref{thm:prod3inf} through careful choices of parameters and analysis of the resulting expressions in~(\ref{eq:multpropinf}); the details can be found in Section~\ref{sec:transRS} and Section~\ref{sec:tripleprod} respectively. The challenges described above also underly the difficulty in extending Theorem~\ref{thm:prod2inf} and Theorem~\ref{thm:prod3inf} to higher-order products, which could lead to codes with even smaller locality.

\subsection{Decoding Dual Tensor Products of Reed-Solomon Codes}
\label{sec:keydecoding}
In this section, we describe our key technical contribution behind the decoder in Theorem~\ref{thm:prod2inf}, which consists of an efficient decoding algorithm for dual tensor products of Reed-Solomon codes. We believe this algorithm to be of independent interest, as it can be viewed as a new generalization of the well-studied Prony's method for inferring a function from limited access to its Fourier transform (see e.g.~\cite{de_prony_essai_1795,candes_robust_2006,candes_super-resolution_2013,moitra_super-resolution_2015,price_robust_2015,kunis_multivariate_2016,chen_fourier-sparse_2016,sauer_pronys_2018,cuyt_multivariate_2018,diederichs_how_2023}).

\subsubsection{Problem Overview and Connections to Prior Work}
Recall that the codewords of a length-$n$ classical Reed-Solomon code are given by lists of evaluations of low-degree polynomials on $n$ distinct points over a finite field. Our decoding problem in Theorem~\ref{thm:prod2inf} ultimately boils down\footnote{The reduction from the quantum decoding to the classical dual tensor decoding can be found in Section~\ref{sec:qdec}.} to decoding the \emph{dual tensor product} of two Reed-Solomon codes, which is the length-$N=n^2$ code whose codewords are matrices consisting of sums of Reed-Solomon codes along rows and columns. This dual tensor code only has distance $O(n)=O(\sqrt{N})$, as it has codewords supported within a single row or column. Yet for our application to quantum decoding, we want to decode against errors $e\in\bF_q^{n\times n}$ of weight $\Theta(N)$. Of course, as the code distance is much lower, we cannot hope to uniquely determine such $e$. Instead, given the corrupted codeword $c\in\bF_q^{n\times n}$, we want to find some true codeword $c'$ whose distance from $c$ is at most $O(|e|)$. That is, we want to compute a constant-factor approximation for the minimum-distance decoding problem.

The problem of decoding such codes with low-weight codewords is a common challenge (often called ``degeneracy'') in quantum error correction. Yet our specific decoding problem described above also has a natural classical interpretation. That is, if we formulate our problem in the Fourier domain, it asks us to recover a sparse function (i.e.~the error) from the restriction of its Fourier transform to a contiguous rectangle of frequencies (i.e.~the parity-check syndrome). Here we use the fact that the parity-check matrix of a (dual tensor) Reed-Solomon code can be viewed as a (bivariate) Fourier transform over a finite field. As such, the side lengths of the rectangle of evaluation points equal the dimensions of the duals of our two initial Reed-Solomon codes.

The problem of recovering a function from restricted access to its Fourier transform dates back to Prony's work in 1795 \cite{de_prony_essai_1795}, which in fact can be viewed as a decoding algorithm for Reed-Solomon codes (hundreds of years before such such codes were implemented). While Prony considered univariate functions, our setting requires multivariate (specifically bivariate) functions, which have also been studied more recently (e.g.~\cite{kunis_multivariate_2016,sauer_pronys_2018,cuyt_multivariate_2018,diederichs_how_2023}). However, these works primarily focus on uniquely recovering a function from as few samples of the Fourier transform as possible, and allow for carefully crafted such sets of samples. In contrast, in our setting the samples must simply come from a contiguous rectangle, and as a result there are multiple valid decodings.

\subsubsection{Overview of Our Algorithm}
We now describe our solution to the problem described above, which ultimately requires a fairly involved algorithm, that draws on elements from the works above on Prony's method, as well as from the literature on decoding algebraic codes (see e.g.~\cite{guruswami_essential_2022}).

First, we may view the dual tensor product of two $k$-dimensional Reed-Solomon codes as the space of matrices in $\bF_q^{n\times n}$ given by evaluations of bivariate polynomials in $\bF_q[X_1,X_2]$ that have nonzero coefficients only for monomials of degree $<k$ in at least one of the two variables.

Assume we are a corrupted codeword $c=a+b\in\bF_q^{n\times n}$, where $a$ is a true dual tensor codeword, and $b$ is a corruption. Our goal is to identify a small set of points containing the support of $b$, which we can then erase from $c$ and perform erasure decoding via Gaussian elimination in polynomial time.

To begin, we compute a nonzero bivariate polynomial $e_0(X_1,X_2)$ of some small degree $s$ such that the product $e_0(X_1,X_2)c(X_1,X_2)$ (and therefore also $e_0(X_1,X_2)b(X_1,X_2)$) lies in the dual tensor product of two $(k+s)$-dimensional Reed-Solomon codes. We could hope that $e_0$ vanishes at every point in the corruption $\supp(b)$, in which case we would be done. However, there is another possibility: if the corruption $b$ agrees with a low-degree rational function (e.g.~$1/X_1$) within a given column (or row), and is $0$ elsewhere, then the rational function's denominator is valid choice of $e_0$ (e.g.~$e_0(X_1,X_2)=X_1$) that does not vanish within that column, and nevertheless ensures that $e_0(X_1,X_2)b(X_1,X_2)$ is a dual tensor codeword, as e.g.~$X_1\cdot 1/X_1=1$ is a low-degree polynomial.

To resolve this issue, we next compute a basis of \emph{all} possible low-degree polynomials $e(X_1,X_2)$ for which $e(X_1,X_2)c(X_1,X_2)$ agrees (at points in $\supp(e_0)$) with a dual tensor codeword of two $(k+s)$-dimensional Reed-Solomon codes. Within each row specificed by $X_1=x_1\in\bF_q$ (resp.~column specified by $X_2=x_2\in\bF_q$), we then compute the greatest common divisor (gcd) of the restriction $e(x_1,X_2)$ (resp.~$e(X_1,x_2)$) of every $e$ to that row (resp.~column). The key idea is that if the corruption $b$ agrees with a low-degree rational function such as $1/X_1$ within a given column (resp.~row), then all valid choices of $e$ should be divisible by the denominator of this rational function. Hence our gcd computation will identify these ``bad'' columns and rows where $b$ looks like a rational function. We can then hope to construct a superset of $\supp(b)$ by taking the union of these bad columns and rows with the set of vanishing points of $e_0$.

The procedure above contains the key idea behind our algorithm, though as described, it will may still fail to distinguish ``good'' columns and rows, inside which $b$ has small support, from ``bad'' ones. Our final algorithm addresses this issue by in fact only labeling some columns/rows as bad in a first pass, and then re-computing a basis of low-degree $e$ on the remaining good columns/rows, and identifying any additional bad columns/rows in this second pass.

The steps above together form the core subroutine of our decoding algorithm; for details, see Algorithm~\ref{alg:decoder} and the accompanying analysis in Lemma~\ref{lem:algtech}. However, another issue remains: this subroutine does not detect corruptions $b$ that agree with some polynomial of degree greater than $k$ but less than $k+s$ within a given column or row. As such, this subroutine only performs a weaker form of decoding in which the output may lie in a slightly larger code than the input. We address this issue with additional subroutines that find a nearby codeword of the original dual tensor code; see Algorithm~\ref{alg:decoderclose} and Algorithm~\ref{alg:decoderfinish}, which are analyzed in Lemma~\ref{lem:algdeg} and Lemma~\ref{lem:algfin} respectively.

The approach of computing a low-degree polynomial $e_0$ such that the product $e_0c$ lies within the desired code is commonly used to decode algebraic codes, such as in the classic Berlekamp-Welch decoder for Reed-Solomon codes (see e.g.~\cite{guruswami_essential_2022}). Yet the approach of computing a basis for all such low-degree polynomials $e$ seems more common in the literature on Prony's method (see e.g.~\cite{kunis_multivariate_2016}). However, \cite{kunis_multivariate_2016} simply uses this basis to compute its set of common vanishing points. Our method described above of computing gcd's along restrictions to subspaces is (to the best of our knowledge) novel.

In order to bound the number of bad columns and rows in the algorithm described above, our analysis makes frequent use the product-expansion property (see \cite{kalachev_two-sided_2023}) of Reed-Solomon codes shown by \cite{polishchuk_nearly-linear_1994}. Because this result only applies to Reed-Solomon codes whose rates sum to $<1$, our algorithm also only applies to such codes. \cite{kalachev_two-sided_2023} show that Reed-Solomon codes whose rates have sum $>1$ necessarily fail to be product-expanding. Hence it is an interesting question to determine if dual tensor products of such codes can nevertheless be efficiently decoded using alternative techniques.

\subsection{Product-Expansion of Randomly Punctured Tensor Reed-Solomon Codes}
\label{sec:keypetRS}
In this section, we describe our key technical contribution behind the distance proof of Theorem~\ref{thm:prod3inf}, for which we construct a family of algebraic codes exhibiting higher-order product-expansion. Specifically, as described in Section~\ref{sec:subproddisinf} and Section~\ref{sec:subprodccz}, to prove Theorem~\ref{thm:prod3inf}, we want a family of classical codes with sufficient algebraic structure to satisfy the multiplication property in~(\ref{eq:multpropinf}), such that the codes and their duals also satisfy the product-expansion conditions in Theorem~\ref{thm:subpeinf}.

\subsubsection{Result Statement}
Our main result constructing such codes is stated in Theorem~\ref{thm:peptRSinf} below. Below, we let $RS(q,k)$ denote the Reed-Solomon code over $\bF_q$ of length $q$ and dimension $k$, which consists of the evaluations of univariate polynomials of degree $<k$ on all $q$ elements of $\bF_q$. We construct our codes below by taking the $u$th tensor power of $RS(n,k)$ to obtain a length-$q^u$ code $RS(n,k)^{\otimes u}$, and then puncturing this code down to $n\ll q^u$ randomly chosen components, where $u$ is a large constant.

\begin{theorem}[Informal statement of Theorem~\ref{thm:peptRS}]
  \label{thm:peptRSinf}
  For every fixed $t\in\bN$ and $\epsilon>0$, there exists a sufficiently large $u\in\bN$ such that the following holds: For positive integers\footnote{For technical reasons, we in fact require $n$ to be of the form $n=m^u$ for $m\in\bN$; see Section~\ref{sec:perp} for details.} $n$ and $\epsilon n\leq k_1,\dots,k_t\leq(1-\epsilon)n$, let $\bF_q$ be a field of size $q\geq 2^{n^{t+1}}$, and sample $t$ uniformly random subsets $E_1,\dots,E_t\subseteq\bF_q^u$ each of size $|E_i|=n$. Then with high probability, the collection of $t$ length-$n$ codes
  \begin{equation}
    \label{eq:ptRSinf}
    \left(RS(q,k_1)^{\otimes u}|_{E_1},\dots,RS(q,k_{t-1})^{\otimes u}|_{E_{t-1}},\;(RS(q,k_t)^{\otimes u}|_{E_t})^\perp\right)
  \end{equation}
  has product-expansion at least $\Omega(1/n^\epsilon)$.
\end{theorem}

Because product-expansion is preserved when passing to subcollections (Lemma~\ref{lem:pesubtuple}), we may restrict attention to the first $t-1$ codes in~(\ref{eq:ptRSinf}) to conclude that collections of arbitrarily many randomly punctured tensor Reed-Solomon codes possess product-expansion. Yet the statement in Theorem~\ref{thm:peptRSinf} is strictly stronger, as it implies that a collection of $t-1$ such codes, along with one dual of such a code, still possess product-expansion. This condition of having $t-1$ codes of one type along with one code of the dual type exactly matches the condition needed in Theorem~\ref{thm:subpeinf} to ensure good quantum code distance. Hence we are able to apply Theorem~\ref{thm:peptRSinf} with Theorem~\ref{thm:subpeinf} to prove the distance bound in Theorem~\ref{thm:prod3inf}. We leave it as an open question to prove an analogue of Theorem~\ref{thm:peptRSinf} with more than one dual code in the collection; such a result could help further reduce the locality in Theorem~\ref{thm:prod3inf} (see the discussions in Section~\ref{sec:subprodccz} and Section~\ref{sec:discussion}).

Indeed, all known applications of product-expanding codes in quantum error correction crucially rely on product-expansion for the codes along with their duals \cite{panteleev_asymptotically_2022,leverrier_quantum_2022-1,dinur_good_2023,dinur_expansion_2024}, in order to ensure good distance in both the $X$ and $Z$ bases.
The only previously known collections of $t>2$ such product-expanding codes were uniformly random codes over large alphabets, as shown in the recent work of \cite{kalachev_maximally_2025}.
Such higher-order (with $t=4$) product-expansion was a crucial ingredient in the recent construction of nearly good quantum locally testable codes \cite{dinur_expansion_2024}. Theorem~\ref{thm:peptRSinf} may therefore mark progress towards constructing quantum locally testable codes with more algebraic structure, which has been proposed as a useful ingredient towards making such codes asymptotically good \cite{panteleev_maximally_2024} or support transversal gates \cite{golowich_quantum_2024}. Though note that as described in Section~\ref{sec:discussion}, in order to be directly applicable to the construction of \cite{dinur_expansion_2024}, Theorem~\ref{thm:peptRSinf} would need to be extended to allow all $t$ codes to be of the dual type.

\subsubsection{Proof Techniques}
Our proof of Theorem~\ref{thm:peptRSinf} builds upon the framework of \emph{maximally extendable (ME)} codes introduced by \cite{kalachev_maximally_2025}, which in turn is closely related to the classical literature on \emph{maximally recoverable (MR)} codes (see e.g.~\cite{gopalan_explicit_2014}). This framework studies codes that have ``maximal,'' or optimal, properties among all codes in a given family.

For instance, given a family of classical codes, a MR code for the family is a code that can correct every pattern of erasures that is correctable by any code in the family. MR codes are typically characterized by having an appropriate set of minors of the generator or parity-check matrix be nonzero. Such codes can then often be constructed by sampling random codes from the family over sufficiently large fields, where the desired minors vanish with high probability by the Schwartz-Zippel lemma (Lemma~\ref{lem:sz}).

\cite{kalachev_maximally_2025} showed how to apply related ideas to prove higher-order product-expansion for uniformly random codes over large fields. Specifically, \cite{kalachev_maximally_2025} first proved that an arbitrary number $t$ of classical locally testable codes have product-expansion. Such codes necessarily have dual codes of low distance, which as mentioned above, precludes the desired quantum applications. However, \cite{kalachev_maximally_2025} instead defined the notion of maximally extendable (ME) codes, which can be thought of as a version of MR for product-expansion. More concretely, \cite{kalachev_maximally_2025} characterized the product-expansion of codes $C_1,\dots,C_t$ as being equivalent (up to constant factors) to having certain minors be nonzero for the generator matrix of the dual tensor code $C^{(1)}+\cdots+C^{(t)}$ (defined in Section~\ref{sec:subproddisinf}). Then because there exist (locally testable) product-expanding codes, these minors will generically be nonzero, meaning that they will be nonzero for uniformly random codes over sufficiently large alphabets by Schwartz-Zippel.

To prove Theorem~\ref{thm:peptRSinf}, our key idea is to replace the family of all classical codes used by \cite{kalachev_maximally_2025} with a family that is \ref{it:algstruct} algebraically structured, and yet is nevertheless large enough to \ref{it:viderman} contain locally testable codes and \ref{it:randommdsinf} permit good dual code distance for generic elements. We show that the family of punctured tensor products of Reed-Solomon codes possesses all three desired properties:
\begin{enumerate}[label=(\roman*)]
\item\label{it:algstruct} Codewords of punctured tensor Reed-Solomon codes consist of evaluations of multivariate low-degree polynomials. Hence component-wise products of codewords lie within such tensor codes of higher degree (i.e.~dimension). This structure ultimately allows us to prove the desired multiplication property~(\ref{eq:multpropinf}) for the transversal $CCZ$ gate in Theorem~\ref{thm:prod3inf}.
\item\label{it:viderman} \cite{viderman_combination_2015} showed that (unpunctured) tensor products of arbitrary codes of good distance, such as Reed-Solomon codes, are locally testable.
\item\label{it:randommdsinf} We prove that randomly punctured tensor Reed-Solomon codes over large enough alphabets are in fact maximum distance separable (MDS; see Lemma~\ref{lem:randommds}), meaning that these codes and their duals both have optimal distance.
\end{enumerate}

Hence we conclude that any number $t$ of randomly punctured tensor Reed-Solomon codes over a large enough alphbet will with high probability have the necessary nonvanishing generator matrix minors to yield product-expansion.

We allow the last of the $t$ codes in Theorem~\ref{thm:peptRSinf} to be dual to our family through an additional trick. Specifically, the proof of \cite{kalachev_maximally_2025} for product-expansion of $t$ locally testable codes in fact still goes through when one of the $t$ codes is an arbitrary code of good distance. Hence we sample this final code first, deducing good distance from property~\ref{it:randommdsinf} above. We then treat this last code as fixed while sampling the other $t-1$ codes and proving product-expansion. The full argument can be found in Section~\ref{sec:perp}.



\subsection{Discussion}
\label{sec:discussion}
In this paper, we provide the first known quantum codes of linear dimension and distance with transversal $CCZ$ gates that permit error correction with sublinear weight measurements. Our constructions are based on products of codes with specific algebraic structure. En route, we prove multiple results of independent interest, pertaining to high-dimensional product codes, efficient decoding algorithms, and high-order product-expansion.

Our constructions possess some conceptual aspects that are perhaps surprising in the context of prior works. For instance, to the best of our knowledge, all previous constructions of quantum codes with low-weight checks and transversal $CCZ$ gates can be viewed through the lens of the cup product operation from algebraic topology \cite{bombin_exact_2007,bombin_topological_2007,bombin_gauge_2015,zhu_non-clifford_2023,scruby_quantum_2024,lin_transversal_2024,golowich_quantum_2024,zhu_topological_2025,breuckmann_cups_2024}. Yet our constructions in this paper appear to obtain transversal $CCZ$ gates without a cup product\footnote{The key property of a cup product $\cup$, which acts on a cochain complex, is the Leibniz rule. Over fields of characteristic $2$, this rule takes the form $\delta(\alpha\cup\beta)=\delta(\alpha)\cup\beta+\alpha\cup\delta(\beta)$, where $\delta$ denotes the coboundary map of the cochain complex, and $\alpha,\beta$ are cochains; see e.g.~\cite{lin_transversal_2024,breuckmann_cups_2024} for a more general discussion of cup products on quantum codes. We are not aware of any way to construct such a cup product on cochain complexes associated to our codes.}, perhaps suggesting that such gates require less rigid structure than previously understood.

We are able to apply the product-expansion property in more relaxed parameter regimes than prior works to prove good quantum code distance. Specifically, prior quantum code constructions simultaneously required product-expansion for a collection of codes $(C_1,\dots,C_t)$ and for the dual codes $(C_1^\perp,\dots,C_t^\perp)$ \cite{panteleev_asymptotically_2022,leverrier_quantum_2022-1,dinur_good_2023,dinur_expansion_2024}. In contrast, our product-expansion condition in Theorem~\ref{thm:subpeinf} is less stringent, as it only requires product-expansion for collections of codes and \emph{subcodes} of their duals.

Our constructions crucially rely on this looser requirement. For instance, \cite{kalachev_two-sided_2023} show that if a pair of Reed-Solomon codes is product-expanding, then the pair consisting of their duals cannot also be product-expanding; hence the prior works listed above cannot use Reed-Solomon codes. Yet our construction in Theorem~\ref{thm:prod2inf} nevertheless is based on product-expanding Reed-Solomon codes. Similarly, our construction in Theorem~\ref{thm:prod3inf} is based on the product-expanding codes we construct in Theorem~\ref{thm:peptRSinf}, which are also only shown to possess product-expansion for subcodes of the dual codes.

Our results suggest that further investigating ways to obtain good quantum codes from weaker product-expansion conditions may be a promising direction. Yet it is also an interesting question to obtain stronger product-expansion results for codes with algebraic structure. Indeed, as discussed in Section~\ref{sec:subprodccz} and Section~\ref{sec:keypetRS}, much of the difficulty in proving Theorem~\ref{thm:prod3inf} (and to a lesser extent, Theorem~\ref{thm:prod2inf}) arises because we need to choose parameter regimes in which our codes simultaneously exhibit the necessary algebraic structure~(\ref{eq:multpropinf}) and product-expansion (Theorem~\ref{thm:subpeinf}). Stronger product-expansion results in wider parameter regimes would make this task easier, which could for instance lead to extensions of Theorem~\ref{thm:prod3inf} to even higher-order products with lower locality.

\section{Preliminaries}
\label{sec:prelim}
This section introduces necessary notation and preliminary notions.

\subsection{Notation}
\label{sec:notation}
For a positive integer $n$, we let $[n]=\{1,\dots,n\}$. For a prime power $q$, we let $\bF_q$ denote the finite field of order $q$, and we let $\bF_q^*=\bF_q\setminus\{0\}$. For a set $I\subseteq[n]$, we let $\1_I\in\bF_q^n$ denote the indicator vector of set $I$, where the vector space $\bF_q^n$ will be made clear from context. For a vector $v\in\bF_q^n$, we let $|v|=|\{i\in[n]:v_i\neq 0\}|$ denote the Hamming weight of $v$, meaning the number of nonzero coordinates in the standard (or otherwise specified) basis. Unless explicitly stated otherwise, we will assume all $\bF_q$-vector spaces that we consider are finite-dimensional, so that all Hamming weights are finite.

Sometimes it will be helpful to consider more general Hamming weights, in which we partition the coordinates into groups and count the number of distinct groups containing a nonzero coordinate; the ordinary Hamming weight is recovered by choosing the finest partition. For instance, for a matrix $c\in\bF_q^{n_1\times n_2}$, we let $|c|_1$ (resp.~$|c|_2$) denote the number of nonzero columns (resp.~rows) of $c$. However, we will always specify such notation before using it.

For $t\in\bN$, $n_1,\dots,n_t\in\bN$, and $i\in[t]$, a direction-$i$ column in $[n_1]\times\cdots\times[n_t]$ is a subset of $n_i$ elements of the form $(j_1,\dots,j_t)$ such that $j_k$ is fixed for all $k\neq i$, and $j_i$ takes all possible values in $[n_i]$. Thus there are $\prod_{k\neq i}n_k$ distinct direction-$i$ columns in $[n_1]\times\cdots\times[n_t]$.

Elements of $\bF_q^{n_1}\otimes\cdots\otimes\bF_q^{n_t}=\bF_q^{n_1\times\cdots\times n_t}$ are sometimes referred to as $t$-dimensional tensors. For $i\in[t]$, a direction-$i$ column in such a tensor is a vector in $\bF_q^{n_i}$ obtained by restricting the tensor to coordinates within some direction-$i$ column in $[n_1]\times\cdots\times[n_t]$.

For vectors $v,v'\in\bF^n$, we let $v\cdot v'=\sum_{i\in[n]}v_iv_i'$ denote the standard bilinear form, and we let $v*v'=(v_iv_i')_{i\in[n]}\in\bF^n$ denote the component-wise product. We extend this latter definition to subspaces $V,V'\subseteq\bF^n$ by letting $V*V'=\spn\{v*v':v\in V,v'\in V'\}\subseteq\bF^n$.

For a set $S\subseteq\bF_q^n$, we let $\ket{S}=(1/\sqrt{|S|})\sum_{s\in S}\ket{s}\in(\bC^q)^{\otimes n}$ denote the quantum state given by the uniformly weighted superposition over elements in $S$.

\subsection{Classical Codes}
This section describes basic notions in classical coding theory.

\begin{definition}
  For a finite field $\bF_q$, a \textbf{classical linear code of length $n$ and dimension $k$ over $\bF_q$} is a $k$-dimensional linear subspace $C\subseteq\bF_q^n$. The \textbf{rate} of $C$ is $R:=k/n$. The \textbf{distance} $d$ of $C$ is the minimum Hamming weight of a nonzero element of $C$, that is $d=\min_{c\in C\setminus\{0\}}|c|$. The \textbf{relative distance} of $C$ is $\delta:=d/n$. We summarize the above parameters by saying that $C$ is an $[n,k,d]_q$ code. The \textbf{dual code} $C^\perp\subseteq\bF_q^n$ of $C$ is defined by $C^\perp=\{x\in\bF_q^n:x\cdot y=0\;\forall y\in C\}$, where $x\cdot y=\sum_{i\in[n]}x_iy_i$ denotes the standard bilinear form.
\end{definition}

All classical codes we consider in this paper will be linear, so we will often simply say ``classical code'' or even (when clear from context) ``code'' to refer to a classical linear code.


\begin{definition}
  \label{def:MDS}
  A \textbf{generator matrix} of a $k$-dimensional classical code $C\subseteq\bF_q^n$ is a matrix $G$ for which $C=\im(G^\top)$, and a \textbf{parity-check matrix} of $C$ is a matrix $H$ for which $C=\ker(H)$.

  The code $C$ is \textbf{maximum distance separable (MDS)} if it has a generator matrix $G\in\bF_q^{k\times n}$ for which every $k\times k$ submatrix has full rank, or equivalently, if it has a parity-check matrix $H\in\bF_q^{(n-k)\times n}$ for which every $(n-k)\times(n-k)$ submatrix has full rank.

  The code $C$ is \textbf{low-density parity-check (LDPC) of locality $w$} if there exists a parity-check matrix $H\in\bF_q^{m\times n}$ for $C$ such that every row and column of $H$ has at most $w$ nonzero entries.

  The code $C$ is a \textbf{locally testable code (LTC) of soundness $\rho$} if the parity-check matrix $H$ furthermore satisfies the following: for every $s\in\im H$, there exists some $e\in\bF_q^n$ with $He=s$ such that
  \begin{equation*}
    \frac{|s|}{m} \geq \rho\cdot\frac{|e|}{n}.
  \end{equation*}
\end{definition}

We will frequently use the following notions of tensor and dual tensor codes:

\begin{definition}
  \label{def:classtensor}
  For classical codes $C_1\subseteq\bF_q^{n_1}$ and $C_2\subseteq\bF_q^{n_2}$, the \textbf{tensor code} $C_1\otimes C_2\subseteq\bF_q^{n_1}\otimes\bF_q^{n_2}=\bF_q^{n_1\times n_2}$ consists of all $n_1\times n_2$ matrices for which every column lies in $C_1$ and every row lies in $C_2$. For $c_1\in C_1,c_2\in C_2$, we let $c_1\otimes c_2\in C_1\otimes C_2$ be the tensor codeword given by $(c_1\otimes c_2)_{i,j}=(c_1)_i(c_2)_j$. We furthermore define the \textbf{dual tensor code} $C_1\boxplus C_2\subseteq\bF_q^{n_1}\otimes\bF_q^{n_2}$ by
  \begin{equation*}
    C_1\boxplus C_2 = (C_1^\perp\otimes C_2^\perp)^\perp = C_1\otimes\bF_q^{n_2}+\bF_q^{n_1}\otimes C_2.
  \end{equation*}
\end{definition}

Note that ``dual tensor'' is a slight abuse of terminology that we use for conciseness, as the phrase ``dual tensor of duals'' would be more accurate.

We will use various codes arising from evaluations of polynomials over finite fields, defined below.

\begin{definition}
  \label{def:evlnotation}
  For a finite field $\bF_q$, an integer $t\in\bN$, and a subset $S\subseteq\bF_q^t$, we define a linear map
  \begin{equation*}
    \evl_S:\bF_q[X_1,\dots,X_t]\rightarrow\bF_q^S
  \end{equation*}
  by
  \begin{equation*}
    \evl_S(f) = (f(x))_{x\in S}.
  \end{equation*}
  When $S=\bF_q^t$, we may omit $S$ from the notation and write $\evl=\evl_S$. For a bounded subset $A\subseteq\bR^t$, we also define
  \begin{equation*}
    \bF_q[X_1,\dots,X_t]^A = \left\{\sum_{a\in A\cap\bZ^t}f_aX^a:f_a\in\bF_q\;\forall a\in A\cap\bZ^t\right\},
  \end{equation*}
  where $X^a=X_1^{a_1}\cdots X_t^{a_t}$, to be the space of $t$-variate polynomials that have nonzero coefficients only for those monomials whose degree lies in $A$. 

  We will perform addition and scalar multiplication on subsets of $\bR^t$ in a pointwise manner, so that for instance $rA=\{ra:a\in A\}$ and $A+A'=\{a+a':a\in A,a'\in A'\}$.
\end{definition}

The notation in Definition~\ref{def:evlnotation} behaves well under addition, pointwise multiplication, and tensor products of codes, as demonstrated by the following lemma, whose proof is immediate from the definitions.

\begin{lemma}
  For every $t\in\bN$ and every $A,A'\subseteq\bR^t$, it holds that
  \begin{align*}
    \evl(\bF_q[X_1,\dots,X_t]^A)+\evl(\bF_q[X_1,\dots,X_t]^{A'}) &= \evl(\bF_q[X_1,\dots,X_t]^{A\cup A'}) \\
    \evl(\bF_q[X_1,\dots,X_t]^A)*\evl(\bF_q[X_1,\dots,X_t]^{A'}) &\subseteq \evl(\bF_q[X_1,\dots,X_t]^{A+A'}) \\
    \evl(\bF_q[X_1,\dots,X_t]^A)\otimes\evl(\bF_q[X_1,\dots,X_t]^{A'}) &= \evl(\bF_q[X_1,\dots,X_{2t}]^{A\times A'}).
  \end{align*}
\end{lemma}

Reed-Solomon codes provide a particularly useful example of polynomial evaluation codes:

\begin{definition}
  \label{def:RS}
  Given a prime power $q$, a positive integer $k$, and a subset $S\subseteq\bF_q$, the \textbf{Reed-Solomon code} $\textrm{RS}_S(q,k)=\evl_S(\bF_q[X]^{[0,k)})$ is the classical code of length $q$ and dimension $k$ over $\bF_q$ consisting of the evaluations of all univariate polynomials of degree $<k$ on points in $S$. When $S=\bF_q$ we let $\textrm{RS}(q,k)=\textrm{RS}_S(q,k)$.
\end{definition}

The following lemma is well-known:

\begin{lemma}
  \label{lem:RSdual}
For all integers $k$, $0 \le k \le q$,  $\textrm{RS}(q,k)^\perp = \textrm{RS}(q,q-k)$.
\end{lemma}


\subsection{Quantum Codes}
This section describes basic notions in quantum coding theory. In this paper, we restrict attention to quantum CSS codes, described below.

\begin{definition}
  \label{def:CSScode}
  For a finite field $\bF_q$, a \textbf{quantum CSS code of length $n$ over $\bF_q$} is a pair $Q=(Q_X,Q_Z)$ of subspaces (i.e.~classical codes) $Q_X,Q_Z\subseteq\bF_q^n$ such that $Q_X^\perp\subseteq Q_Z$. The \textbf{dimension} of $Q$ is $k:=\dim(Q_Z)-\dim(Q_X^\perp)$, and the \textbf{rate} is $R:=k/n$. The \textbf{distance} of $Q$ is
  \begin{equation*}
    d := \min_{c\in(Q_X\setminus Q_Z^\perp)\cup(Q_Z\setminus Q_X^\perp)}|c|.
  \end{equation*}
  We summarize the above parameters by saying that $Q$ is an $[[n,k,d]]_q$ code.

  The quantum code $Q$ is \textbf{low-density parity-check (LDPC) of locality $w$} if $Q_X,Q_Z$ are classical LDPC codes of locality $w$. Similarly, $Q$ is a \textbf{locally testable code (LTC) of soundness $\rho$} if $Q_X,Q_Z$ are classical LTCs of soundness $\rho$.
  

\end{definition}

A quantum CSS code $Q=(Q_X,Q_Z)$ of length $n$ and dimension $k$ formally encodes a $k$-qudit message $\ket{\phi}\in(\bC^q)^{\otimes k}$ into an $n$-qudit code state $\ket{\psi}\in\cQ$, where
\begin{equation}
  \label{eq:qcodespace}
  \cQ = \spn\left\{\ket{z+Q_X^\perp}:z\in Q_Z\right\} \subseteq (\bC^q)^{\otimes n}.
\end{equation}
Here we recall from Section~\ref{sec:notation} that $\ket{z+Q_X^\perp}$ denotes the uniform superposition over all elements of the coset $z+Q_X^\perp$. Note that any isometry from $(\bC^q)^{\otimes k}$ to $\cQ$ forms a valid encoding map; we will need specific choices of such maps when constructing transversal gates on our codes.

In this paper, we often simply say ``quantum code'' to refer to a quantum CSS code. We are interested in constructing families of quantum LDPC and locally testable codes with locality $w$ significantly less than the block length $n$. Typically, the study of qLDPC codes focuses on constructing families of codes with constant locality $w$ as the block length $n$ grows. However, we remark that \cite{hastings_quantum_2023,wills_tradeoff_2024} gives a method for reducing the locality $w$ to a constant at the cost of decreasing the code's distance and soundness. Thus quantum LDPC and locally testable codes of constant locality can be obtained from codes of sufficiently low (but growing) locality.

\begin{definition}
  An infinite family of classical or quantum codes is \textbf{asymptotically good} if the codes' dimension $k=\Theta(n)$ and distance $d=\Theta(n)$ both grow linearly as the block length $n\rightarrow\infty$, and the alphabet size is a constant $q=O(1)$. The family is \textbf{LDPC} if the locality $w=O(1)$ remains bounded by a constant.
\end{definition}

We will use the following quantum analogue of Definition~\ref{def:RS}:

\begin{definition}
  \label{def:QRS}
  A \textbf{quantum Reed-Solomon code} is a quantum CSS code $Q=(Q_X,Q_Z)$ such that $Q_X,Q_X^\perp,Q_Z,Q_Z^\perp$ are all classical Reed-Solomon codes with the same set of evaluation points, in the sense of Definition~\ref{def:RS}.
\end{definition}

By Lemma~\ref{lem:RSdual}, Reed-Solomon codes with the entire field as the evaluation set have duals that are also Reed-Solomon codes on the entire field, and hence such codes yield quantum Reed-Solomon codes. Throughout this paper, it will be sufficient to restrict attention to quantum Reed-Solomon codes whose evaluation set is the entire field.

\subsection{Single-Sector Chain Complexes and Homological Products}
\label{sec:sscc}

In this paper, we will often use the langugae of chain complexes to describe the quantum CSS codes we consider. We follow \cite{bravyi_homological_2014} in restricting attention to ``single-sector'' chain complexes, defined below, which consist of just a single vector space and boundary map.


\begin{definition}
  \label{def:sscc}
  A \textbf{single-sector chain complex $\cC_*=(C,\partial^{\cC})$ over $\bF_q$} consists of a $\bF_q$-vector space $C$ and a linear \textbf{boundary map} $\partial^{\cC}:C\rightarrow C$ satisfying $(\partial^{\cC})^2=0$. When clear from context, we omit the superscript and write $\partial=\partial^{\cC}$.

  Assuming that $C$ has a fixed basis, then the \textbf{locality} $w^{\cC}$ of $\cC$ is the maximum number of nonzero entries in any row or column of $\partial$ in this fixed basis. We furthermore define the following (standard) vector spaces:
  \begin{align}
    \label{eq:cbhdefs}
    \begin{split}
      \text{the space of \bf cycles } Z_*(\cC) &:= \ker(\partial) \subseteq C \\
      \text{the space of \bf boundaries } B_*(\cC) &:= \im(\partial) \subseteq C \\
      \text{the \bf homology } H_*(\cC) &:= Z_*(\cC)/B_*(\cC)
    \end{split}
  \end{align}
  
  The \textbf{single-sector cochain complex $\cC^*$} associated to $\cC_*$ is the chain complex with vector space $C$ and boundary map given by the \textbf{coboundary map} $\delta=\partial^\top$. We analogously define the \textbf{cocycles $Z^*(\cC)=\ker(\delta)$, coboundaries $B^*(\cC)=\im(\delta)$, and cohomology $H^*(\cC)=Z^*(\cC)/B^*(\cC)$}.
\end{definition}

In this paper, we restrict attention to single-sector chain complexes $\cC$ in which $C$ is a finite-dimensional $\bF_q$-vector space, so the following standard lemma holds.

\begin{lemma}[Well known]
  \label{lem:hombasis}
  Let $\cC$ be a single-sector chain complex over $\bF_q$ with all $C$ finite-dimensional. Then
  \begin{equation*}
    \dim H_*(\cC)=\dim H^*(\cC).
  \end{equation*}
  Furthermore, the bilinear form $\langle\cdot,\cdot\rangle:H^*(\cC)\times H_*(\cC)\rightarrow\bF_q$ given by
  \begin{equation*}
    \langle c'+B^*(\cC),c+B_*(\cC)\rangle:=c'\cdot c
  \end{equation*}
  is a well-defined nondegenerate form, independent of the choice of coset representatives $c',c$. In particular, letting $k=\dim H_*(\cC)$, then there exists a basis $c^1+B^*(\cC),\dots,c^{k}+B^*(\cC)$ for $H^*(\cC)$ and a basis $c_1+B_*(\cC),\dots,c_{k}+B_*(\cC)$ for $H_*(\cC)$ such that
  \begin{equation}
    \label{eq:dualbases}
    \langle c^j+B^*(\cC),c_\ell+B_*(\cC)\rangle=c^j\cdot c_\ell=\1_{j=\ell} \hspace{1em} \forall \; j,\ell\in[k].
  \end{equation}
\end{lemma}

Lemma~\ref{lem:hombasis} can be proven by basic linear-algebraic manipulations. For readers familiar with the language of quantum stabilizer codes, the lemma says that the space of logical operators of a CSS code can be decomposed into anticommuting $X$ and $Z$ operators; a proof can be found in Section~10.5.7 of~\cite{nielsen_quantum_2010}.

\begin{definition}
  Let $\cC$ be a single-sector chain complex with $k=\dim H_*(\cC)$. We say that vectors $c^1,\dots,c^{k}\in Z^*(\cC)$ and $c_1,\dots,c_{k}\in Z_*(\cC)$ form \textbf{dual bases for cohomology and homology} if they satisfy~(\ref{eq:dualbases}).
\end{definition}

Recall that by definition $\im(\partial)=\ker(\partial^\top)^\perp$ and that $\delta=\partial^\top$. Therefore the chain complex condition $\partial^2=0$ can be equivalently stated as saying that $\ker(\delta)^\perp\subseteq\ker(\partial)$, which is precisely the CSS condition in Definition~\ref{def:CSScode}. This observation motivates the following definition:

\begin{definition}
  \label{def:cctoqcode}
  For a single-sector chain complex $\cC$, the \textbf{quantum code associated to $\cC$} is the CSS code $Q=(Q_X,Q_Z)$ with $Q_X=\ker\delta$ and $Q_Z=\ker\partial$.
\end{definition}

Definition~\ref{def:cctoqcode} associates cohomology with $Q_X/Q_Z^\perp$ and associates homology with $Q_Z/Q_X^\perp$. We maintain this convention throughout the paper, though note that in some other works these associations are swapped.

The following lemma is immediate from the above definitions:

\begin{lemma}
  Let $Q$ be the quantum code associated to a single-sector chain complex $\cC$. Then the dimension of $Q$ is $k=\dim H_*(\cC)=\dim H^*(\cC)$, and the locality of $Q$ is at most the locality of $\cC$.
\end{lemma}

We may similarly express distance and local testability in terms of chain complexes:

\begin{definition}
  \label{def:sysdisexp}
  For a single-sector chain complex $\cC$, the \textbf{systolic distance $d_*(\cC)$} and the \textbf{cosystolic distance $d^*(\cC)$} are defined as
  \begin{equation*}
    d_*(\cC) = \min_{c\in Z_*(\cC)\setminus B_*(\cC)}|c|, \hspace{2em} d^*(\cC) = \min_{c\in Z^*(\cC)\setminus B^*(\cC)}|c|.
  \end{equation*}
  The \textbf{filling constant $\mu_*(\cC)$} and the \textbf{cofilling constant $\mu^*(\cC)$} are defined as
  \begin{equation*}
    \mu_*(\cC) = \max_{b\in B_*(\cC)}\min_{c\in C:\partial(c)=b}\frac{|c|}{|b|}, \hspace{2em} \mu^*(\cC) = \max_{b\in B^*(\cC)}\min_{c\in C:\delta(c)=b}\frac{|c|}{|b|}.
  \end{equation*}
\end{definition}

If a chain complex has small filling and cofilling constants, then high-weight errors on the associated quantum code must have large syndromes. Equivalently, low-weight syndromes can only arise from low-weight errors. This phenomenon is equivalent to local testability, as described below.

By definition, the distance $d(Q)$ of the quantum code $Q$ associated to $\cC$ is equal to
\begin{equation*}
  d(Q) = \min\{d_i(\cC),d^i(\cC)\}.
\end{equation*}
Similarly, the soundness $\rho(Q)$ (i.e.~local testability, see Definition~\ref{def:CSScode}) of $Q$ is equal to
\begin{equation*}
  \rho(Q) = \min\left\{\frac{1}{\mu_*(\cC)},\;\frac{1}{\mu^*(\cC)}\right\}.
\end{equation*}
That is, the (co)filling constants describe the local testability of $Q$. The reciprocols of the (co)filling constants are sometimes called \textbf{(co)cycle expansion constants}. 


The above definitions motivate studying ways to construct larger chain complexes from smaller ones, as a way to construct larger quantum codes from smaller ones. The following definition describes one of the most natural such operations on chain complexes, namely a product. We carry over this single-sector version of a homological product from \cite{bravyi_homological_2014}.
We will restrict attention to fields of characteristic $2$ when considering single-sector products in order to improve simplicity in the presentation by avoiding signing issues.

\begin{definition}
  \label{def:sshomprod}
  For single-sector chain complexes $\cA$ and $\cB$ over a field $\bF_q$ of characteristic $2$, the \textbf{(single-sector) homological product} $\cC=\cA\otimes\cB$ is the single-sector chain complex over $\bF_q$ given by the vector space
  \begin{equation*}
    C := A \otimes B
  \end{equation*}
  and the boundary map
  \begin{equation*}
    \partial^{\cC} := \partial^{\cA}\otimes I+I\otimes\partial^{\cB}.
  \end{equation*}
\end{definition}

Given $t$ single-sector complexes $(\cC_i=(C_i,\partial_i))_{i\in[t]}$, we may take $t-1$ products as defined in Definition~\ref{def:sshomprod} to obtain a single-sector complex $\cA=\cC_1\otimes\cdots\otimes\cC_t$ given by the vector space
\begin{equation*}
  A = C_1\otimes\cdots\otimes C_t
\end{equation*}
with boundary map
\begin{equation*}
  \partial^{\cA} = \partial_1\otimes I^{\otimes t-1} + I\otimes\partial_2\otimes I^{\otimes t-2} + \cdots + I^{\otimes t-1}\otimes\partial_t.
\end{equation*}
Such higher-order products are considered in Theorem~\ref{thm:sspe} in Section~\ref{sec:peprod}. Note that if $w_i$ denotes the locality of $\cC_i$, then by definition $\cA$ has locality $\leq w_1+\cdots+w_t$.

The homological product in Definition~\ref{def:sshomprod}, along with its subsystem variant (see Section~\ref{sec:subsystem} below), is the principal operation studied in this paper. The following result shows how the homology groups behave under such products.


\begin{proposition}[K\"{u}nneth formula for single-sector complexes; see \cite{bravyi_homological_2014}]
  \label{prop:sskunneth}
  Let $\cA$ and $\cB$ be single-sector chain complexes over a field of characteristic $2$. Then
  \begin{equation*}
    H_*(\cA\otimes\cB) \cong H_*(\cA)\otimes H_*(\cB).
  \end{equation*}
  Furthermore, for $a\in Z_*(\cA)$ and $b\in Z_*(\cB)$, the isomorphism above maps
  \begin{equation*}
    a\otimes b+B_*(\cA\otimes\cB) \mapsfrom (a+B_*(\cA))\otimes(b+B_*(\cB)).
  \end{equation*}
\end{proposition}

Below we present some useful corollaries of the K\"{u}nneth formula.

\begin{corollary}
  \label{cor:ssprodbases}
  If $\{a^1,\dots,a^{k^{\cA}}\},\{a_1,\dots,a_{k^{\cA}}\}$ and $\{b^1,\dots,b^{k^{\cB}}\},\{b_1,\dots,b_{k^{\cB}}\}$ are dual cohomology/homology bases for single-sector chain complexes $\cA$ and $\cB$ respectively, then
  \begin{equation*}
    \{a^i\otimes b^j:i\in[k^{\cA}],j\in[[k^{\cB}]\},\{a_i\otimes b_j:i\in[k^{\cA}],j\in[[k^{\cB}]\}
  \end{equation*}
  are dual cohomology/homology bases for $\cA\otimes\cB$.
\end{corollary}

\begin{corollary}
  \label{cor:ssprodcycles}
  For single-sector chain complexes $\cA_1,\dots,\cA_t$,
  \begin{equation*}
    Z_*(\cA_1\otimes\cdots\otimes\cA_t) = Z_*(\cA_1)\otimes\cdots\otimes Z_*(\cA_t) + B_*(\cA_1\otimes\cdots\otimes\cA_t).
  \end{equation*}
\end{corollary}


Thus the dimension of quantum codes behaves nicely under homological products. The goal of our work is to understand how the distance, local testability, and transversal gates of quantum codes behave under such products.


\subsection{Subsystem Codes and Products}
\label{sec:subsystem}
Sometimes it will be helpful to consider quantum codes in which we only encode logical qudits into certain degrees of freedom in the code. Such codes are called subsystem codes. Below we present a definition using the CSS framework that will be sufficient for our purposes.

\begin{definition}
  \label{def:subsystem}
  A \textbf{quantum subsystem CSS code of length $n$ over $\bF_q$} is a pair $Q=(Q_X,Q_Z)$ of subspaces (i.e.~classical codes) $Q_X,Q_Z\subseteq\bF_q^n$. The \textbf{dimension} of $Q$ is
  \begin{equation*}
    k:=\dim(Q_Z+Q_X^\perp)-\dim(Q_X^\perp)=\dim(Q_Z)-\dim(Q_Z\cap Q_X^\perp),
  \end{equation*}
  and the \textbf{distance} is
  \begin{equation*}
    d := \min_{c\in((Q_X+Q_Z^\perp)\setminus Q_Z^\perp)\cup((Q_Z+Q_X^\perp)\setminus Q_X^\perp)}|c|.
  \end{equation*}
  We summarize the above parameters by saying that $Q$ is an $[[n,k,d]]_q$ subsystem code.
  
  The subsystem code $Q$ has \textbf{locality} $w$ if there exist parity-check matrices $H_X\in\bF_q^{m_X\times n},\;H_Z\in\bF_q^{m_Z\times n}$ for $Q_X,Q_Z$ respectively, meaning that $Q_X=\ker H_X$ and $Q_Z=\ker H_Z$, such that every row and column of $H_X$ and $H_Z$ has at most $w$ nonzero entries.
\end{definition}

Definition~\ref{def:subsystem} essentially presents subsystem codes $Q=(Q_X,Q_Z)$ as CSS codes where we remove the orthogonality condition $Q_X^\perp\subseteq Q_Z$. To recover the orthogonality condition, we may define a CSS code $Q'=(Q_X',Q_Z')$ with $Q_X'=Q_X+Q_Z^\perp$ and $Q_Z'=Q_Z+Q_X^\perp$. Codewords (or more precisely in the language of stabilizer codes, logical $Z$ operators) of the CSS code $Q'$ correspond to cosets in $Q_Z'/Q_X'^\perp$. We may then view $Q$ as a ``$k$-dimensional subsystem code of $Q'$'', where we only encode the length-$k$ message into the subspace, or ``subsystem,'' of codewords in $Q_Z/{Q_X'}^\perp\subseteq Q_Z'/Q_X'^\perp$. Therefore subsystem codes can be viewed as CSS codes where only a subspace of the entire code space is used to encode the message. Note that we still good distance for all codewords in $Q_X'\setminus Q_Z^\perp$ and $Q_Z'\setminus Q_X^\perp$, so that the information in our length-$k$ message is protected regardless of the values of the additional (unprotected) degrees of freedom. Also note that we allow $Q_Z^\perp\setminus(Q_X\cap Q_Z^\perp)$ and $Q_X^\perp\setminus(Q_Z\cap Q_X^\perp)$ to contain low-weight codewords, which corresond to low-weight errors that only corrupt the unprotected degrees of freedom in our subsystem code.

For readers familiar with stabilizer codes, the spaces $Q_Z^\perp$ and $Q_X^\perp$ correspond to gauge operators, the spaces $Q_X\cap Q_Z^\perp$ and $Q_Z\cap Q_X^\perp$ correspond to stabilizers, and the spaces $(Q_X+Q_Z^\perp)/Q_Z^\perp\cong Q_X/(Q_X\cap Q_Z^\perp)$ and $(Q_Z+Q_X^\perp)/Q_X^\perp\cong Q_Z/(Q_Z\cap Q_X^\perp)$ correspond to logical operators.

At the level of quantum states, if the subsystem code $Q$ has length $n$ and dimension $k$, then it encodes message states $\ket{\phi}\in(\bC^q)^{\otimes k}$ into code states $\ket{\psi}\in\cQ'$, where
\begin{equation}
  \label{eq:qsubcodespace}
  \cQ' = \spn\left\{\ket{z+(Q_Z\cap Q_X^\perp)}:z\in Q_Z+Q_X^\perp\right\} \subseteq (\bC^q)^{\otimes n}
\end{equation}
is the code space (as defined in~(\ref{eq:qcodespace})) associated to the (non-subsystem) CSS code $Q'=(Q_X',Q_Z')$ defined above. Note that the dimension $k'$ of $Q'$ can be larger than the dimension $k$ of $Q$, and hence each message state $\ket{\phi}$ has multiple valid encodings $\ket{\psi}$. In other words, $Q$ actually encodes a total of $k'$ qudits, but $k'-k$ of these qudits are not used to encode the true message, and can take on arbitrary values; these extra $k'-k$ qudits are often called \textit{gauge qudits}.

There are multiple reasons to introduce such gauge qudits. For instance, if $Q$ is such that there are low-weight errors that can corrupt these gauge qudits, then by separating them from those qudits used to encode the message, the subsystem code $Q$ can have greater distance than the CSS code $Q'$. Furthermore, the locality of $Q$ can also be better (i.e.~lower) than the locality of $Q'$, if $Q_X,Q_Z$ permit sparser parity-check matrices than $Q_X',Q_Z'$. At the level of quantum states, this phenomenon corresponds to the code $Q'$ having high-weight stabilizers, but permitting lower-weight checks (which together generate the stabilizers) at the cost of corrupting the gauge qudits. In this paper, we primarily leverage this second advantage of subsystem codes, in order to construct codes with transversal non-Clifford gates while still permitting low-weight checks.

Some subsystem codes do have the disadvantage of not being decodable in the presence of measurement errors, as the measurements corresponding to their parity-checks in general do not commute and, and hence repeated measurements will not give the same result (whereas for non-subsystem codes, the measurements do commute and hence can be repeated). However, in Appendix~\ref{sec:ec}, we show that subsystem product codes we consider (see Definition~\ref{def:subhomprod} below) nevertheless are decodable in the presence of measurement errors.


We now present a sort of subsystem product, which takes as input a set of non-subsystem CSS codes, and outputs a subsystem product code. Such subsystem products were previously studied by \cite{zeng_minimal_2020}.

\begin{definition}
  \label{def:subhomprod}
  For some $t\in\bN$, let $(Q^i=(Q^i_X,Q^i_Z))_{i\in[t]}$ be (non-subsystem) CSS codes over some field $\bF_q$. We define the \textbf{subsystem product} $Q=(Q_X,Q_Z)=\bigotimes_{i\in[t]}Q^i$ by
  \begin{align*}
    Q_X &= \bigotimes_{i\in[t]}Q^i_X \\
    Q_Z &= \bigotimes_{i\in[t]}Q^i_Z.
  \end{align*}
\end{definition}

In a slight abuse of notation, we use the symbol $\otimes$ to denote both homological products of single-sector chain complexes (Definition~\ref{def:sshomprod}) and subsystem products of CSS codes (Definition~\ref{def:subhomprod}); the type of the factors (single-sector complexes vs.~CSS codes) disambiguates the type of product.

Similarly as for single-sector homological products, locality transforms additively under subsystem products. That is, if $Q^1,\dots,Q^t$ are CSS codes of respective localities $w_1,\dots,w_t$, then by definition the subsystem product $\bigotimes_{i\in[t]}Q^i$ has locality $\leq w_1+\cdots+w_t$.

An analogue of the K\"{u}nneth formula also holds for subsystem products:

\begin{proposition}
  For some $t\in\bN$, let $(Q^i=(Q^i_X,Q^i_Z))_{i\in[t]}$ be (non-subsystem) CSS codes over some field $\bF_q$, and let $Q=(Q_X,Q_Z)=\bigotimes_{i\in[t]}Q^i$ be the subsystem product. If each $Q^i$ has dimension $k_i$, then $Q$ is a subsystem code of dimension $k=\prod_{i\in[t]}k_i$, and there is an isomorphism
  \begin{equation*}
    (Q_Z+Q_X^\perp)/Q_X^\perp \cong \bigotimes_{i\in[t]}(Q^i_Z/{Q^i_X}^\perp)
  \end{equation*}
  Furthermore, for $(a_i\in Q^i_Z)_{i\in[t]}$, the isomorphism above maps
  \begin{equation*}
    \bigotimes_{i\in[t]}a_i+Q_X^\perp \mapsfrom \bigotimes_{i\in[t]}(a_i+{Q^i_X}^\perp).
  \end{equation*}
\end{proposition}
\begin{proof}
  For each $i\in[t]$, let $k_i=\dim({Q^i_X}^\perp)$, $k_i'=\dim(Q^i_Z)$, and fix a basis $\{b^i_1,\dots,b^i_{k_i'}\}$ for $Q^i_Z$ such that $\{b^i_1,\dots,b^i_{k_i}\}$ is a basis for ${Q^i_X}^\perp$. Then by definition $\{b^i_{k_i+1}+{Q^i_X}^\perp,\dots,b^i_{k_i'}+{Q^i_X}^\perp\}$ is a basis for $Q^i_Z/{Q^i_X}^\perp$, and $\{\bigotimes_{i\in[t]}b^i_{j_i}+Q_X^\perp:j_i\in\{k_i+1,\dots,k_i'\}\;\forall i\in[t]\}$ is a basis for $(Q_Z+Q_X^\perp)/Q_X^\perp$, so the result follows.
\end{proof}

\subsection{Product-Expansion}
\label{sec:pe}
In this section, we review the notion of product-expansion; we generally follow the exposition in \cite{kalachev_two-sided_2023}, which introduced the general high-dimensional formulation of this notion, and in \cite{kalachev_maximally_2025}, which showed that random codes over large alphabets satisfy this high-dimensional notion. The two-dimensional version of product-expansion was used to construct asymptotically good qLDPC codes in \cite{panteleev_asymptotically_2022,leverrier_quantum_2022-1,dinur_good_2023}; the high-dimensional version was used to construct nearly good qLTCs in \cite{dinur_expansion_2024}.

Product-expansion is a property of a finite set of classical codes, which generalizes ordinary code distance to such larger collections of codes. At a high level, it describes the extent to which elements of the dual tensor (see Definition~\ref{def:classtensor}) of these codes have natural low-weight decompositions.

\subsubsection{Definition}

In this section, we define product-expansion. We will first need the following notation:

\begin{definition}
  \label{def:dtnot}
  Fix some $t\in\bN$ and some classical codes $C_i\subseteq\bF_q^{n_i}$ for $i\in[t]$. Then for every $i\in[t]$, define $C^{(i)}\subseteq\bigotimes_{j\in[t]}\bF_q^{n_j}=\bF_q^{\prod_{j\in[t]}n_j}$ by
  \begin{equation*}
    C^{(i)}=\left(\bigotimes_{j=1}^{i-1}\bF_q^{n_j}\right)\otimes C_i\otimes\left(\bigotimes_{j=i+1}^{t}\bF_q^{n_j}\right),
  \end{equation*}
  and for every $i,j\in[t]$, define
  \begin{equation*}
    C^{(i,j)} := C^{(i)}\cap C^{(j)}.
  \end{equation*}
  When $t$ is not clear from context, we write $C^{(i;t)}=C^{(i)}$ and $C^{(i,j;t)}=C^{(i,j)}$.
\end{definition}

In words, $C^{(i)}$ is the space of $t$-dimensional tensors for which every direction-$i$ column (in which the $i$th coordinate varies while the other $t-1$ coordinates are fixed) lies in $C_i$. Similarly, for $i\neq j$, then $C^{(i,j)}$ is the space of $t$-dimensional tensors for which every direction-$(i,j)$ plane (in which the $i$th and $j$th coordinates vary while the other $t-2$ coordinates are fixed) lies in $C_i\otimes C_j$.

Definition~\ref{def:classtensor} and Definition~\ref{def:dtnot} immediately imply that
\begin{equation}
  \label{eq:boxplus}
  C_1\boxplus\cdots\boxplus C_t = C^{(1)}+\cdots+C^{(t)}.
\end{equation}
Product-expansion measures the extent to which the sum on the right hand side of~(\ref{eq:boxplus}) can cancel to yield a low-weight element of $C_1\boxplus\cdots\boxplus C_t$ from a sum of high-weight elements of $C^{(1)},\dots,C^{(t)}$. Of course, some cancellations can occur simply because $C^{(i,j)}=C^{(i)}\cap C^{(j)}$ is nonzero for all $i,j$, assuming that $C_i,C_j$ are nonzero. Product-expansion is therefore defined to measure the extent to which ``non-trivial'' cancellations can occur in~(\ref{eq:boxplus}), which are not captured by the spaces $C^{(i,j)}$, as formalized below.


We will need the following modified notion of Hamming weight for elements of $C^{(i)}$.

\begin{definition}
  \label{def:dirham}
  Defining $C^{(i)}$ as in Definition~\ref{def:dtnot}, then for $c\in C^{(i)}$, we let $|c|_i$ denote the number of nonzero direction-$i$ columns in $c$. Formally,
  \begin{equation*}
    |c|_i = \Biggl|\biggl\{(k_1,\dots,k_{i-1},k_{i+1},\dots,k_t)\in\prod_{j\in[t]\setminus\{i\}}[n_j]:\exists k_i\in[n_i]\text{ with }c_{(k_1,\dots,k_t)}\neq 0\biggr\}\Biggr|.
  \end{equation*}
\end{definition}

We are now ready to define product-expansion.

\begin{definition}[\cite{kalachev_two-sided_2023}]
  \label{def:pe}
  Fix some $t\in\bN$ and some prime power $q$. The \textbf{product-expansion $\rho$} of a collection $(C_i\subseteq\bF_q^{n_i})_{i\in[t]}$ of classical codes is the largest real number $\rho\geq 0$ such that for every $c\in C_1\boxplus\cdots\boxplus C_t$, there exists a decomposition
  \begin{equation*}
    c=c_1+\cdots+c_t
  \end{equation*}
  with each $c_i\in C^{(i)}$ such that
  \begin{equation*}
    |c| \geq \rho\sum_{i\in[t]}n_i|c_i|_i.
  \end{equation*}
\end{definition}

For some basic intuition, observe that when $t=1$, a single code $C\subseteq\bF_q^n$ of distance $d$ has product-expansion $\rho=d/n$ equal to the relative distance of $C$.

\subsubsection{Known Examples}
We now present known classes of product-expanding codes. We begin with pairs of random or Reed-Solomon codes:

\begin{theorem}[\cite{kalachev_two-sided_2023,dinur_good_2023}]
  \label{thm:pe2rand}
  For every $\epsilon>0$, there exists a real number $\rho=\rho(\epsilon)>0$ such that the following holds: For every $n\in\bN$ and every $k_1,k_2\leq(1-\epsilon)n$, a uniformly random pair of codes $C_1,C_2\subseteq\bF_q^n$ of respective dimensions $k_1,k_2$ has product-expansion $\geq\rho$ with probability approaching $1$ as $n\rightarrow\infty$.
\end{theorem}

\begin{theorem}[\cite{polishchuk_nearly-linear_1994}]
  \label{thm:pe2RS}
  For every $\epsilon>0$, there exists $\rho=\rho(\epsilon)>0$ such that the following holds: For every prime power $q$, every choice of integers\footnote{\cite{polishchuk_nearly-linear_1994} only presented their proof for the case where $k_1=k_2$. However, the proof goes through flawlessly when $k_1\neq k_2$, as was for instance observed in \cite{kalachev_two-sided_2023}.} $k_1,k_2,n\leq q$ with $k_1+k_2\leq(1-\epsilon)n$, every choice of subsets $E_1,E_2\subseteq\bF_q$ of size $|E_1|=|E_2|=n$, the codes $(\evl_{E_1}(\bF_q[X]^{[0,k_1)}),\evl_{E_2}(\bF_q[X]^{[0,k_2)}))$ have product-expansion $\geq\rho$.
\end{theorem}

Higher-order product-expansion (of an arbitrary constant number of codes) also holds with high probability for random codes over a sufficiently large alphabet:

\begin{theorem}[\cite{kalachev_maximally_2025}]
  \label{thm:petrand}
  For every $\epsilon>0$ and every $t\in\bN$, there exists a real number $\rho=\rho(\epsilon,t)>0$ such that the following holds: For every $n\in\bN$ and every $k_1,\dots,k_t\leq(1-\epsilon)n$, letting $q=2^{(n+3)^t}$, then a uniformly random tuple of codes $C_1,\dots,C_t\subseteq\bF_q^n$ of respective dimensions $k_1,\dots,k_t$ has product-expansion $\geq\rho$ with probability approaching $1$ as $n\rightarrow\infty$.
\end{theorem}




\subsubsection{Technical Lemmas}
\label{sec:petech}
We now present some technical lemmas pertaining to product-expansion that we will use throughout the paper.

Sometimes we have a given decomposition $c=c_1+\cdots+c_t$ with each $c_i\in C^{(i)}$, and we wish to transform it into a low-weight decomposition $c'=c_1'+\cdots+c_t'$ guaranteed to exist by product-expansion. The following lemma shows that any two decompositions of $c$ differ by sums of elements of the $C^{(i,j)}$'s.

\begin{lemma}[Follows from \cite{kalachev_two-sided_2023}]
  \label{lem:homvanexp}
  Let $c\in C_1\boxplus\cdots\boxplus C_t$. For any two decompositions
  \begin{equation*}
    c=c_1+\cdots+c_t=c_1'+\cdots+c_t'
  \end{equation*}
  with $c_i,c_i'\in C^{(i)}$ for all $i\in[t]$, then there exist choices of $c_{i,j}\in C^{(i,j)}$ for $1\leq i<j\leq t$ such that
  \begin{equation}
    \label{eq:homvanexp}
    c_i-c_i' = \sum_{j=1}^{i-1}c_{j,i}-\sum_{j=i+1}^tc_{i,j}
  \end{equation}
  for all $i\in[t]$.
\end{lemma}

Lemma~\ref{lem:homvanexp} essentially follows from the discussion in \cite[Appendix~B]{kalachev_two-sided_2023}, though for completeness we provide a proof in Appendix~\ref{sec:peproofs}.

The following lemma shows that passing to subcodes preserves product-expansion, up to some loss.

\begin{lemma}[\cite{kalachev_maximally_2025}]
  \label{lem:pesub}
  For $t,n\in\bN$, if a tuple of codes $(C_i\subsetneq\bF_q^n)_{i\in[t]}$ has product-expansion at least $\rho>0$, then every tuple of subcodes $(C_i'\subseteq C_i)_{i\in[t]}$ has product-expansion at least $\rho^{2^t}/2^t$.
\end{lemma}




\subsection{Transversal $C^{r-1}Z$ Gates}
\label{sec:transgates}
In this section, we describe transversal $C^{r-1}Z$ gates on quantum codes. Specifically, we give a definition for subsystem codes (see Section~\ref{sec:subsystem}) that will be sufficient for our purposes, though more general definitions are also possible. Our presentation in this section is similar to that of \cite{krishna_towards_2019,golowich_asymptotically_2025,golowich_quantum_2024}.

Below, we define the $C^{r-1}Z$ gate over $q$-dits, as well as a related single-qudit gate we call $U^r$.

\begin{definition}
  \label{def:crz}
  For a finite field $\bF_q$ of characteristic $p$, an element $a\in\bF_q$, and an integer $r\geq 2$, the $r$-qudit gate (i.e.~unitary operator) $C^{r-1}Z_q^a:(\bC^{\bF_q})^{\otimes r}\rightarrow(\bC^{\bF_q})^{\otimes r}$ acts as\footnote{See Footnote~\ref{footnote:trdef} for the definition of the field trace map $\tr_{\bF_q/\bF_p}:\bF_q\rightarrow\bF_p$.}
  \begin{equation*}
    C^{r-1}Z_q^a\ket{z_1,\dots,z_r} = e^{2\pi i\tr_{\bF_q/\bF_p}(a\cdot z_1\cdots z_r)/p}\ket{z_1,\dots,z_r}
  \end{equation*}
  for $(z_1,\dots,z_r)\in\bF_q^r$, and the single-qudit gate $U_q^{r,a}:\bC^{\bF_q}\rightarrow\bC^{\bF_q}$ acts as
  \begin{equation*}
    U_q^{r,a}\ket{z} = e^{2\pi i\tr_{\bF_q/\bF_p}(a\cdot z^r)/p}\ket{z}
  \end{equation*}
  for $z\in\bF_q$. When the field (i.e.~local qudit dimension) is clear from context, we often omit the $q$ subscript. When $a=1$, we write $C^{r-1}Z^1=C^{r-1}Z$ and $U^{r,a}=U^r$.
\end{definition}

We now present our main definition of transversal $C^{r-1}Z$ and $U^r$ gates on subsystem codes; we will subsequently show that this definition corresponds to transversal gates in the usual sense.

\begin{definition}
  \label{def:transversal}
  For integers $r\geq 2$, $n\geq 1$, $\ell\geq 1$, let $(Q^h=(Q^h_X,Q^h_Z))_{h\in[r]}$ be an $r$-tuple of length-$n$ quantum subsystem codes over some field $\bF_q$, with specified injective encoding maps $(\Enc^h:\bF_q^\ell\rightarrow Q^h_Z/(Q^h_Z\cap{Q^h_X}^\perp))_{h\in[r]}$. We say that \textbf{$(Q^h,\Enc^h)_{h\in[r]}$ supports a transversal $C^{r-1}Z$ gate (on $\ell$ logical qudits)} if there exists a \textbf{coefficients vector $a\in\bF_q^n$} such that the following holds: for every $(z^h\in\bF_q^\ell)_{h\in[r]}$ and every $({z^h}'\in\Enc^h(z^h))_{h\in[r]}$, we have
  \begin{equation}
    \label{eq:transdef}
    \sum_{j\in[\ell]}z^1_j\cdots z^r_j = \sum_{j\in[n]}a_j\cdot{z^1_j}'\cdots{z^r_j}'.
  \end{equation}
  If $(Q^1,\Enc^1)=\cdots=(Q^r,\Enc^r)$, we say that \textbf{$(Q^1,\Enc^1)$ supports both a transversal $C^{r-1}Z$ gate and a transversal $U^r$ gate (on $\ell$ logical qudits)} with \textbf{coefficients vector $a$}.
\end{definition}

In the language of subsystem stabilizer codes, Definition~\ref{def:transversal} implies that if we fix all of our gauge qudits as well as all but $\ell$ of our logical qudits\footnote{These extra logical qudits are protected by the code's distance, so we may assume they are always kept in the $\ket{0}$ state. Meanwhile, the gauge qudits can be measured and set to $\ket{0}$ using low-weight measurements, assuming the code has small locality.} to be $\ket{0}$ in the $Z$ basis, then our code supports a transversal $C^{r-1}Z$ gate on the remaining $\ell$ logical qudits in the standard sense; Lemma~\ref{lem:transdef} below formalizes this statement. This result is essentially the same as in \cite[Lemma~2.11]{golowich_asymptotically_2025}, but we provide the proof for completeness in Appendix~\ref{sec:peproofs}, as \cite{golowich_asymptotically_2025} did not phrase their result for subsystem codes.

Below, we recall from Section~\ref{sec:notation} that for a set $S\subseteq\bF_q^n$, we let $\ket{S}\in(\bC^q)^{\otimes n}$ denote the uniform superposition over elements of $S$.


\begin{lemma}[Lemma~2.11 of \cite{golowich_asymptotically_2025}]
  \label{lem:transdef}
  For integers $r\geq 2$, $\ell\geq 1$, let $(Q^h,\Enc^h)_{h\in[r]}$ be an $r$-tuple of length-$n$ quantum subsystem codes $Q^h=(Q^h_X,Q^h_Z)$ over $\bF_q$ with associated injective encoding maps $(\Enc^h:\bF_q^\ell\rightarrow Q^h_Z/(Q^h_Z\cap{Q^h_X}^\perp))_{h\in[r]}$, which supports a transversal $C^{r-1}Z$ gate on $\ell$ logical qudits with coefficients vector $a\in\bF_q^n$ in the sense of Definition~\ref{def:transversal}. For $h\in[r]$, define an isometry $\Enc_{\bC}^h:(\bC^q)^{\otimes\ell}\rightarrow(\bC^q)^{\otimes n}$ by $\Enc_{\bC}^h\ket{z}=\ket{\Enc^h(z)}$. Then for every $\ket{\phi^1},\dots,\ket{\phi^r}\in(\bC^q)^{\otimes\ell}$, we have
  \begin{align*}
    \bigotimes_{j\in[n]}C^{r-1}Z_q^{a_j}(\Enc_{\bC}^1\otimes\cdots\Enc_{\bC}^r)\ket{\phi^1,\dots,\phi^r}
    &= (\Enc_{\bC}^1\otimes\cdots\Enc_{\bC}^r)\bigotimes_{j\in[\ell]}C^{r-1}Z_q\ket{\phi^1,\dots,\phi^r},
  \end{align*}
  where the $j$th $C^{r-1}Z_q$ gate in the tensor products above acts on the $j$th qudit of the $r$ length-$\ell$ or length-$n$ states given as the input.

  Similarly, if $(Q,\Enc)=(Q^1,\Enc^1)$ supports a transversal $U_q^r$ gate with coefficients vector $a$, then letting $\Enc_{\bC}=\Enc_{\bC}^1$, for every $\ket{\phi}\in(\bC^q)^{\otimes\ell}$ we have
  \begin{align*}
    \left(\bigotimes_{j\in[n]}U_q^{a_j}\right)\Enc_{\bC}\ket{\phi}
    &= \Enc_{\bC}\left(\bigotimes_{j\in[\ell]}U_q\right)\ket{\phi}
  \end{align*}
\end{lemma}

Formally, letting
\begin{equation*}
  {\cQ^h}' = \spn\left\{\ket{z+(Q_Z^h\cap {Q_X^h}^\perp)}:z\in Q_Z^h+{Q_X^h}^\perp\right\} \subseteq (\bC^q)^{\otimes n}
\end{equation*}
be the quantum code space for $Q'$ (see~(\ref{eq:qsubcodespace})), then given code states $(\ket{\psi^i}\in{\cQ^i}')_{i\in[r]}$, we can first ``gauge fix'' out code states by measuring the $Z$-gauge operators $Z^c$ for a generating set\footnote{Such a generating set is given by the set of (scalar multiples of) rows of the parity-check matrix $H_Z^h$ for $Q_Z^h$.} of $c\in{Q_Z^h}^\perp$ and performing corrections to map each $\ket{\psi^i}$ to a state inside
\begin{equation*}
  \spn\left\{\ket{z+(Q_Z^h\cap {Q_X^h}^\perp)}:z\in Q_Z^h\right\} \subseteq {\cQ^h}',
\end{equation*}
while preserving the encoded logical information. Lemma~\ref{lem:transdef} implies that we can then perform physical transversal $C^{r-1}Z^{a_j}$ gates on these gauge-fixed code states to induce logical transversal $C^{r-1}Z$ gates on $\ell$ logical (i.e.~encoded message) qudits, assuming that the remaining logical qudits are set to $\ket{0}$.

In this paper, we use algebraic codes over growing-sized fields $\bF_q$ to construct transversal $C^{r-1}Z$ gates. However, \cite{golowich_asymptotically_2025,nguyen_good_2025,golowich_quantum_2024} show how to reduce the alphabet size at a small cost in code parameters. Below we state such an alphabet-reduction result in our context of $C^{r-1}Z$ gates on subsystem codes; the proof is similar to that in \cite{nguyen_good_2025} (and also in \cite{golowich_quantum_2024}), so we simply sketch the main idea in Appendix~\ref{sec:peproofs}. We remark that the loss in parameters incurred by alphabet reduction could likely be slightly reduced by more carefully applying the concatenation techniques of \cite{golowich_asymptotically_2025,nguyen_good_2025}, but for simplicity we do not pursue this direction.

\begin{lemma}
  \label{lem:alphred}
  For $h\in[r]$, let $Q^h$ be an $[[n,k_h,d_h]]_q$ subsystem code of locality $w_h$ such that $(Q^h)_{h\in[r]}$ supports a transversal $C^{r-1}Z$ gate on $\ell$ logical $q$-dits. For every subfield $\bF_{q'}\subseteq\bF_q$, if $q={q'}^e$, then for $h\in[r]$ there exists a $[[n'=e^r\cdot n,\;k_h'=e\cdot k_h,\;d_h'\geq d_h]]_{q'}$ subsystem code ${Q'}^h$ of locality $w_h'\leq e^r\cdot w_h\cdot$ such that $({Q'}^h)_{h\in[r]}$ supports a transversal $C^{r-1}Z$ gate on $\ell$ logical $q'$-dits.
\end{lemma}

We will specifically apply Lemma~\ref{lem:alphred} to subsystem codes supporting transversal $CCZ$ that arise from Reed-Solomon codes, which have alphabet size $q=\poly(n)$. In this case, Lemma~\ref{lem:alphred} allows us to reduce the alphabet size to a constant (e.g.~$q=2$) with only a $\poly(\log n)$ loss in rate and relative distance.

\begin{remark}
  The work \cite{he_quantum_2025} shows how to perform \emph{addressable} transversal $C^{r-1}Z$ gates on quantum codes arising from algebraic classical codes, such as Reed-Solomon codes. Here addressability refers to the ability to only perform logical $C^{r-1}Z$ gates on selected subsets of logical qudits, by permuting the physical qudits and changing the coefficient vector $a\in\bF_q^n$. We believe that our codes also support such addressable transversal gates, but for simplicity in this paper we do not pursue this direction.
\end{remark}





\section{Products of Product-Expanding Codes}
\label{sec:peprod}
In this section, we prove general bounds on the distance of quantum homological product and subsystem product codes based on product-expanding classical codes. We then apply these bounds to multiple different instantiations. Specifically, we generalize and strengthen the results of \cite{bravyi_homological_2014} on homological products of single-sector chain complexes based on random codes. We also apply our techniques to homological products of single-sector chain complexes from Reed-Solomon codes, as well as to subsystem products of quantum Reed-Solomon codes. We will show how to perform transversal non-Clifford gates on this latter construction in Section~\ref{sec:transversal}.

Our main result on homological products is stated below.

\begin{theorem}
  \label{thm:sspe}
  Fix some $t\in\bN$ and some field $\bF_q$ of characteristic $2$. For each $i\in[t]$ let $\cC_i=(C_i,\partial_i)$ be a single-sector chain complex over $\bF_q$, and let $n_i=\dim C_i$. For each $i\in[t]$, let $\rho_i$ denote the product-expansion of the collection of codes $(B_*(\cC_1),\dots,B_*(\cC_{i-1}),Z_*(\cC_i))$. Then the homological product $\cA=(A,\partial^{\cA}):=\bigotimes_{i\in[t]}\cC_i$ has systolic distance
  \begin{equation}
    \label{eq:sspedis}
    d_*(\cA) \geq \prod_{i\in[t]}(\rho_in_i).
  \end{equation}
  
  Furthermore, if $t=2$, letting $\rho'$ denote the product-expansion of $(B_*(\cC_1),B_*(\cC_2))$ and letting $\Delta_i=d(Z_*(\cC_i))/n_i$ denote the relative distance of the classical code $Z_*(\cC_i)$, then the homological product $\cA$ has filling constant
  \begin{equation}
    \label{eq:sspeexp}
    \mu_*(\cA) \leq \frac{1}{\rho'\cdot\min\{\rho',\Delta_1,\Delta_2\}}.
  \end{equation}
\end{theorem}

Theorem~\ref{thm:sspe} may be applied to arbitrary CSS codes $(Q_X,Q_Z)$ with $\dim Q_X=\dim Q_Z$ by first transforming such codes into single-sector complexes using the following transformation; a similar method was used in \cite{bravyi_homological_2014}.

\begin{lemma}
  \label{lem:codetocomplex}
  Let $Q=(Q_X,Q_Z)$ be a CSS code of length $n$ over $\bF_q$ of characteristic $2$ such that $\dim Q_X=\dim Q_Z$. Fix some full-rank parity check matrices $H_X,H_Z$, so that $Q_X=\ker H_X$, $Q_Z=\ker H_Z$. Define a single-sector chain complex $\cC$ with vector space $C=\bF_q^n$ and $\partial=H_X^\top H_Z$. Then $Q$ equals the quantum code associated to $\cC$.
\end{lemma}
\begin{proof}
  Since $H_ZH_X^\top=0$ by the CSS orthogonality condition, we have $\partial^2=H_X^\top H_ZH_X^\top H_Z=0$. Therefore $\cC$ is a well-defined single-sector chain complex. Now because $H_X,H_Z$ are full-rank, $\ker\delta=\ker H_X=Q_X$ and $\ker\partial=\ker H_Z=Q_Z$, as desired.
\end{proof}

Our main result on subsystem products is stated below.

\begin{theorem}
  \label{thm:subpe}
  Fix some $t\in\bN$ and some field $\bF_q$. For each $i\in[t]$, let $Q^i=(Q^i_X,Q^i_Z)$ be a (non-subsystem) CSS code over $\bF_q$ of length $n_i$. For each $i\in[t]$, let $\rho^i_X$ (resp.~$\rho^i_Z$) denote the product-expansion of the collection of codes $({Q^1_Z}^\perp,\dots,{Q^{i-1}_Z}^\perp,Q^i_X)$ (resp.~$({Q^1_X}^\perp,\dots,{Q^{i-1}_X}^\perp,Q^i_Z)$). Then the subsystem product $Q=\bigotimes_{i\in[t]}Q^i$ has distance
  \begin{equation*}
    d \geq \min\left\{\prod_{i\in[t]}(\rho^i_Xn_i),\; \prod_{i\in[t]}(\rho^i_Zn_i)\right\}.
  \end{equation*}
\end{theorem}

We will first describe some applications of Theorem~\ref{thm:sspe} and Theorem~\ref{thm:subpe} in Section~\ref{sec:prodtwo} and Section~\ref{sec:prodmany} below, before proving the theorems in Section~\ref{sec:sspeproof}.

\subsection{Product of Two Codes}
\label{sec:prodtwo}
\cite{bravyi_homological_2014} showed that the homological product of $t=2$ two (appropriately sampled) random single-sector chain complexes yields asymptotically good $[[N=n^2,\Theta(N),\Theta(N)]]$ quantum codes of locality $w=\Theta(\sqrt{N})$. Theorem~\ref{thm:sspe} strengthens this result in a couple of ways. First, the following corollary shows that such random product codes are also locally testable:

\begin{corollary}
  \label{cor:ssperandom}
  For every $\epsilon>0$, there exist $\Delta=\Delta(\epsilon)>0$ and $\rho=\rho(\epsilon)>0$ such that the following holds: For $i\in[2]$, let $Q^i=(Q_X^i,Q_Z^i)$ be a uniformly random length-$n$, rate-$R_i$ CSS code over any $\bF_q$ of characteristic $2$ such that $\dim(Q_X^i)=\dim(Q_Z^i)$, and such that the chosen rates satisfy $R_1,R_2\leq 1-\epsilon$. Let $\cC_i$ be the single-sector complex obtained from $Q^i$ via Lemma~\ref{lem:codetocomplex}. Then with probability approaching $1$ as $n\rightarrow\infty$, the quantum code associated to the homological product $\cA=\cC_1\otimes\cC_2$ is a $[[n^2,\;R_1R_2\cdot n^2,\;\Delta\cdot n^2]]$ code that is locally testable with locality $\leq 2n$ and soundness $\geq\rho$.
\end{corollary}
\begin{proof}
  The corollary follows immediately from Theorem~\ref{thm:sspe}, Lemma~\ref{lem:codetocomplex}, and Theorem~\ref{thm:pe2rand}, as in a uniformly random CSS code $Q^i$, all the codes $Q_X^i,Q_Z^i,{Q_X^i}^\perp,{Q_Z^i}^\perp$ are uniformly random classical codes (of the specified dimensions).
\end{proof}

Below, we also show that for appropriate choices of the rates of $Q^1,Q^2$, the random CSS codes above can be replaced by quantum Reed-Solomon codes, thereby yielding an explicit construction of length-$N$ quantum codes with dimension and distance $\Theta(N)$, locality $O(\sqrt{N})$, and constant soundness. The key idea is to apply the product-expansion of Reed-Solomon codes shown by \cite{polishchuk_nearly-linear_1994} (see Theorem~\ref{thm:pe2RS}).

\begin{corollary}
  \label{cor:sspeRS}
  For every $\epsilon>0$, there exist $\Delta=\Delta(\epsilon)>0$ and $\rho=\rho(\epsilon)>0$ such that the following holds: For $i\in[2]$, let $Q^i=(Q_X^i,Q_Z^i)$ be a length-$n$, rate-$R_i$ quantum Reed-Solomon code over some $\bF_q$ of characteristic $2$ such that $\dim(Q_X^i)=\dim(Q_Z^i)$, and such that the chosen rates satisfy $R_1,R_2\leq 1-\epsilon$ and $|R_1-R_2|\geq\epsilon$. Let $\cC_i$ be the single-sector complex obtained from $Q^i$ via Lemma~\ref{lem:codetocomplex}. Then the quantum code associated to the homological product $\cA=\cC_1\otimes\cC_2$ is a $[[n^2,\;R_1R_2\cdot n^2,\;\Delta\cdot n^2]]$ code that is locally testable with locality $\leq 2n$ and soundness $\geq\rho$.
\end{corollary}
\begin{proof}
  The corollary follows immediately from Theorem~\ref{thm:sspe}, Lemma~\ref{lem:codetocomplex}, and Theorem~\ref{thm:pe2RS}. Specifically, assuming without loss of generality that $R_1-R_2\geq\epsilon$, then $B_*(\cC_1),\;Z_*(\cC_2)$ are classical Reed-Solomon codes of rates $(1-R_1)/2,\;(1+R_2)/2$ respectively, so Theorem~\ref{thm:pe2RS} ensures that this pair of codes has product-expansion at least some $\rho=\rho(\epsilon/2)>0$. We may similarly conclude that $(B_*(\cC_1),B_*(\cC_2))$ has product-expansion at least $\rho$, so we may apply Theorem~\ref{thm:sspe} to conclude the desired result.
\end{proof}

Below, we present a subsystem analogue of Corollary~\ref{cor:sspeRS}, by applying Theorem~\ref{thm:subpe} instead of Theorem~\ref{thm:sspe}. In Section~\ref{sec:transRS} below, we will show how to perform transversal non-Clifford gates on these subsystem codes. Note that Theorem~\ref{thm:subpe} also implies an analogous subsystem version of Corollary~\ref{cor:ssperandom}, which we do not state to avoid redundancy.

\begin{corollary}
  \label{cor:subpeRS}
  For every $\epsilon>0$, there exists $\Delta=\Delta(\epsilon)>0$ such that the following holds: for $i\in[2]$, let $Q^i=(Q^i_X,Q^i_Z)$ be a length-$n$, dimension-$k_i$ quantum Reed-Solomon code over some field $\bF_q$ such that
  \begin{align*}
    \dim(Q^1_X) &\leq (1-\epsilon)n \\
    \dim(Q^1_Z) &\leq (1-\epsilon)n \\
    \dim({Q^1_Z}^\perp)+\dim(Q^2_X) &\leq (1-\epsilon)n \\
    \dim({Q^1_X}^\perp)+\dim(Q^2_Z) &\leq (1-\epsilon)n.
  \end{align*}
  Then the subsystem product code $Q:=Q^1\otimes Q^2$ is a $[[n^2,\;k_1k_2,\;\Delta\cdot n^2]]_q$ subsystem code with locality $\leq 2n$.
\end{corollary}
\begin{proof}
  The corollary follows immediately from Theorem~\ref{thm:subpe} and Theorem~\ref{thm:pe2RS}. Specifically, by definition $Q^1_X,Q^1_Z$ are Reed-Solomon codes of dimension $\leq(1-\epsilon)n$ and hence have distance $\geq\epsilon n$. Furthermore, Theorem~\ref{thm:pe2RS} implies that both pairs $({Q^1_Z}^\perp,Q^2_X)$ and $({Q^1_X}^\perp,Q^2_Z)$ have product-expansion at least some $\rho=\rho(\epsilon)>0$. Therefore Theorem~\ref{thm:subpe} implies that $Q$ has distance $d\geq\Delta\cdot n^2$ for $\Delta=\Delta(\epsilon)=\epsilon\rho>0$, as desired.
\end{proof}

While product-expanding random codes have been used for the ``local codes'' in constructions of asymptotically good quantum LDPC codes \cite{panteleev_asymptotically_2022,leverrier_quantum_2022-1,dinur_good_2023}, the product-expansion of Reed-Solomon codes seems to be too weak for this purpose, as observed in \cite{kalachev_two-sided_2023}. Specifically, this application requires classical codes $C_1,C_2$ such that both $(C_1,C_2)$ and $(C_1^\perp,C_2^\perp)$ are product-expanding, which cannot be obtained using Theorem~\ref{thm:pe2RS} due to the requirement that $R_1+R_2<1$. \cite{kalachev_two-sided_2023} explains how this issue seems to stem from the fact that the dual of a Reed-Solomon code is another Reed-Solomon code, but $(C_1,C_2)$ will have poor product-expansion if $C_1^\perp\subseteq C_2$. It is therefore perhaps surprising that we are able to use the product-expansion of quantum Reed-Solomon codes in Corollary~\ref{cor:sspeRS} and Corollary~\ref{cor:subpeRS} to obtain quantum codes of linear dimension and distance with nontrivial locality.

\subsection{Higher Order Products}
\label{sec:prodmany}
Theorem~\ref{thm:sspe} also provides a way of obtaining higher-dimensional homological products with good distance, if we are given CSS codes satisfying appropriate product-expansion properties; \cite{bravyi_homological_2014} conjectured that such higher products with good distance would exist.




The following corollary resolves this conjecture of \cite{bravyi_homological_2014} in the affirmative over sufficiently large alphabets. Specifically, applying the product-expansion result of \cite{kalachev_maximally_2025} described in Theorem~\ref{thm:petrand} with Theorem~\ref{thm:sspe} immediately yields the following:

\begin{corollary}
  \label{cor:sspemanyrandom}
  For every $\epsilon>0$ and $t\in\bN$, there exists $\Delta=\Delta(\epsilon,t)>0$ such that the following holds: For $i\in[t]$, let $Q^i=(Q_X^i,Q_Z^i)$ be a uniformly random length-$n$, rate-$R_i$ CSS code over $\bF_q$, $q=2^{(n+3)^t}$, such that $\dim(Q_X^i)=\dim(Q_Z^i)$, and such that the chosen rates satisfy $R_i\leq 1-\epsilon$. Let $\cC_i$ be the single-sector complex obtained from $Q^i$ via Lemma~\ref{lem:codetocomplex}. Then with probability approaching $1$ as $n\rightarrow\infty$, the quantum code associated to the homological product $\cA=\cC_1\otimes\cdots\otimes\cC_t$ is a $[[n^t,\; (\prod_{i\in[t]}R_i)\cdot n^t,\; \Delta\cdot n^t]]_q$ code with locality $\leq tn$.
\end{corollary}
\begin{proof}
  By the definition of a random CSS code, $(Q_X^i)_{i\in[t]}$ and $(Q_Z^i)_{i\in[t]}$ are both collections of random linear codes in $\bF_q^n$, subject to the $i$'th code in each collection having rate $(1+R_i)/2$; by definition the codes within each collection are independently drawn, while there are dependencies between the collections. Therefore by Theorem~\ref{thm:petrand}, there exists some $\rho=\rho(\epsilon,t)>0$ such that both $(Q_X^i)_{i\in[t]}$ and $(Q_Z^i)_{i\in[t]}$ are $\rho$-product-expanding with probability $\rightarrow 1$ as $n\rightarrow\infty$. Then Theorem~\ref{thm:sspe} along with Lemma~\ref{lem:pesub} immediately implies the desired result.
\end{proof}

\begin{remark}
  \label{remark:subpemanyrandom}
  Note that by applying Theorem~\ref{thm:subpe} instead of Theorem~\ref{thm:sspe} above, we could also obtain a subsystem analogue of Corollary~\ref{cor:sspemanyrandom}; we omit the details to avoid redundancy.
\end{remark}

While Corollary~\ref{cor:sspemanyrandom} provides $[[N=n^t,\Theta(N),\Theta(N)]]_q$ codes, the alphabet size $q=2^{(n+3)^t}$ grows exponentially in the block length, and the locality $tn=tN^{1/t}$ grows as a small polynomial in the block length for large constant $t$. Reducing the alphabet size would require proving product-expansion of collections of $t$ codes over some smaller alphabet than achieved in Theorem~\ref{thm:petrand}, which remains an interesting open question. The locality could potentially be reduced using a locality-reduction procedure like that of \cite{hastings_quantum_2023}. However, \cite{hastings_quantum_2023} only considers binary alphabets, so their techniques would need to be generalized to larger alphabets in order to apply to the codes in Corollary~\ref{cor:sspemanyrandom}. 

\subsection{Proof of Distance Bounds}
\label{sec:sspeproof}
In this section, we prove Theorem~\ref{thm:sspe} and Theorem~\ref{thm:subpe}. The first bound~(\ref{eq:sspedis}) in Theorem~\ref{thm:sspe}, as well as the bound in Theorem~\ref{thm:subpe}, will follow directly from the following lemma; we will then prove the bound~(\ref{eq:sspeexp}) in Theorem~\ref{thm:sspe} separately.


Below, similarly as in Section~\ref{sec:pe}, given $t\in\bN$ and $n_1,\dots,n_t\in\bN$, for a subspace $V\subseteq\bF_q^{n_i}$ we let $V^{(i)}\subseteq\bigotimes_{i\in[t]}\bF_q^{n_i}$ denote the space of $t$-dimensional tensors in which every direction-$i$ column lies in $V$. For a linear map $f$ acting on $\bF_q^{n_i}$, we similarly let $f^{(i)}=I^{\otimes i-1}\otimes f\otimes I^{t-i}$ denote the map that simply applies $f$ to each direction-$i$ column. When $t$ is not clear from context, we write $V^{(i;t)}=V^{(i)}$ and $f^{(i;t)}=f^{(i)}$.

\begin{lemma}
  \label{lem:sspedis}
  Define all variables as in Theorem~\ref{thm:subpe}. Then for every
  \begin{equation*}
    a \in (Q_Z+Q_X^\perp)\setminus Q_X^\perp, 
  \end{equation*}
  it holds that
  \begin{equation}
    \label{eq:pedisbound}
    |a| \geq \prod_{i\in[t]}(\rho^i_Zn_i).
  \end{equation}
\end{lemma}
\begin{proof}
  We prove the lemma by induction on $t$. For the base case, if $t=1$, then $a\in Q^1_Z\setminus{Q^1_X}^\perp\subseteq Q^1_Z\setminus\{0\}$, so $|a|$ is at least the distance of the classical code $Q^1_Z$, which by definition equals $\rho^1_Zn_1$, and~(\ref{eq:pedisbound}) holds, as desired.

  For the inductive step, let $t\geq 2$, and assume that the lemma holds for $t-1$. Choose any $a_0\in\bigotimes_{i\in[t]}Q^i_Z$ and $a_i\in ({Q^i_X}^\perp)^{(i)}$ for $i\in[t]$ such that $a=a_0+a_1+\cdots+a_t$. By definition $a_0$ and $a_t$ both lie in ${Q^t_Z}^{(t)}$, so
  \begin{equation*}
    a_1+\cdots+a_{t-1}+(a_0+a_t) \in ({Q^1_X}^\perp)^{(1)}+\cdots+({Q^{t-1}_X}^\perp)^{(t-1)}+{Q^t_Z}^{(t)}.
  \end{equation*}
  Recalling that $\rho^t_Z$ denotes the product-expansion of $({Q^1_X}^\perp,\dots,{Q^{t-1}_X}^\perp,Q^t_Z)$, Definition~\ref{def:pe} with Lemma~\ref{lem:homvanexp} implies that there exist choices of $a_{i,t}\in ({Q^i_X}^\perp)^{(i)}\cap {Q^t_Z}^{(t)}$ for all $i\in[t-1]$ such that
  \begin{equation}
    \label{eq:applype}
    |a| \geq \rho^t_Zn_t\biggl|a_0+a_t+\sum_{i\in[t-1]}a_{i,t}\biggr|_t,
  \end{equation}
  where as in Definition~\ref{def:dirham}, $|\cdot|_t$ denotes the number of nonzero direction-$t$ columns in a given tensor.

  By the assumption that $a\notin Q_X^\perp=(\bigotimes_{i\in[t]}Q^i_X)^\perp$, there exist $x_i\in Q^i_X$ for $i\in[t]$ such that
  \begin{equation*}
    \left(\bigotimes_{i\in[t]}x_i\right)\cdot a \neq 0.
  \end{equation*}
  Define $a'\in\bigotimes_{i\in[t-1]}\bF_q^{n_i}$ by
  \begin{equation*}
    a' := (I^{\otimes t-1}\otimes {x_t}^\top)\biggl(a_0+a_t+\sum_{i\in[t-1]}a_{i,t}\biggr) = (I^{\otimes t-1}\otimes {x_t}^\top)\biggl(a_0+\sum_{i\in[t-1]}a_{i,t}\biggr),
  \end{equation*}
  where the second equality above holds because every direction-$t$ column of $a_t$ by definition lies in ${Q^t_X}^\perp$ and is therefore orthogonal to $x_t\in Q^t_X$. Thus we find that $a'=a_0'+\sum_{i\in[t-1]}a_i'$ for vectors
  \begin{align*}
    a_0' &:= (I^{\otimes t-1}\otimes{x_t}^\top)a_0 \in \bigotimes_{i\in[t-1]}Q^i_Z \\
    a_i' &:= (I^{\otimes t-1}\otimes{x_t}^\top)a_{i,t} \in ({Q^i_X}^\perp)^{(i;t-1)} \hspace{1em} \forall i\in[t-1],
  \end{align*}
  so in particular $a'\in\bigotimes_{i\in[t-1]}Q^i_Z+(\bigotimes_{i\in[t-1]}Q^i_X)^\perp$.
  Furthermore,
  \begin{equation*}
    \left(\bigotimes_{i\in[t-1]}x_i\right)\cdot a' = \left(\bigotimes_{i\in[t]}{x_i}^\top\right)a_0 = \left(\bigotimes_{i\in[t]}x_i\right)\cdot a \neq 0,
  \end{equation*}
  as for $i\geq 1$, every direction-$i$ column of $a_i$ and $a_i'$ lies in ${Q^i_X}^\perp$ and therefore is orthogonal to $x_i\in Q^i_X$. Note that here we use the fact that dotting a tensor with $\bigotimes_ix_i$ is equivalent to applying the functional $\bigotimes_i{x_i}^\top$ to the tensor. It follows that $a'\notin(\bigotimes_{i\in[t-1]}Q^i_X)^\perp$.

  Thus we have shown that $a'$ satisfies the conditions of the lemma statement with $t$ replaced by $t-1$, so the inductive hypothesis implies that
  \begin{equation}
    \label{eq:peapplyind}
    |a'| \geq \prod_{i\in[t-1]}(\rho^i_Zn_i).
  \end{equation}
  By definition, a given component of $a'$ can only be nonzero if the respective direction-$t$ column of $a_0+a_t+\sum_{i\in[t-1]}a_{i,t}$ is nonzero. Thus~(\ref{eq:applype}) and~(\ref{eq:peapplyind}) imply that
  \begin{equation*}
    |a| \geq \prod_{i\in[t]}(\rho^i_Zn_i),
  \end{equation*}
  completing the inductive step.
\end{proof}

Lemma~\ref{lem:sspedis} immediately implies Theorem~\ref{thm:subpe}:

\begin{proof}[Proof of Theorem~\ref{thm:subpe}]
  Lemma~\ref{lem:sspedis} implies that every $a\in(Q_Z+Q_X^\perp)\setminus Q_X^\perp$ has $|a|\geq\prod_{i\in[t]}(\rho^i_Zn_i)$. Swapping the roles of $Q_X$ and $Q_Z$, Lemma~\ref{lem:sspedis} similarly implies that every $a\in(Q_X+Q_Z^\perp)\setminus Q_Z^\perp$ has $|a|\geq\prod_{i\in[t]}(\rho^i_Xn_i)$. These two inequalities give the desired result.
\end{proof}

Below, we similarly apply Lemma~\ref{lem:sspedis} to prove the bound~(\ref{eq:sspedis}) in Theorem~\ref{thm:sspe}.

\begin{proof}[Proof of~(\ref{eq:sspedis}) in Theorem~\ref{thm:sspe}]
  For any $a\in Z_*(\cA)\setminus B_*(\cA)$, the K\"{u}nneth formula (Proposition~\ref{prop:sskunneth}, Corollary~\ref{cor:ssprodbases}, and Corollary~\ref{cor:ssprodcycles}) guarantees that
  \begin{equation*}
    a \in \bigotimes_{i\in[t]}Z_*(\cC_i)+\im(\partial^{\cA}) \subseteq \bigotimes_{i\in[t]}Z_*(\cC_i)+\sum_{i\in[t]}B_*(\cC_i)^{(i)},
  \end{equation*}
  and that
  \begin{equation*}
    a \notin \sum_{i\in[t]}B_*(\cC_i)^{(i)} = \left(\bigotimes_{i\in[t]}Z^*(\cC_i)\right)^\perp.
  \end{equation*}
  Specifically, if instead $a\in(\bigotimes_{i\in[t]}Z^*(\cC_i))^\perp$, then because Corollary~\ref{cor:ssprodbases} implies that $\bigotimes_{i\in[t]}Z^*(\cC_i)$ spans the cohomology $H^*(\cA)$, the nondegeneracy of the natural cohomology/homology bilinear form would imply that $a\in B_*(\cA)$, contradicting the assumption that $a\in Z_*(\cA)\setminus B_*(\cA)$.

  Thus applying Lemma~\ref{lem:sspedis} with $Q_X^i=Z^*(\cC_i)$ and $Q_Z^i=Z_*(\cC_i)$ implies that $|a|\geq\prod_{i\in[t]}(\rho_in_i)$, as desired.
\end{proof}

We now prove the second bound (Equation~(\ref{eq:sspeexp})) in Theorem~\ref{thm:sspe}.

\begin{proof}[Proof of~(\ref{eq:sspeexp}) in Theorem~\ref{thm:sspe}]
  Consider any boundary $b\in B_*(\cA)$. Letting $\mu$ denote the right-hand side of~(\ref{eq:sspeexp}), then our goal is to construct a filling of $b$ of weight $\leq\mu|b|$. To begin, choose an arbitrary filling $a\in A$ of $b$, so that
  \begin{equation*}
    \partial^{\cA}a = \partial_1^{(1)}a+\partial_2^{(2)}a = b.
  \end{equation*}
  Applying the $\rho'$-product expansion of $(B_*(\cC_1),B_*(\cC_2))$ to $b$, we conclude by Lemma~\ref{lem:homvanexp} that there exists some $e\in B_*(\cC_1)\otimes B_*(\cC_2)$ such that
  \begin{equation*}
    |b| \geq \rho' \cdot \bigl( n_1\cdot |e+\partial_1^{(1)}a|_1+n_2 \cdot |e+\partial_2^{(2)}a|_2 \bigr) \ .
  \end{equation*}
  Because $e\in\im(\partial_1)\otimes\im(\partial_2)$, there exist preimages $e_1,e_2\in A$ of $e$ under $\partial_1^{(1)},\partial_2^{(2)}$ respectively such that $e_1\in\bF_q^{n_1}\otimes\im(\partial_2)$ and $e_2\in\im(\partial_1)\otimes\bF_q^{n_2}$. Define $a':=a+e_1+e_2$. Then $\partial^{\cA}a'=\partial^{\cA}a+e+e=b$, and
  \begin{equation}
    \label{eq:addes}
    n_1|\partial_1^{(1)}a'|_1+n_2|\partial_2^{(2)}a'|_2 = n_1\cdot |e+\partial_1^{(1)}a|_1+n_2 \cdot |e+\partial_2^{(2)}a|_2 \leq \frac{|b|}{\rho'}.
  \end{equation}

  If $|b|/n_1n_2\geq\rho'\cdot\min\{\Delta_1,\Delta_2\}$, then $|a|\leq n_1n_2\leq\mu|b|$, so $a$ is the desired filling of $b$, and we are done. Therefore assume that $|b|/n_1n_2<\rho'\cdot\min\{\Delta_1,\Delta_2\}$. It follows that
  \begin{align*}
    \frac{|b|}{\rho'n_1} &<\Delta_2n_2 = d(Z_*(\cC_2)) \\
    \frac{|b|}{\rho'n_2} &<\Delta_1n_1 = d(Z_*(\cC_1)).
  \end{align*}
  By~(\ref{eq:addes}), $|\partial_1^{(1)}a'|_1\leq|b|/\rho'n_1$ and $|\partial_2^{(2)}a'|_2\leq|b|/\rho'n_2$. Now for $i=1,2$, let $S_i\subseteq[n_{3-i}]$ denote the set of direction-$i$ vectors in $a'$ that lie outside of $Z_*(\cC_i)$, so that $|S_i|=|\partial_i^{(i)}a'|_i$. Then we have shown that $a'\in A=\bF_q^{n_1\times n_2}$ is a matrix in which $|S_1|<d(Z_*(\cC_2))$ columns lie outside of $Z_*(\cC_1)$, and $|S_2|<d(Z_*(\cC_1))$ rows lie outside of $Z_*(\cC_2)$.

  Therefore there exists some $f\in Z_*(\cC_1)\otimes Z_*(\cC_2)$ that agrees with $a'$ in all columns in $[n_2]\setminus S_1$, and that agrees with $a'$ in all rows in $[n_1]\setminus S_2$. Specifically, we may obtain $f$ by first writing $a'|_{([n_1]\setminus S_2)\times([n_2]\setminus S_1)}\in C_1|_{[n_1]\setminus S_2}\otimes C_2|_{[n_2]\setminus S_1}$ as a sum of pure tensors $\sum_jc_{1,j}\otimes c_{2,j}$, and then letting $f=\sum_j\hat{c}_{1,j}\otimes\hat{c}_{2,j}$, where $\hat{c}_{i,j}$ denotes the unique element of $C_i$ whose restriction to components in $[n_i]\setminus S_{3-i}$ equals $c_{i,j}$.

  Now define $a'':=a'+f$. Then $\partial^{\cA}a''=\partial^{\cA}a'=b$ because $f\in\ker\partial^{\cA}$ by definition. Furthermore,
  \begin{align*}
    |a''|
    &\leq |S_1|\cdot|S_2| = |\partial_1^{(1)}a'|_1\cdot|\partial_2^{(2)}a'|_2 \leq \frac{|b|^2}{{\rho'}^2n_1n_2} \leq \frac{|b|}{{\rho'}^2} \leq \mu|b|,
  \end{align*}
  as desired, where the second inequality above holds by~(\ref{eq:addes}).
\end{proof}


\section{Decoding Dual Tensor Products of Reed-Solomon Codes}
\label{sec:dualtensordec}
In this section, we provide an efficient decoding algorithm for dual tensor products of classical Reed-Solomon codes (see Definition~\ref{def:classtensor} and Definition~\ref{def:RS}). This algorithm is the core component of our efficient decoders for quantum codes constructed as products of Reed-Solomon codes, as described in Section~\ref{sec:qdec} below.

We define the decoding task below. We present the definition for products of two codes, which will be sufficient for our purposes. Below, for an element $c\in\bF_q^n$ and a code $C\subseteq\bF_q^n$, we let $d(c,C)=\min_{c'\in C}|c-c'|$ denote the minimum Hamming distance from $c$ to an element of $C$.

\begin{definition}
  For classical codes $C_1\subseteq\bF_q^{n_1},C_2\subseteq\bF_q^{n_2}$, an \textbf{$\alpha$-decoder for $C_1\boxplus C_2$} is an algorithm that takes as input an element $c\in\bF_q^{n_1\times n_2}$ (along with a description of the codes $C_1,C_2$), and outputs some $\tilde{c}\in C_1\boxplus C_2$ satisfying $|\tilde{c}-c|\leq\alpha\cdot d(c,C_1\boxplus C_2)$.
\end{definition}

We now state our main result on decoding.

\begin{theorem}
  \label{thm:dualtensordec}
  For every $\epsilon>0$, there exists $\alpha=\alpha(\epsilon)>0$ such that there is a polynomial-time $\alpha$-decoder for every pair of Reed-Solomon codes of the same length whose rates sum to $\leq 1-\epsilon$.

  Formally, the $\alpha$-decoder takes as input $n\in\bN$, a prime power $q\geq n$, integers $k_1,k_2\in\bN$ with $k_1+k_2\leq(1-\epsilon)n$, subsets $E_1,E_2\subseteq\bF_q$ with $|E_1|=|E_2|=n$, and an element $c\in\bF_q^{E_1\times E_2}\cong\bF_q^{n\times n}$, and in time $\poly(n,q)$ (not depending on $\epsilon$) outputs some $\tilde{c}\in C_1\boxplus C_2$ satisfying $|\tilde{c}-c|\leq\alpha\cdot d(c,C_1\boxplus C_2)$, where $C_i=\evl_{E_i}(\bF_q[X_i]^{[0,k_i)})$.
\end{theorem}

Our decoding algorithm to prove Theorem~\ref{thm:dualtensordec} consists of three subroutines, given in Algorithm~\ref{alg:decoder}, Algorithm~\ref{alg:decoderclose}, and Algorithm~\ref{alg:decoderfinish}, which we analyze respectively in Lemma~\ref{lem:algtech}, Lemma~\ref{lem:algdeg}, and Lemma~\ref{lem:algfin}. In Section~\ref{sec:algdesc} below, we state these lemmas and show how to combine them to prove Theorem~\ref{thm:dualtensordec}. We then prove the lemmas in Section~\ref{sec:alganal}.

In the remainder of this section, we let $\epsilon,n,k_1,k_2,E_1,E_2,c,C_1,C_2$ be defined as in Theorem~\ref{thm:dualtensordec}. We also let $C_1',C_2',s,\rho$ be defined as in Algorithm~\ref{alg:decoder}, and we use the notation $X=(X_1,X_2)$, $x=(x_1,x_2)$, $E=E_1\times E_2$, and $E'=E_1'\times E_2'$ described in Algorithm~\ref{alg:decoder}.

\subsection{Algorithm Description}
\label{sec:algdesc}
In this section, we describe the three subroutines used by our $\alpha$-decoder, and prove Theorem~\ref{thm:dualtensordec} conditional on some lemmas analyzing these subroutines.

\begin{algorithm}[!p]
  \caption{\label{alg:decoder}
    The first of three subroutines making up the decoder in Theorem~\ref{thm:dualtensordec} for the dual tensor product of two Reed-Solomon codes $(C_i=\evl_{E_i}(\bF_q[X]^{[0,k_i)}))_{i\in[2]}$ satisfying $k_1+k_2\leq(1-\epsilon)n$.
    As in Lemma~\ref{lem:algtech}, we let $\rho=\rho(\epsilon/2)$ be the quantity from Theorem~\ref{thm:pe2RS}.
    As a shorthand, we denote $X=(X_1,X_2)$, $x=(x_1,x_2)$, $E=E_1\times E_2$, and $E'=E_1'\times E_2'$.
   We also let $s=\lceil\rho\epsilon n/1000\rceil$, and for $i\in[2]$ we let $C_i'=\evl_E(\bF_q[X_i]^{[0,k_i+s)})$.
  }
  
  \SetKwInOut{Input}{Input}
  \SetKwInOut{Output}{Output}

  \SetKwFunction{FnDecInit}{DecInit}
  \SetKwProg{Fn}{Function}{:}{}

  \Input{$c\in\bF_q^E$ satisfying $d(c,C_1\boxplus C_2)\leq d_0$ for $d_0$ defined as in Lemma~\ref{lem:algtech}}
  \Output{$c'\in C_1'\boxplus C_2'$ for which Lemma~\ref{lem:algtech} shows $c'-c$ is low-weight}

  \Fn{\FnDecInit{$c$}}{
    Compute some nonzero $e_0(X)\in\bF_q[X]^{[0,s]^2}$ such that
    \begin{equation}
      \label{eq:e0req}
      (e_0(x)c(x))_{x\in E} \in C_1'\boxplus C_2'.
    \end{equation} \\ \label{li:e0}
    Let
    \begin{align*}
      E_1' &\gets \{x_1\in E_1:e_0(x_1,X_2)\neq 0\in\bF_q[X_2]\} \\
      E_2' &\gets \{x_2\in E_2:e_0(X_1,x_2)\neq 0\in\bF_q[X_1]\}.
    \end{align*} \\ \label{li:Ep}
    Compute a basis $\cE$ of all $e(X)\in\bF_q[X]^{[0,s]^2}$ for which
    \begin{equation}
      \label{eq:basiscond}
      (e(x)c(x))_{x\in\supp(e_0)\cap E'} \in (C_1'\boxplus C_2')|_{\supp(e_0)\cap E'}.
    \end{equation} \\ \label{li:basis}
    Compute
    \begin{align*}
      g_1^{x_1}(X_2) &\gets \gcd\{e(x_1,X_2)\in\bF_q[X_2]:e(X)\in\cE\} \hspace{1em}\forall x_1\in E_1' \\
      g_2^{x_2}(X_1) &\gets \gcd\{e(X_1,x_2)\in\bF_q[X_1]:e(X)\in\cE\} \hspace{1em}\forall x_2\in E_2'.
    \end{align*} \\ \label{li:gcd}
    \While{$\exists(x_1,x_2)\in E_1'\times E_2'$ such that $g_1^{x_1}(x_2)=0$ or $g_2^{x_2}(x_1)=0$}{ \label{li:while}
      Remove $x_1$ from $E_1'$ and remove $x_2$ from $E_2'$. \\ \label{li:remove}
    }
    $($Re$)$compute a basis $\cE'$ of all $e(X)\in\bF_q[X]^{[0,s]^2}$ for which
    \begin{equation}
      \label{eq:basiscond2}
      (e(x)c(x))_{x\in\supp(e_0)\cap E'} \in (C_1'\boxplus C_2')|_{\supp(e_0)\cap E'}.
    \end{equation} \\ \label{li:basis2}
    $($Re$)$compute
    \begin{align*}
      {g_1'}^{x_1}(X_2) &\gets \gcd\{e(x_1,X_2)\in\bF_q[X_2]:e(X)\in\cE'\} \hspace{1em}\forall x_1\in E_1' \\
      {g_2'}^{x_2}(X_1) &\gets \gcd\{e(X_1,x_2)\in\bF_q[X_1]:e(X)\in\cE'\} \hspace{1em}\forall x_2\in E_2'.
    \end{align*} \\ \label{li:gcd2}
    For $i\in[2]$, for every $x_i\in E_i'$ with ${g_i'}^{x_i}\neq 1$, remove $x_i$ from $E_i'$. \\ \label{li:remove2}
    \KwRet{any $c'\in C_1'\boxplus C_2'$ with $c'|_{\supp(e_0)\cap E'}=c|_{\supp(e_0)\cap E'}$.} \label{li:return}
  }
\end{algorithm}

\begin{algorithm}[!t]
  \caption{\label{alg:decoderclose}
    The second of three subroutines making up the decoder in Theorem~\ref{thm:dualtensordec}.
    Here we carry over the definitions of $C_1',C_2',s,\rho$ and the notation $X=(X_1,X_2)$, $x=(x_1,x_2)$, $E=E_1\times E_2$, and $E'=E_1'\times E_2'$ from Algorithm~\ref{alg:decoder}.
  }
  
  \SetKwInOut{Input}{Input}
  \SetKwInOut{Output}{Output}

  \SetKwFunction{FnDecClose}{DecClose}
  \SetKwProg{Fn}{Function}{:}{}
  \Input{$c'\in C_1'\boxplus C_2'$ satisfying the conditions described in Lemma~\ref{lem:algdeg}}
  \Output{$c''\in C_1\boxplus C_2$ for which Lemma~\ref{lem:algdeg} shows $c''-c'$ is low-weight}

  \Fn{\FnDecClose{$c'$}}{
    Let $f(X)=\sum_{j=(j_1,j_2)}f_jX^j\in\bF_q[X]^{[0,n)^2}$ be the unique polynomial with $c'=\evl_E(f)$ \\ \label{li:f}
    \For{$j_1\in\{k_1,\dots,k_1+s-1\}$}{ \label{li:for1}
      Compute some $r_1^{j_1}\in\bF_q^{E_2}$ with $|r_1^{j_1}|\leq\epsilon n/25$ such that
      \begin{equation*}
        \evl_{E_2}\left(\sum_{j_2\in[n]}f_{(j_1,j_2)}X_2^{j_2}\right)-r_1^{j_1} \in C_2'.
      \end{equation*} \label{li:r1}
    }
    \For{$j_2\in\{k_2,\dots,k_2+s-1\}$}{ \label{li:for2}
      Compute some $r_2^{j_2}\in\bF_q^{E_1}$ with $|r_2^{j_2}|\leq\epsilon n/25$ such that
      \begin{equation*}
        \evl_{E_1}\left(\sum_{j_1\in[n]}f_{(j_1,j_2)}X_1^{j_1}\right)-r_2^{j_2} \in C_1'.
      \end{equation*} \label{li:r2}
    }
    \KwRet{
      \begin{equation*}
        c'' = c' - \sum_{j_1=k_1}^{k_1+s-1}\evl_{E_1}(X_1^{j_1})\otimes r_1^{j_1} - \sum_{j_2=k_2}^{k_2+s-1}r_2^{j_2}\otimes\evl_{E_2}(X_2^{j_2})
      \end{equation*}
    } \label{li:cpp}
  }
\end{algorithm}

\begin{algorithm}[!t]
  \caption{\label{alg:decoderfinish}
    The third of three subroutines making up the decoder in Theorem~\ref{thm:dualtensordec}.
    Here we carry over the definitions of $C_1',C_2',s,\rho$ and the notation $X=(X_1,X_2)$, $x=(x_1,x_2)$, $E=E_1\times E_2$, and $E'=E_1'\times E_2'$ from Algorithm~\ref{alg:decoder}.
  }
  
  \SetKwInOut{Input}{Input}
  \SetKwInOut{Output}{Output}

  \SetKwFunction{FnDecFinish}{DecFinish}
  \SetKwProg{Fn}{Function}{:}{}
  \Input{$y\in\bF_q^E$ satisfying the conditions described in Lemma~\ref{lem:algfin}.}
  \Output{$y'\in y+C_1\boxplus C_2$ for which Lemma~\ref{lem:algfin} shows $|y'|=O(d(y,C_1\boxplus C_2))$}

  \Fn{\FnDecClose{$y$}}{
    \While{True}{
      \If{$\exists c_1\in C_1\setminus\{0\},\; x_2\in E_2$ such that $|c_1-y|_{E_1\times\{x_2\}}|<\epsilon n/2$}{ \label{li:ifc1}
        $y\gets y-c_1\otimes\1_{x_2}$ \label{li:addc1}
      }
      \If{$\exists c_2\in C_2\setminus\{0\},\; x_1\in E_1$ such that $|c_2-y|_{\{x_1\}\times E_2}|<\epsilon n/2$}{ \label{li:ifc2}
        $y\gets y-\1_{x_1}\otimes c_2$ \label{li:addc2}
      }
      \If{Both {\bf if} statements above were not satisfied}{
        \KwRet{$y$}
      }
    }
  }
\end{algorithm}

Lemma~\ref{lem:algtech} below analyzes the first subroutine, namely Algorithm~\ref{alg:decoder}. Specifically, we show that if it is given an input $c\in\bF_q^E\cong\bF_q^{n\times n}$ that is not too far (i.e.~at most distance $d_0=O(n^2)$) from $C_1\boxplus C_2$, then it will output some $c'\in C_1'\boxplus C_2'$ for some Reed-Solomon codes $C_1',C_2'$ of slightly higher degree than $C_1,C_2$ respectively. That is, Algorithm~\ref{alg:decoder} solves a ``relaxed'' decoding problem where it is allowed to output an element of a larger code than the true codeword.

\begin{lemma}
  \label{lem:algtech}
  Define $\rho=\rho(\epsilon/2)>0$ to be the quantity from Theorem~\ref{thm:pe2RS}, so that every pair of Reed-Solomon codes of the same length (on arbitrary evaluation points) whose rates sum to $\leq1-\epsilon/2$ has product-expansion $\geq\rho$. Define
  \begin{align*}
    d_0
    &= \left(\frac{\rho\epsilon n}{1000}\right)^2.
  \end{align*}
  If $d_0\geq 1$ and $d(c,C_1\boxplus C_2)\leq d_0$, then Algorithm~\ref{alg:decoder} outputs $c'\in C_1'\boxplus C_2'$ satisfying
  \begin{equation*}
    |c'-c|\leq \frac{\rho\epsilon}{50}\cdot n^2,
  \end{equation*}
  where $C_1',C_2'$ are as defined in Algorithm~\ref{alg:decoder}.
\end{lemma}

Lemma~\ref{lem:algdeg} below analyzes the second subroutine, namely Algorithm~\ref{alg:decoderclose}. This subroutine takes as input an element $c'\in C_1'\boxplus C_2'$ that is close to an element of $C_1\boxplus C_2$ (where we recall that $C_1',C_2',s,\rho$ are defined as in Algorithm~\ref{alg:decoder}). Algorithm~\ref{alg:decoderclose} then outputs some $c''\in C_1\boxplus C_2$ that is close to $c'$. That is, Algorithm~\ref{alg:decoderclose} converts our ``relaxed'' decoding $c'\in C_1'\boxplus C_2'$ that was the output of Algorithm~\ref{alg:decoder} into a ``true'' decoding $c''\in C_1\boxplus C_2$.

\begin{lemma}
  \label{lem:algdeg}
  If the input $c'\in C_1'\boxplus C_2'$ to Algorithm~\ref{alg:decoderclose} satisfies $|c'-c|\leq \rho\epsilon n^2/50$, where $c=a+b$, $a\in C_1\boxplus C_2$, $|b|=d(c,C_1\boxplus C_2)\leq d_0$ as defined in Lemma~\ref{lem:algtech}, then for $i\in[2]$ there exists $H_i\subseteq E_i$ with $|H_i|\leq \epsilon n/25$ such that the output $c''$ of Algorithm~\ref{alg:decoderclose} lies in $C_1\boxplus C_2$ and satisfies
  \begin{equation*}
    c''|_{(E_1\setminus H_1)\times(E_2\setminus H_2)}=a|_{(E_1\setminus H_1)\times(E_2\setminus H_2)}.
  \end{equation*}
\end{lemma}

Lemma~\ref{lem:algfin} below analyzes the third subroutine, namely Algorithm~\ref{alg:decoderfinish}. This subroutine takes as input an elemnt $y\in\bF_q^E$ that equals the sum of a low-weight dual tensor codeword $a'\in C_1\boxplus C_2$ and a low-weight element $b\in\bF_q^E$, and outputs a coset element $y'\in y+C_1\boxplus C_2$ of weight proportional to $|b|$. Algorithm~\ref{alg:decoderfinish} will allow us to convert a decoding $c''\in C_1\boxplus C_2$ that was output by Algorithm~\ref{alg:decoderclose}, which was ``somewhat'' close to the original corrupted codeword $c\in\bF_q^E$, into a decoding $c'''\in C_1\boxplus C_2$ that is (up to constant factors) optimally close to $c$. We will specifically feed $y=c''-c$ into Algorithm~\ref{alg:decoderfinish}, and then if the algorithm outputs $y'$, we let $c'''=c+y'$.

\begin{lemma}
  \label{lem:algfin}
  Assume that the input $y\in\bF_q^E$ to Algorithm~\ref{alg:decoderfinish} satisfies $y=a'+b$ for some $a'\in C_1\boxplus C_2$ and $b\in\bF_q^E$ such that
  \begin{align*}
    a'|_{(E_1\setminus H_1)\times(E_2\setminus H_2)} &= 0 \\
    |b| &\leq d_0,
  \end{align*}
  where $H_i\subseteq E_i$ is some subset of size $|H_i|\leq\epsilon n/25$ for $i\in[2]$. Then Algorithm~\ref{alg:decoderfinish} terminates after at most $n^2$ iterations of the while loop, and its output $y'\in y+C_1\boxplus C_2$ satisfies
  \begin{equation*}
    |y'| \leq \frac{8|b|}{\epsilon}.
  \end{equation*}
\end{lemma}

Lemma~\ref{lem:runtime} below analyzes the runtime of Algorithm~\ref{alg:decoder}, Algorithm~\ref{alg:decoderclose}, and Algorithm~\ref{alg:decoderfinish}.

\begin{lemma}
  \label{lem:runtime}
  Algorithm~\ref{alg:decoder}, Algorithm~\ref{alg:decoderclose}, and Algorithm~\ref{alg:decoderfinish} can be implemented in time $\poly(n,q)$ (independent of $k_1,k_2,\epsilon$).
\end{lemma}
\begin{proof}
  In Algorithm~\ref{alg:decoder}, lines~\ref{li:e0}, \ref{li:basis}, \ref{li:basis2}, \ref{li:return} simply require solving a linear system of $\poly(n)$ equations over $\bF_q$, lines~\ref{li:gcd} and \ref{li:gcd2} require computing the greatest common divisor of $\poly(n)$ univariate polynomials over $\bF_q$, and the remaining lines simply perform basic $\poly(n)$-time manipulation of variables. Thus each line executes in time $\poly(n,q)$, and the while loop by definition executes for at most $n$ iterations, so the entire algorithm runs in time $\poly(n,q)$.

  In Algorithm~\ref{alg:decoderclose}, line~\ref{li:f} simply requires solving a linear system of $\poly(n)$ equations over $\bF_q$. Lines~\ref{li:r1} and~\ref{li:r2} can be implemented by running a $\poly(n,q)$-time decoder such as the Welch-Berlekamp decoder (see e.g.~\cite{guruswami_essential_2022}) for the Reed-Solomon codes $C_2'$ and $C_1'$ respectively. Line~\ref{li:cpp} by definition runs in time $\poly(n,q)$. Thus the entire algorithm runs in time $\poly(n,q)$.

  In Algorithm~\ref{alg:decoderfinish}, line~\ref{li:ifc1} can be implemented by running the Welch-Berlekamp decoder for $C_1$ on $y|_{E_1\times\{x_2\}}$ for every $x_2\in E_2$, and line~\ref{li:ifc2} can be implemented by running the Welch-Berlekamp decoder for $C_2$ on $y|_{\{x_1\}\times E_2}$ for every $x_1\in E_1$. Thus every iteration of the while runs in time $\poly(n,q)$, and Lemma~\ref{lem:algfin} implies that the algorithm terminates after $\leq n^2$ iterations, so the entire runtime is $\poly(n,q)$.
\end{proof}

We now combine the above lemmas to prove Theorem~\ref{thm:dualtensordec}.

\begin{proof}[Proof of Theorem~\ref{thm:dualtensordec}]
  Let
  \begin{align*}
    \alpha = \alpha(\epsilon) = \left(\frac{1000}{\rho\epsilon}\right)^2.
  \end{align*}
  
  First observe that if $d_0<1$ (for $d_0$ defined as in Lemma~\ref{lem:algtech}), then $\alpha>n^2$, so the $\alpha$-decoding proprty simply requires that we output $c$ if $c\in C_1\boxplus C_2$, and otherwise we may output any element of $C_1\boxplus C_2$. As we can check if $c\in C_1\boxplus C_2$ in time $\poly(n,q)$, it only remains to consider the case where $d_0\geq 1$.

  Assume that $d_0\geq 1$. Because $\alpha\geq n^2/d_0$, if $d(c,C_1\boxplus C_2)\leq d_0$ then we can again satisfy $\alpha$-decoding property by outputing $c$ if $c\in C_1\boxplus C_2$, and otherwise outputting any element of $C_1\boxplus C_2$. Thus if we have a decoder that succeeds on every input $c\in\bF_q^E$ with $d(c,C_1\boxplus C_2)\leq d_0$, then we can run this decoder, and if it fails (meaning that it either does not terminate in time $\poly(n,q)$, or does not output an element of $C_1\boxplus C_2$) then output any element of $C_1\boxplus C_2$.

  Thus it remains to consider the case where $d_0\geq 1$ and $d(c,C_1\boxplus C_2)\leq d_0$. For this purpose, we first run Algorithm~\ref{alg:decoder} on input $c$. By Lemma~\ref{lem:algtech}, the resulting output $c'\in C_1'\boxplus C_2'$ satisfies the condition $|c'-c|\leq\rho\epsilon n^2/50$ in Lemma~\ref{lem:algdeg}, so we may feed $c'$ into Algorithm~\ref{alg:decoderclose}. By Lemma~\ref{lem:algdeg}, the resulting output $c''\in C_1\boxplus C_2$ is such that $y:=c''-c$ satisfies the conditions in Lemma~\ref{lem:algfin} with $|b|=d(c,C_1\boxplus C_2)$, so we may feed $y$ into Algorithm~\ref{alg:decoderfinish}. Lemma~\ref{lem:algfin} then implies that the resulting output $y'\in y+C_1\boxplus C_2$ satisfies $|y'|\leq 8|b|/\epsilon=8d(c,C_1\boxplus C_2)/\epsilon$, so $c''':=c+y'$ satisfies $c'''\in C_1\boxplus C_2$ and $|c'''-c|\leq 8d(c,C_1\boxplus C_2)/\epsilon\leq \alpha d(c,C_1\boxplus C_2)$. Thus we may output $c'''$ to satisfy the desired $\alpha$-decoding property.

  Thus to summarize, we have shown that if $d_0\geq 1$ and $d(c,C_1\boxplus C_2)\leq d_0$, we can successively run Algorithm~\ref{alg:decoder}, Algorithm~\ref{alg:decoderclose}, and Algorithm~\ref{alg:decoderfinish} to output $c'''\in C_1\boxplus C_2$ satisfying the desired property $|c'''-c|\leq\alpha\cdot d(c,C_1\boxplus C_2)$. Lemma~\ref{lem:runtime} then implies that the combined runtime of these three subroutines is $\poly(n,q)$. Meanwhile, if instead $d_0<1$, we simply check if $c\in C_1\boxplus C_2$, and output $c$ if so, or any element of $C_1\boxplus C_2$ otherwise. If $d_0\geq 1$ but $d(c,C_1\boxplus C_1)\leq d_0$, we still run the three subroutines as described above (as we do not know the value of $d(c,C_1\boxplus C_1)$), but they may fail arbitrarily or not terminate in time $\poly(n,q)$, as our analysis of the subroutines' success and runtime was conditional on the input satisfying $d(c,C_1\boxplus C_2)$; in this case, we abort and can conclude that $d(c,C_1\boxplus C_1)>d_0$, so we can output any element of $C_1\boxplus C_2$. Note that if $d(c,C_1\boxplus C_2)>d_0$, the subroutines may still successfully output some arbitrary element of $C_1\boxplus C_2$ in time $\poly(n,q)$, in which case we will not detect a failure; but when $d(c,C_1\boxplus C_2)>d_0$ our $\alpha$-decoder can output any element of $C_1\boxplus C_2$, so our output will still be valid. Thus we have constructed an $\alpha$-decoder for $C_1\boxplus C_2$ whose overall running time is $\poly(n,q)$, as desired.
\end{proof}

\subsection{Analysis of Subroutines}
\label{sec:alganal}
In this section, we prove Lemma~\ref{lem:algtech}, Lemma~\ref{lem:algdeg}, and Lemma~\ref{lem:algfin}, which analyze Algorithm~\ref{alg:decoder}, Algorithm~\ref{alg:decoderclose}, and Algorithm~\ref{alg:decoderfinish} respectively, thereby completing the proof of Theorem~\ref{thm:dualtensordec}.

\begin{proof}[Proof of Lemma~\ref{lem:algtech}]
  Let $c=a+b$, where $a\in C_1\boxplus C_2$ and $|b|=d(c,C_1\boxplus C_2)\leq d_0$. We proceed through the Algorithm~\ref{alg:decoder} line-by-line, showing that each step successfully computes the desired quantities; we define all variables as in Algorithm~\ref{alg:decoder}.

  To begin, the computation of $e_0$ in line~\ref{li:e0} succeeds because $e_0(X)$ has $(s+1)^2>(\rho\epsilon n/1000)^2\geq d_0\geq|b|$ coefficients that can take on arbitrary values, and to satisfy~(\ref{eq:e0req}) it is sufficient to have $e_0(x)=0$ for every $x\in\supp(b)$. These equations $e_0(x)=0$ impose $|b|<(s+1)^2$ linear homogenous constraints on the $(s+1)^2$ coefficients of $e_0$, so there is some nonzero choice of $e_0$ satisfying all the equations.

  \begin{claim}
    \label{claim:E0size}
    For $i\in[2]$, let $E_i^0$ equal the set $E_i'$ following the execution of line~\ref{li:Ep}. Then $|E_i'|\geq n-s$.
  \end{claim}
  \begin{proof}
    Because $e_0(X)\neq 0$ has degree $\leq s$ in each variable, $e_0(x_1,X_2)$ and $e_0(X_1,x_2)$ can only vanish for $\leq s$ values of $x_1$ and $x_2$ respectively, so $|E_i^0|\geq n-s$.
  \end{proof}

  Now by definition $(e_0(x)c(x))_{x\in E}$ and $(e_0(x)a(x))_{x\in E}$ both lie in $\evl_E(\bF_q[X]^{[0,s]^2})*(C_1\boxplus C_2) = C_1'\boxplus C_2'$, so $(e_0(x)b(x))_{x\in E}$ must also lie in this code because $b=c-a$. Writing in a slight abuse of notation $\evl_E(e_0b):=(e_0(x)b(x))_{x\in E}$, then by definition $|\evl_E(e_0b)|\leq|b|\leq d_0$. Therefore by the definition of $\rho$ in the lemma statement, we may express
  \begin{equation}
    \label{eq:e0bdecomp}
    \evl_E(e_0b) = \evl_E(f_1)+\evl_E(f_2)
  \end{equation}
  for some $f_1(X)\in\bF_q[X]^{[0,n)\times[0,k_2+s)}$ and $f_2(X)\in\bF_q[X]^{[0,k_1+s)\times[0,n)}$ such that the sets
  \begin{align*}
    F_1 &= \{x_1\in E_1:f_1(x_1,X_2)\neq 0\in\bF_q[X_2]\} \\
    F_2 &= \{x_2\in E_2:f_2(X_1,x_2)\neq 0\in\bF_q[X_1]\}
  \end{align*}
  have size $|F_1|,|F_2|\leq d_0/\rho n$.

  \begin{claim}
    For $i\in[2]$, we have $F_i\subseteq E_i^0$.
  \end{claim}
  \begin{proof}
    If there were some $x_1\in F_1\setminus E_1^0$, then we would have $e_0(x_1,X_2)=0$, so $\evl_{\{x_1\}\times E_2}(e_0b)=0$. Yet we also must have that $\evl_{\{x_1\}\times E_2}(e_0b)=\evl_{\{x_1\}\times E_2}(f_1)+\evl_{\{x_1\}\times E_2}(f_2)$, where by definition $|\evl_{\{x_1\}\times E_2}(f_1)|\geq n-(k_2+s)>\epsilon n/2$ and $|\evl_{\{x_1\}\times E_2}(f_2)|\leq|F_2|\leq d_0/\rho n\leq\epsilon n/2$. Hence $|\evl_{\{x_1\}\times E_2}(e_0b)|>0$, contradicting the fact that $\evl_{\{x_1\}\times E_2}(e_0b)=0$, so we must indeed have $F_1\subseteq E_1^0$, as desired. The proof that $F_2\subseteq E_2^0$ is analogous.
  \end{proof}

  \begin{claim}
    \label{claim:gcddivprod}
    Following the execution of line~\ref{li:gcd}, it holds that
    \begin{align*}
      g_1^{x_1}(X_2) \mid \prod_{x_2\in F_2}(X_2-x_2) &\hspace{1em}\forall x_1\in E_1'\setminus F_1 \\
      g_2^{x_2}(X_1) \mid \prod_{x_1\in F_1}(X_1-x_1) &\hspace{1em}\forall x_2\in E_2'\setminus F_2.
    \end{align*}
  \end{claim}
  \begin{proof}
    It suffices to show that~(\ref{eq:basiscond}) is satisfied by letting $e=e_1$ for
    \begin{equation*}
      e_1(X) = \prod_{x_1\in F_1}(X_1-x_1)\cdot\prod_{x_2\in F_2}(X_2-x_2),
    \end{equation*}
    as then for every $x_1\in E_1'\setminus F_1$ and $x_2\in E_2'\setminus F_2$ we obtain the desired divisibility conditions
    \begin{align*}
      g_1^{x_1}(X_2)&\mid e_1(x_1,X_2)\propto\prod_{x_2\in F_2}(X_2-x_2) \\
      g_2^{x_2}(X_1)&\mid e_1(X_1,x_2)\propto\prod_{x_1\in F_1}(X_1-x_1).
    \end{align*}
    Note that here $e_1\in\bF_q[X]^{[0,s]^2}$ because $|F_1|,|F_2|\leq d_0/\rho n\leq s$.

    By~(\ref{eq:e0bdecomp}), $\evl_E(e_0b)$ is supported inside inside $F_1\times E_2\cup E_1\times F_2$, so $\evl_{\supp(e_0)\cap E'}(b)$ is supported inside $F_1\times E_2'\cup E_1'\times F_2$, and hence $\evl_{\supp(e_0)\cap E'}(e_1b)=0$. Therefore
    \begin{equation*}
      \evl_{\supp(e_0)\cap E'}(e_1c) = \evl_{\supp(e_0)\cap E'}(e_1a) \in \evl_{\supp(e_0)\cap E'}(\bF_q[X]^{([0,k_1+|F_1|)\times[0,n))\cup([0,n)\times[0,k_2+|F_2|))}),
    \end{equation*}
    so because $|F_1|,|F_2|\leq d_0/\rho n\leq s$, Equation~(\ref{eq:basiscond}) holds for $e=e_1$, as desired.
  \end{proof}

  \begin{claim}
    \label{claim:removebound}
    For $i\in[2]$, the while loop in lines~\ref{li:while}-\ref{li:remove} removes at most $|F_1|+|F_2|$ values from $E_i'$.
  \end{claim}
  \begin{proof}
    By Claim~\ref{claim:gcddivprod}, for $x=(x_1,x_2)\in(E_1'\setminus F_1)\times(E_2'\setminus F_2)$ (for the values of $E_1',E_2'$ at any point during the while loop), the roots of $g_1^{x_1}(X_2)$ and of $g_2^{x_2}(X_1)$ lie inside $F_2$ and $F_1$ respectively, so $g_1^{x_1}(x_2)\neq 0$ and $g_2^{x_2}(x_1)\neq 0$. Hence every time line~\ref{li:remove} executes to remove some $x_1,x_2$ from $E_1',E_2'$ respectively, we must have that either $x_1\in F_1$ or $x_2\in F_2$. It follows that line~\ref{li:remove} can execute at most $|F_1|+|F_2|$ times, at which point $E_1',E_2'$ would contain no more points in $F_1,F_2$ respectively, so the loop would have to terminate.
  \end{proof}

  \begin{claim}
    \label{claim:gcd1}
    For $i\in[2]$, let $E_i^1$ equal the set $E_i'$ following the termination of the while loop in lines~\ref{li:while}-\ref{li:remove}, but prior to the execution of line~\ref{li:remove2}. Then for every $x_i\in E_i^1\setminus F_i$, we have ${g_i'}^{x_i}=1$.
  \end{claim}
  \begin{proof}
    We show the result for $i=1$; the proof for $i=2$ is analogous. As the sets $E_1',E_2'$ only lose elements throughout the execution of the algorithm, every $e(X)\in\bF_q[X]^{[0,s]^2}$ that satisfies~(\ref{eq:basiscond}) will also satisfy~(\ref{eq:basiscond2}), and hence $\spn(\cE)\subseteq\spn(\cE')$. Therefore for every $x_1\in E_1^1$, we have ${g_1'}^{x_1}\mid g_1^{x_1}$.

    By Claim~\ref{claim:gcddivprod}, there exists some subset $F_2^{x_1}\subseteq F_2$ such that
    \begin{equation*}
      g_1^{x_1}(X_2) = \prod_{x_2\in F_2^{x_1}}(X_2-x_2).
    \end{equation*}
    Meanwhile, reapplying the exact same argument as in Claim~\ref{claim:gcddivprod}, but now to ${g_1'}^{x_1}$ instead of $g_i^{x_1}$, gives that
    \begin{equation*}
      {g_1'}^{x_1}(X_2) \mid \prod_{x_2\in F_2\cap E_2^1}(X_2-x_2),
    \end{equation*}
    and therefore
    \begin{equation*}
      {g_1'}^{x_1}(X_2) \mid \prod_{x_2\in F_2^{x_1}\cap E_2^1}(X_2-x_2),
    \end{equation*}
    But for every $x_2\in F_2^{x_1}$, then $g_1^{x_1}(x_2)=0$, and hence the while loop in lines~\ref{li:while}-\ref{li:remove} must at some point remove either $x_1$ from $E_1'$, or remove $x_2$ from $E_2'$. Because by definition $x_1\in E_1^1$, it follows that $x_2$ must have been removed from $E_2'$, meaning that $x_2\notin E_2^1$. Therefore we have shown that $F_2^{x_1}\cap E_2^1=\emptyset$, so ${g_1'}^{x_1}=1$, as desired.
  \end{proof}

  \begin{claim}
    \label{claim:F1p}
    Define
    \begin{align*}
      F_1' &= \{x_1\in F_1\cap E_1^1:e_0(x_1,X_2)\nmid f_1(x_1,X_2)\} \\
      F_2' &= \{x_2\in F_2\cap E_2^1:e_0(X_1,x_2)\nmid f_2(X_1,x_2)\}.
    \end{align*}
    Then for $i\in[2]$, it holds for every $x_i\in F_i'$ that ${g_i'}^{x_i}\neq 1$.
  \end{claim}
  \begin{proof}
    We show the result for $i=1$; the proof for $i=2$ is analogous. Fix some $x_1\in F_1'$. Then by definition $e_0(x_1,X_2)\nmid f_1(x_1,X_2)$, or equivalently,
    \begin{equation}
      \label{eq:h1def}
      h_1^{x_1}(X_2) := \frac{e_0(x_1,X_2)}{\gcd(e_0(x_1,X_2),f_1(x_1,X_2))} \neq 1.
    \end{equation}
    We now argue that $h_1^{x_1}|{g_1'}^{x_1}$, which in turn implies that ${g_1'}^{x_1}\neq 1$, as desired. Assume for a contradiction that $h_1^{x_1}\nmid {g_1'}^{x_1}$. Then there exists some $e(X)\in\cE'$ such that
    \begin{equation}
      \label{eq:hndiv}
      h_1^{x_1}(X_2)\nmid e(x_1,X_2).
    \end{equation}
    By definition $e(X)\in\bF_q[X]^{[0,s]^2}$ satisfies~(\ref{eq:basiscond2}), meaning that $\evl_{\supp(e_0)\cap E^1}(ec)\in(C_1'\boxplus C_2')|_{\supp(e_0)\cap E^1}$.
    But by definition $\evl_{\supp(e_0)\cap E^1}(ea) \in (C_1'\boxplus C_2')|_{\supp(e_0)\cap E^1}$, so as $b=c-a$, we also have $\evl_{\supp(e_0)\cap E^1}(eb) \in (C_1'\boxplus C_2')|_{\supp(e_0)\cap E^1}$.

    Therefore there exists some $v \in C_1'\boxplus C_2'$ such that
    \begin{equation}
      \label{eq:ebv}
      \evl_{\supp(e_0)\cap E^1}(eb)=v|_{\supp(e_0)\cap E^1}.
    \end{equation}
    By definition
    \begin{align*}
      |v|
      &\leq |\evl_{\supp(e_0)\cap E^1}(eb)|+|E|-|\supp(e_0)\cap E^1| \\
      &\leq |\evl_E(f_1)|+|\evl_E(f_2)| + |E\setminus E^1| + |bE^1\setminus\supp(e_0)| \\
      &\leq n|F_1|+n|F_2| + 2n(s+|F_1|+|F_2|) + ns \\
      &= 3n(|F_1|+|F_2|+s) \\
      &\leq 9ns.
    \end{align*}
    Note that the third inequality above holds because $|E_i\setminus E_i^1|\leq s+|F_1|+|F_2|$ by Claim~\ref{claim:E0size} and Claim~\ref{claim:removebound}, and because $E^1\setminus\supp(e_0)|\leq ns$ as $e_0(x_1',X_2)$ by definition equals a nonzero polynomial of degree $\leq s$ for each $x_1'\in E_1^1$. The fourth inequality above holds because $|F_1|,|F_2|\leq d_0/\rho n\leq s$. Now it follows by Theorem~\ref{thm:pe2RS} that we can express
    \begin{equation*}
      v = \evl_E(u_1)+\evl_E(u_2)
    \end{equation*}
    for some $u_1(X)\in\bF_q[X]^{[0,n)\times[0,k_2+s)}$ and $u_2(X)\in\bF_q[X]^{[0,k_1+s)\times[0,n)}$ such that the sets
    \begin{align*}
      U_1 &= \{x_1\in E_1:u_1(x_1,X_2)\neq 0\in\bF_q[X_2]\} \\
      U_2 &= \{x_2\in E_2:u_2(X_1,x_2)\neq 0\in\bF_q[X_1]\}
    \end{align*}
    have size $|U_1|,|U_2|\leq 9ns/\rho n=9s/\rho$.

    Thus by~(\ref{eq:ebv})
    \begin{align*}
      \evl_{\supp(e_0)\cap(\{x_1\}\times(E_2^1\setminus U_2))}(eb)
      &= \evl_{\supp(e_0)\cap(\{x_1\}\times(E_2^1\setminus U_2))}(u_1).
    \end{align*}
    Meanwhile, (\ref{eq:e0bdecomp}) implies that $b$ agrees with $f_1/e_0$ at all points $(x_1,x_2)\in E\cap\supp(e_0)$ with $x_2\notin F_2$, which implies that
    \begin{align*}
      \evl_{\supp(e_0)\cap(\{x_1\}\times(E_2^1\setminus F_2))}(eb)
      &= \evl_{\supp(e_0)\cap(\{x_1\}\times(E_2^1\setminus F_2))}(ef_1/e_0)
    \end{align*}
    (where here we extend our $\evl$ notation to allow for rational functions that do not vanish within the evaluation set), and therefore
    \begin{align*}
      \evl_{\supp(e_0)\cap(\{x_1\}\times(E_2^1\setminus(U_2\cup F_2)))}(e_0u_1)
      &= \evl_{\supp(e_0)\cap(\{x_1\}\times(E_2^1\setminus(U_2\cup F_2)))}(ef_1).
    \end{align*}
    Dividing both $e_0(x_1,X_2)$ and $f_1(x_1,X_2)$ by $g_0(X_2):=\gcd(e_0(x_1,X_2),f_1(x_1,X_2))$, and recalling the definition of $h_1^{x_1}$ from~(\ref{eq:h1def}), then letting $S=\supp(e_0)|_{\{x_1\}\times E_2^1}\cap(E_2^1\setminus(U_2\cup F_2))$, it follows that
    \begin{align}
      \label{eq:polysdiff}
      \evl_S(h_1^{x_1}(X_2)u_1(x_1,X_2))
      &= \evl_S\left(e(x_1,X_2)\cdot\frac{f_1(x_1,X_2)}{g_0(X_2)}\right),
    \end{align}
    where $h_1^{x_1}(X_2)$ and $f_1(x_1,X_2)/g_0(X_2)$ by definition share no common factors. But by~(\ref{eq:h1def}) and~(\ref{eq:hndiv}), $h_1^{x_1}(X_2)\neq 1$ and $h_1^{x_1}(X_2)\nmid e(x_1,X_2)$, so it follows that $h_1^{x_1}(X_2)$ does not divide the polynomial on the RHS of~(\ref{eq:polysdiff}), and hence the polynomials on the LHS and RHS of~(\ref{eq:polysdiff}) are not equal. Therefore for~(\ref{eq:polysdiff}) to hold, we must have
    \begin{align}
      \label{eq:degvssize}
      |S| &\leq \deg\left(h_1^{x_1}(X_2)u_1(x_1,X_2) - e(x_1,X_2)\cdot\frac{f_1(x_1,X_2)}{g_0(X_2)}\right).
    \end{align}
    But by definition
    \begin{align*}
      |S|
      &\geq |E_2^1|-\deg(e_0(x_1,X_2))-|U_2|-|F_2| \\
      &\geq (n-s) - (|F_1|+|F_2|) - s - 9s/\rho - |F_2| \\
      &\geq n-5s-9s/\rho \\
      &> (1-\epsilon/2)n,
    \end{align*}
    where the third inequality holds because $|F_1|,|F_2|\leq s$ and $|U_2|\leq 9s/\rho$, as shown previously, and the fourth inequality holds because $s=\lceil\rho\epsilon n/1000\rceil\leq\rho\epsilon n/500$ as $d_0\geq 1$. Meanwhile,
    \begin{align*}
      \hspace{1em}&\hspace{-1em}\deg\left(h_1^{x_1}(X_2)u_1(x_1,X_2) - e(x_1,X_2)\cdot\frac{f_1(x_1,X_2)}{g_0(X_2)}\right) \\
                  &\leq \max\left\{\deg(e_0(x_1,X_2))+\deg(u_1(x_1,X_2)),\; \deg(e(x_1,X_2))+\deg(f_1(x_1,X_2))\right\} \\
                  &\leq k_2+2s \\
                  &< (1-\epsilon/2)n,
    \end{align*}
    where the final inequality again applies the fact that $s\leq\rho\epsilon n/500$. But the above two inequalities contradict~(\ref{eq:degvssize}). Thus the original assumption that $h_1^{x_1}\nmid{g_1'}^{x_1}$ was false, as desired.
  \end{proof}

  We now show that there exists a valid choice of $c'$ in line~\ref{li:return}. For $i\in[2]$, let $E_i^2$ equal the set $E_i'$ after the execution of line~\ref{li:remove2}. By Claim~\ref{claim:F1p}, $F_i'\cap E_i^2=\emptyset$, so there exist $f_1'(X)\in\bF_q[X]^{[0,n)\times[0,k_2+s)}$ and $f_2'(X)\in\bF_q[X]^{[0,k_1+s)\times[0,n)}$ such that
  \begin{align*}
    f_1'(x_1,X_2) &= \frac{f_1(x_1,X_2)}{e_0(x_1,X_2)} \hspace{1em}\forall x_1\in E_1^2 \\
    f_2'(X_1,x_2) &= \frac{f_2(X_1,x_2)}{e_0(X_1,x_2)} \hspace{1em}\forall x_2\in E_2^2,
  \end{align*}
  where here we recall that $f_1(x_1,X_2)$ and $f_2(X_1,x_2)$ vanish for every $x_1\in E_1\setminus F_1$ and $x_2\in E_2\setminus F_2$, respectively. Then define
  \begin{equation*}
    c_0' = a+\evl_{E^2}(f_1'+f_2') \in C_1'\boxplus C_2'.
  \end{equation*}

  \begin{claim}
    We have $c_0'|_{\supp(e_0)\cap E^2}=c|_{\supp(e_0)\cap E^2}$; that is, $c_0'$ is a valid choice of $c'$ in line~\ref{li:return}.
  \end{claim}
  \begin{proof}
    By~(\ref{eq:e0bdecomp}) along with the definition of $f_1',f_2'$, at every point $x\in\supp(e_0)\cap E^2$ we must have $b(x)=(f_1(x)+f_2(x))/e_0(x)=f_1'(x)+f_2'(x)$, and hence $c(x)=a(x)+b(x)=a(x)+f_1'(x)+f_2'(x)=c_0'(x)$, as desired.
  \end{proof}

  The above claims together imply that if $d_0\geq 1$ and $d(c,C_1\boxplus C_2)\leq d_0$, then Algorithm~\ref{alg:decoder} successfully outputs some $c'\in C_1'\boxplus C_2'$, which satisfies $c'|_{\supp(e_0)\cap E^2}=c|_{\supp(e_0)\cap E^2}$. By Claim~\ref{claim:gcd1}, line~\ref{li:remove2} can only remove points in $F_i$ from $E_i'$, so
  \begin{equation*}
    |E_i^2| \geq |E_i^1|-|F_i| \geq |E|-(s+|F_1|+|F_2|+|F_i|) \geq n-4s.
  \end{equation*}
  where the first inequality above holds by Claim~\ref{claim:E0size} and Claim~\ref{claim:removebound}, and the second inequality holds because $|F_1|,|F_2|\leq d_0/\rho n\leq s$. Therefore because for every $x_1\in E_1^2$, the polynomial $e_0(x_1,X_2)$ is nonzero and has degree $\leq s$, and hence has $\leq s$ roots, we have
  \begin{equation*}
    |\supp(e_0)\cap E^2| \geq (n-4s)^2-(n-4s)s \geq n-9sn \geq \left(1-\frac{\rho\epsilon}{50}\right)\cdot n^2.
  \end{equation*}
  Thus
  \begin{equation*}
    |c'-c| \leq n^2-|\supp(e_0)\cap E^2| \leq \frac{\rho\epsilon}{50}\cdot n^2,
  \end{equation*}
  as desired.
\end{proof}

\begin{proof}[Proof of Lemma~\ref{lem:algdeg}]
  Let $b'=c'-a\in C_1'\boxplus C_2'$, so that $|b'|\leq|c'-c|+|c-a|\leq\rho\epsilon n^2/25$. Then by Theorem~\ref{thm:pe2RS}, we may express
  \begin{equation*}
    b'=\evl_E(h_1(X)+h_2(X))
  \end{equation*}
  for some $h_1(X)\in\bF_q[X]^{[0,n)\times[0,k_2+s)}$ and $h_2(X)\in\bF_q[X]^{[0,k_1+s)\times[0,n)}$ such that the sets
  \begin{align*}
    H_1 &= \{x_1\in E_1:h_1(x_1,X_2)\neq 0\in\bF_q[X_2]\} \\
    H_2 &= \{x_2\in E_2:h_2(X_1,x_2)\neq 0\in\bF_q[X_1]\}
  \end{align*}
  have size $|H_1|,|H_2|\leq\epsilon n/25$.
  
  Define $f(X)=\sum_jf_jX^j$ as in line~\ref{li:f} of Algorithm~\ref{alg:decoderclose}. By definition $f(X)\in\bF_q[X]^{([0,k_1+s)\times[0,n))\cup([0,n)\times[0,k_2+s))}$, and furthermore $\evl_E(f-h_1-h_2)=c'-b'=a\in C_1\boxplus C_2$, so $f-h_1-h_2\in\bF_q[X]^{([0,k_1)\times[0,n))\cup([0,n)\times[0,k_2))}$. Therefore letting $h_i(X)=\sum_{j=(j_1,j_2)}h_{i,j}X^j$, then for every $j_1\in\{k_1,\dots,k_1+s-1\}$ we have
  \begin{align*}
    \left(\sum_{j_2\in[n]}f_{(j_1,j_2)}X_2^{j_2}\right) - \left(\sum_{j_2\in[n]}h_{2,(j_1,j_2)}X_2^{j_2}\right)
    &\in \bF_q[X_2]^{[0,k_2+s)}.
  \end{align*}
  Hence
  \begin{align*}
    \evl_{E_2}\left(\sum_{j_2\in[n]}f_{(j_1,j_2)}X_2^{j_2}\right) - \evl_{E_2}\left(\sum_{j_2\in[n]}h_{2,(j_1,j_2)}X_2^{j_2}\right)
    &\in \evl_{E_2}(\bF_q[X_2]^{[0,k_2+s)}) = C_2'.
  \end{align*}
  For $x_2\in E_2$, by definition $\sum_{j_2\in[n]}h_{2,(j_1,j_2)}x_2^{j_2}$ is precisely the degree-$j_1$ coefficient of the polynomial $h_2(X_1,x_2)$. But by definition $h_2(X_1,x_2)=0$ for every $x_2\in E_2\setminus H_2$, and hence $\evl_{E_2}(\sum_{j_2\in[n]}h_{2,(j_1,j_2)}X_2^{j_2})$ is supported inside $H_2$, and therefore has weight at most $|H_2|\leq\epsilon n/25$. Therefore $\evl_{E_2}(\sum_{j_2\in[n]}h_{2,(j_1,j_2)}X_2^{j_2})$ is a valid choice of $r_1^{j_1}$ in line~\ref{li:r1}. Furthermore, because $C_2'$ is a Reed-Solomon code of distance $\geq\epsilon n/2$, this choice of $r_1^{j_1}$ is unique.

  Thus we have shown that for every $j_1\in\{k_1,\dots,k_1+s-1\}$, line~\ref{li:r1} computes
  \begin{align*}
    r_1^{j_1} &= \evl_{E_2}\left(\sum_{j_2\in[n]}h_{2,(j_1,j_2)}X_2^{j_2}\right),
  \end{align*}
  which is supported inside $H_2$. By analagous reasoning, for every $j_2\in\{k_2,\dots,k_2+s-1\}$, line~\ref{li:r2} computes
  \begin{align*}
    r_2^{j_2} &= \evl_{E_1}\left(\sum_{j_1\in[n]}h_{1,(j_1,j_2)}X_1^{j_1}\right),
  \end{align*}
  which is supported inside $H_2$. Thus the value $c''$ returned in line~\ref{li:cpp} by definition agrees with $c'$, and therefore also with $a=c'-b'$, inside $(E_1\setminus H_1)\times(E_2\setminus H_2)$. Furthermore, by definition 
  \begin{align*}
    c''
    &= \evl_E\left(f(X)-h_1(X)-h_2(X) + \sum_{j_1=0}^{n-1}\sum_{j_2=0}^{k_2-1}h_{1,j=(j_1,j_2)}X^j + \sum_{j_1=0}^{k_1-1}\sum_{j_2=0}^{n-1}h_{2,j=(j_1,j_2)}X^j\right) \\
    &= a + \evl_E\left(\sum_{j_1=0}^{n-1}\sum_{j_2=0}^{k_2-1}h_{1,j=(j_1,j_2)}X^j + \sum_{j_1=0}^{k_1-1}\sum_{j_2=0}^{n-1}h_{2,j=(j_1,j_2)}X^j\right) \\
    &\in C_1\boxplus C_2,
  \end{align*}
  as desired.
\end{proof}

\begin{proof}[Proof of Lemma~\ref{lem:algfin}]
  First note that by definition, Algorithm~\ref{alg:decoderfinish} only adds elements of $C_1\boxplus C_2$ to $y$, so by definition its output $y'$ lies in $y+C_1\boxplus C_2$. The desired bound on $|y'|$ will follow from the following three claims.

  \begin{claim}
    \label{claim:ystruct}
    Define
    \begin{align*}
      H_1' &= H_1\cup\left\{x_1\in E_1:|b|_{\{x_1\}\times E_2}|\geq\frac{\epsilon n}{100}\right\} \\
      H_2' &= H_2\cup\left\{x_2\in E_2:|b|_{E_1\times\{x_2\}}|\geq\frac{\epsilon n}{100}\right\}.
    \end{align*}
    Then $|H_1'|,|H_2'|\leq \epsilon n/20$. Furthermore, at every point during Algorithm~\ref{alg:decoderfinish}, $y-b\in C_1\boxplus C_2$ satisfies $(y-b)|_{(E_1\setminus H_1')\times(E_2\setminus H_2')}=0$.
  \end{claim}
  \begin{proof}
    For $i\in[2]$, by definiton
    \begin{equation*}
      |H_i'| \leq |H_i|+\frac{|b|}{\epsilon n/100} \leq \frac{\epsilon n}{25}+\frac{d_0}{\epsilon n/100} \leq \frac{\epsilon n}{20}.
    \end{equation*}
    Thus it remains to be shown that at every point in the execution of Algorithm~\ref{alg:decoderfinish}, $(y-b)|_{(E_1\setminus H_1')\times(E_2\setminus H_2')}=0$. Assume for a contradiction that this statement does not hold, and let $y_0$ be the value of $y$ directly before the first moment that either line~\ref{li:addc1} or line~\ref{li:addc2} adds to $y$ some $c_1\otimes\1_{x_2}$ or $\1_{x_1}\otimes c_2$, respectively, that does not vanish inside $(E_1\setminus H_1')\times(E_2\setminus H_2')$. Specifically assume that line~\ref{li:addc1} adds such a $c_1\otimes\1_{x_2}$; the argument for lin~\ref{li:addc2} is analogous. Then $|c_1-y_0|_{E_1\times\{x_2\}}|<\epsilon n/2$. By assumption $(c_1\otimes\1_{x_2})|_{(E_1\setminus H_1')\times(E_2\setminus H_2')}\neq 0$, and therefore $x_2\in E_2\setminus H_2'$. Thus $(y_0-b)|_{(E_1\setminus H_1')\times\{x_2\}}=0$ and $|b|_{E_1\times\{x_2\}}|\leq\epsilon n/100$, so $y_0|_{E_1\times\{x_2\}}\leq|H_1'|+\epsilon n/100\leq\epsilon n/10$. Therefore
    \begin{equation*}
      |c_1| \leq |c_1-y_0|_{E_1\times\{x_2\}}|+|y_0|_{E_1\times\{x_2\}}| < \frac{\epsilon n}{2}+\frac{\epsilon n}{10} < \epsilon n,
    \end{equation*}
    which contradicts the fact that by definition $c_1\in C_1\setminus\{0\}$ and thus $|c_1|\geq\epsilon n$, as $C_1$ is a code of distance $\geq\epsilon n$. Therefore the assumption that we have $(y-b)|_{(E_1\setminus H_1')\times(E_2\setminus H_2')}=0$ at some point in Algorithm~\ref{alg:decoderfinish} was false, as desired.
  \end{proof}

  \begin{claim}
    \label{claim:ysmall}
    The while loop in Algorithm~\ref{alg:decoderfinish} can only terminate when $|y|\leq 8|b|/\epsilon$.
  \end{claim}
  \begin{proof}
    Fix the value of $y$ at be beginning of some iteration of the while loop. By Claim~\ref{claim:ystruct}, we may write
    \begin{equation}
      \label{eq:ydecomp}
      y = \evl_E(h_1'+h_2')+b
    \end{equation}
    for some $h_1'(X)\in\bF_q[X]^{[0,n)\times[0,k_2)}$ and $h_2'(X)\in\bF_q[X]^{[0,k_1)\times[0,n)}$ such that $h_1'(x_1,X_2)=0$ for every $x_1\in E_1\setminus H_1'$, and $h_2'(X_1,x_2)=0$ for every $x_2\in E_2\setminus H_2'$. Let
    \begin{align*}
      H_1'' &= \left\{x_1\in E_1:|b|_{\{x_1\}\times E_2}|\geq\frac{\epsilon n}{3}\right\} \\
      H_2'' &= \left\{x_2\in E_2:|b|_{E_1\times\{x_2\}}|\geq\frac{\epsilon n}{3}\right\}.
    \end{align*}
    Then if there is some $x_1\in E_1\setminus H_1''$ such that $h_1'(x_1,X_2)\neq 0$, by~(\ref{eq:ydecomp}) we have
    \begin{align*}
      |y|_{\{x_1\}\times E_2}-\evl_{\{x_1\times E_2\}}(h_1')|
      &\leq |b|_{\{x_1\}\times E_2}|+|\evl_{\{x_1\}\times E_2}(h_2')| \\
      &\leq \frac{\epsilon n}{3}+|H_2'| \\
      &< \frac{\epsilon n}{2},
    \end{align*}
    where the third inequality above holds by Claim~\ref{claim:ystruct}. Therefore $\evl_{\{x_1\}\times E_2}(h_1')$ is a valid choice for $c_2\in C_2\setminus\{0\}$ in line~\ref{li:ifc2}. Analogous reasoning implies that if there is some $x_2\in E_2\setminus H_2''$ such that $h_2'(X_1,x_2)\neq 0$, then $\evl_{\{x_1\}\times E_2}(h_2')$ is a valid choice for $c_1\in C_1\setminus\{0\}$ in line~\ref{li:ifc1}.

    Thus in order for the while loop to terminate, we would need to have $h_1'(x_1,X_2)=0$ for every $x_1\in E_1\setminus H_1''$, and $h_2'(X_1,x_2)=0$ for every $x_2\in E_2\setminus H_2''$, which by~(\ref{eq:ydecomp}) implies that
    \begin{align*}
      |y|
      &\leq n(|H_1''|+|H_2''|)+|b| \leq \frac{6}{\epsilon}\cdot|b|+|b| \leq \frac{8}{\epsilon}\cdot|b|.
    \end{align*}
    where the third inequality above holds because by definition $|H_i''|\leq|b|/(\epsilon n/3)$.
  \end{proof}

  \begin{claim}
    \label{claim:ydecrease}
    In every iteration of the while loop in Algorithm~\ref{alg:decoderfinish} where the function does not return, $|y|$ decreases.
  \end{claim}
  \begin{proof}
    We will show that if the {\bf if} statement in line~\ref{li:ifc1} is satisfied, so that line~\ref{li:addc1} executes, then $|y|$ decreases; the proof for {\bf if} statement in line~\ref{li:ifc2} is analogous. By definition $c_1\in C_1\setminus\{0\}$ has $|c|\geq\epsilon n$ as $C_1$ is a code of distance $\geq\epsilon n$, so $c_1$ must agree with $y|_{E_1\times\{x_2\}}$ at $>\epsilon n/2$ points in $\supp(c)$, as by definition $c_1$ disagrees with $y|_{E_1\times\{x_2\}}$ at $<\epsilon n/2$ points in $\supp(c)$. Thus $|y|_{E_1\times\{x_2\}}-c_1|<|y|_{E_1\times\{x_2\}}|$, so $|y-c_1\otimes\1_{x_2}|<|y|$, as desired. 
  \end{proof}

  Thus Claim~\ref{claim:ydecrease} implies that the while loop in Algorithm~\ref{alg:decoderfinish} terminates after at most $n^2$ iterations, while Claim~\ref{claim:ysmall} implies that the final output $y'$ of the algorithm will satisfy $|y|\leq 8|b|/\epsilon$, as desired.
\end{proof}

\section{Decoding Products of Quantum Reed-Solomon Codes}
\label{sec:qdec}
In this section, we apply our decoder in Section~\ref{sec:dualtensordec} for dual tensor products of Reed-Solomon codes to construct efficient decoders for both homological and subsystem products of quantum Reed-Solomon codes. In particular, decoders for both types of products will follow from the following proposition.

\begin{proposition}
  \label{prop:qdectech}
  For every $0<\epsilon,\delta'\leq 1$, there exists a positive real number $\delta=\delta(\epsilon,\delta')\leq\delta'$ and an algorithm satisfying the following. The algorithm takes as input a pair $(Q^i=(Q^i_X,Q^i_Z))_{i\in[2]}$ of length-$n$ quantum Reed-Solomon codes such that $\dim(Q^1_Z),\;\dim({Q^1_X}^\perp)+\dim(Q^2_Z)\leq(1-\epsilon)n$, and an element $c\in\bF_q^{n\times n}$ such that $|c-\tilde{c}|\leq\delta n^2$ for some $\tilde{c}\in Q_Z':=(Q^1_Z\otimes Q^2_Z)+(Q^1_X\otimes Q^2_X)^\perp$. The algorithm then outputs some $c'\in Q_Z'$ such that $|c-c'|\leq\delta' n^2$. The algorithm runs in time $\poly(n,q)$ (independent of $\epsilon,\delta'$).
\end{proposition}

Proposition~\ref{prop:qdectech} follows by applying our dual tensor decoder in Theorem~\ref{thm:dualtensordec}, and then applying an additional subroutine similar to Algorithm~\ref{alg:decoderclose}. Hence the proof is similar to that of Lemma~\ref{lem:algdeg}; we provide the details in Appendix~\ref{sec:omitdec}.

In the remainder of this section, for simplicity we restrict attention to Reed-Solomon codes whose evaluation set is the entire field.

\subsection{Single-Sector Homological Products}
In this section, we present a polynomial-time decoder for the CSS codes in Corollary~\ref{cor:sspeRS}, which are constructed as the homological product of a pair of single-sector chain complexes associated to quantum Reed-Solomon codes. To begin, we define the decoding problem for CSS codes.

\begin{definition}
  \label{def:CSSdec}
  For a length-$n$ quantum CSS code $Q=(Q_X,Q_Z)$ over $\bF_q$, a \textbf{$\delta$-decoder for $Q$} is an algorithm that takes as input a pair of elements $c_X,c_Z\in\bF_q^n$ (along with a description of the codes $Q_X,Q_Z$) such that there exist $\tilde{c}_X\in Q_X,\;\tilde{c}_Z\in Q_Z$ with $|c_X-\tilde{c}_X|,|c_Z-\tilde{c}_Z|\leq\delta n$, and then outputs (representative elements of) the cosets $\tilde{c}_X+Q_Z^\perp,\;\tilde{c}_Z+Q_X^\perp$.
\end{definition}

In words, Definition~\ref{def:CSSdec} says that the decoder takes as input corrupted codewords $c_X,c_Z$ of $Q_X,Q_Z$ respectively (with corruption weight $\leq\delta n$), and returns cosets in $Q_X/Q_Z^\perp,Q_Z/Q_X^\perp$ that are close to these respective corrupted codewords.

\begin{remark}
  \label{remark:CSSsyndromedec}
   Often the decoding problem for CSS codes is instead defined to have input given by the \textit{syndromes} $s_X=H_Xc_X,\;s_Z=H_Zc_Z$ of the corrupted codewords $c_X,c_Z$, where $H_X,H_Z$ are the parity-check matrices for $Q_X,Q_Z$ respectively. The output is then defined to be a pair of (representative elements of) cosets $b_X+Q_Z^\perp,b_Z+Q_X^\perp$ of low-weight corrections $b_X,b_Z\in\bF_q^n$ such that $H_Xb_X=s_X,\;H_Zb_Z=s_Z$. However, this formulation reduces to our problem stated in Definition~\ref{def:CSSdec}. Specifically, given $s_X,s_Z$, we may first run Gaussian elimination to find any $c_X',c_Z'\in\bF_q^n$ with syndromes $H_Xc_X'=s_X,\;H_Zc_Z'=s_Z$. If our decoder from Definition~\ref{def:CSSdec} outputs $\tilde{c}_X',\tilde{c}_Z'$, then $c_X'-\tilde{c}_X'+Q_Z^\perp,\;c_Z'-\tilde{c}_Z'+Q_X^\perp$ provide the desired output cosets in the syndrome formulation.
\end{remark}

\begin{corollary}
  For every $\epsilon>0$, there exists $\delta=\delta(\epsilon)>0$ such that there is a polynomial-time $\delta$-decoder for the codes in Corollary~\ref{cor:sspeRS}.

  Formally, the $\delta$-decoder takes as input a prime power $n=q$, integers\footnote{To simplify notation, here the value of $\epsilon$ is double that in Corollary~\ref{cor:sspeRS}; for instance, in the corollary we only had $k_1,k_2\leq(1-\epsilon/2)n$.} $k_1,k_2\leq(1-\epsilon)n$ with $|k_1-k_2|\geq\epsilon n$, and a pair of elements $c_X,c_Z\in\bF_q^{n\times n}$. Let $\cC_i$ be the single-sector complex obtained from $Q^i:=(Q^i_X=RS(n,k_i),\;Q^i_Z=RS(n,k_i))$ for $i\in[2]$ via Lemma~\ref{lem:codetocomplex}, and let $Q=(Q_X,Q_Z)$ be the length-$n^2$ quantum CSS code associated to the homological product $\cC_1\otimes\cC_2$. If there exist $\tilde{c}_X\in Q_X,\;\tilde{c}_Z\in Q_Z$ with $|c_X-\tilde{c}_X|,|c_Z-\tilde{c}_Z|\leq\delta n^2$, then the $\delta$-decoder outputs (representative elements of) the cosets $\tilde{c}_X+Q_Z^\perp,\;\tilde{c}_Z+Q_X^\perp$. The $\delta$-decoder runs in time $\poly(n,q)$ (independent of $\epsilon$).
\end{corollary}
\begin{proof}
  Assume without loss of generality that $k_1\geq k_2$, so that $\dim({Q^1_X}^\perp)+\dim(Q^2_Z)=(n-k_1)+k_2\leq 1-\epsilon$; the proof for the $k_2\geq k_1$ case is analogous.
  
  We will prove the $Z$-decoding statement (i.e.~for $c_Z$); the proof of the $X$-decoding statement (i.e.~for $c_X$) is analogous. Define $Q_Z'=(Q^1_Z\otimes Q^2_Z)+(Q^1_X\otimes Q^2_X)^\perp$ as in Proposition~\ref{prop:qdectech}, and define $Q_X''=Q^1_X\otimes Q^2_X$. By Proposition~\ref{prop:sskunneth}, we have
  \begin{equation}
    \label{eq:qdecapplykunneth}
    Q_Z\subseteq Q_Z'; \hspace{1em} Q_X^\perp\subseteq{Q_X''}^\perp; \hspace{1em} Q_Z/Q_X^\perp \cong (Q^1_Z/{Q^1_X}^\perp)\otimes(Q^2_Z/{Q^2_X}^\perp) \cong Q_Z'/{Q_X''}^\perp.
  \end{equation}

  Define $\delta'=\delta'(\epsilon)=\epsilon\cdot\rho(\epsilon)/4$ for $\rho(\epsilon)$ as defined in Theorem~\ref{thm:pe2RS}, so that by Lemma~\ref{lem:sspedis}, every element of $Q_Z'\setminus{Q_X''}^\perp$ has weight $\geq 4\delta' n^2$.
  Then defining $\delta=\delta(\epsilon,\delta'(\epsilon))$ as in Proposition~\ref{prop:qdectech}, we can run the $\poly(n,q)$-time algorithm in Proposition~\ref{prop:qdectech} to obtain some $c'\in Q_Z'$ such that $|c_Z-c'|\leq\delta' n^2$. Therefore $|c'-\tilde{c}_Z|\leq|c'-c_Z|+|c_Z-\tilde{c}_Z|\leq \delta' n^2+\delta n^2\leq 2\delta' n^2$, so $c'-\tilde{c}_Z\in Q_Z'$ must lie inside ${Q_X''}^\perp$, or equivalently, $c'+{Q_X''}^\perp=\tilde{c}_Z+{Q_X''}^\perp\in Q_Z'/{Q_X''}^\perp$.

  Now by~(\ref{eq:qdecapplykunneth}), this coset in $Q_Z'/{Q_X''}^\perp$ contains a unique coset in $Q_Z/Q_X^\perp$, which must therefore equal $\tilde{c}_Z+Q_X^\perp$, as $\tilde{c}_Z\in Q_Z$. Thus our decoder may simply find the unque coset in $Q_Z/Q_X^\perp$ that lies inside $c'+{Q_X''}^\perp$, and then output any representative element. This computation simply involves solving a linear system of equations, so it can be performed in time $\poly(n,q)$, as desired.
\end{proof}

\subsection{Subsystem Products}
In this section, we present a polynomial-time decoder for the subsystem codes in Corollary~\ref{cor:subpeRS}, which are constructed as the subsystem product of a pair of quantum Reed-Solomon codes. To begin, we define the decoding problem of subsystem codes.

\begin{definition}
  \label{def:subdec}
  For a length-$n$ quantum subsystem code $Q=(Q_X,Q_Z)$ over $\bF_q$, a \textbf{$\delta$-decoder for $Q$} is an algorithm that takes as input a pair of elements $c_X,c_Z\in\bF_q^n$ (along with a description of the codes $Q_X,Q_Z$) such that there exist $\tilde{c}_X\in Q_X':=Q_X+Q_Z^\perp,\;\tilde{c}_Z\in Q_Z':=Q_Z+Q_X^\perp$ with $|c_X-\tilde{c}_X|,|c_Z-\tilde{c}_Z|\leq\delta n$, and then outputs (representative elements of) the cosets $\tilde{c}_X+Q_Z^\perp,\;\tilde{c}_Z+Q_X^\perp$.
\end{definition}

In words, Definition~\ref{def:subdec} says that the decoder takes as input corrupted codewords $c_X,c_Z$ of $Q_X',Q_Z'$ respectively (with corruption weight $\leq\delta n$), and returns cosets in $Q_X'/Q_Z^\perp,Q_Z'/Q_X^\perp$ that are close to these respective corrupted codewords.

\begin{remark}
  \label{remark:subsyndromedec}
  Similarly as in the non-subsystem case described in Remark~\ref{remark:CSSsyndromedec}, our notion of a subsystem code decoder in Definition~\ref{def:subdec} can be used to decode when the input instead consists of syndromes $s_X=H_Xc_X,\;s_Z=H_Zc_Z$, where $H_X,H_Z$ are parity-check matrices for $Q_X,Q_Z$ respectively, and the desired output consists of (representative elements of) cosets $b_X+Q_Z^\perp,\;b_Z+Q_X^\perp$ for low-weight $b_X,b_Z\in\bF_q^n$ such that $c_X-b_X\in Q_X',\;c_Z-b_Z\in Q_Z'$. Specifically, given such syndromes $s_X,s_Z$, we may first run Gaussian elimination to find any $c_X',c_Z'\in\bF_q^n$ with syndromes $H_Xc_X'=s_X,\; H_Zc_Z'=s_Z$. If our decoder from Definition~\ref{def:subdec} outputs $\tilde{c}_X',\tilde{c}_Z'$, then $c_X'-\tilde{c}_X'+Q_Z^\perp,\;c_Z'-\tilde{c}_Z'+Q_X^\perp$ provide the desired output cosets in the syndrome formulation.
\end{remark}

\begin{corollary}
  \label{cor:subRSdec}
  For every $\epsilon>0$, there exists $\delta=\delta(\epsilon)>0$ such that there is a polynomial-time $\delta$-decoder for the codes in Corollary~\ref{cor:subpeRS}.

  Formally, the $\delta$-decoder takes as input a prime power $n=q$, integers $k^1_X,k^1_Z,k^2_X,k^2_Z\leq(1-\epsilon)n$ such that $k^1_X+k^1_Z,\;k^2_X+k^2_Z\geq n$ and $(n-k^1_Z)+k^2_X,\;(n-k^1_X)+k^2_Z\leq(1-\epsilon)n$, and a pair of elements $c_X,c_Z\in\bF_q^{n\times n}$. For $i\in[2]$ let $Q^i:=(Q^i_X=\text{RS}(n,k^i_X),\;Q^i_Z=\text{RS}(n,k^i_Z))$, and let $Q=(Q_X,Q_Z)=Q^1\otimes Q^2$ be the length-$n^2$ quantum subsystem product code. If there exist $\tilde{c}_X\in Q_X':=Q_X+Q_Z^\perp,\;\tilde{c}_Z\in Q_Z':=Q_Z+Q_X^\perp$ with $|c_X-\tilde{c}_X|,|c_Z-\tilde{c}_Z|\leq\delta n^2$, then the $\delta$-decoder outputs (representative elements of) the cosets $\tilde{c}_X+Q_Z^\perp,\;\tilde{c}_Z+Q_X^\perp$. The $\delta$-decoder runs in time $\poly(n,q)$ (independent of $\epsilon$).
\end{corollary}
\begin{proof}
  We will prove the $Z$-decoding statement (i.e.~for $c_Z$); the proof of the $X$-decoding statement (i.e.~for $c_X$) is analogous. Define $\delta'=\delta'(\epsilon)=\epsilon\cdot\rho(\epsilon)/4$ for $\rho(\epsilon)$ as defined in Theorem~\ref{thm:pe2RS}, so that by Lemma~\ref{lem:sspedis}, every element of $Q_Z'\setminus Q_X^\perp$ has weight $\geq 4\delta' n^2$.
  Then defining $\delta=\delta(\epsilon,\delta'(\epsilon))$ as in Proposition~\ref{prop:qdectech}, we can run the $\poly(n,q)$-time algorithm in Proposition~\ref{prop:qdectech} to obtain some $c'\in Q_Z'$ such that $|c_Z-c'|\leq\delta' n^2$. Therefore $|c'-\tilde{c}_Z|\leq|c'-c_Z|+|c_Z-\tilde{c}_Z|\leq \delta' n^2+\delta n^2\leq 2\delta' n^2$, so $c'-\tilde{c}_Z\in Q_Z'$ must lie inside $Q_X^\perp$, or equivalently, $c'+Q_X^\perp=\tilde{c}_Z+Q_X^\perp\in Q_Z'/Q_X^\perp$. Thus our decoder may simply output $c'$.
\end{proof}

\section{Transversal $C^{r-1}Z$ Gates on Subsystem Products}
\label{sec:transversal}
In this section, we show how to perform transversal $C^{r-1}Z$ (and $U^r$) gates on subsystem product codes. In Section~\ref{sec:transgen} below, we present a general condition under which a subsystem product supports such transversal gates. In Section~\ref{sec:transRS}, we then instantiate this general construction with the product of two quantum Reed-Solomon codes, which in particular yields $[[N,\Theta(N),\Theta(N)]]_N$ subsystem products of locality $O(\sqrt{N})$ supporting transversal $CCZ$. In Section~\ref{sec:tripleprod} below, we subsequently instantiate the general construction with the product of three quantum codes arising from randomly punctured tensor Reed-Solomon codes, which yields $[[N,N^{1-o(1)},N^{1-o(1)}]]_q$ subsystem products of locality $O(N^{1/3})$ supporting transversal $CCZ$ (albeit over an exponentially large alphabet $q$).

\subsection{General Condition for Transversal $C^{r-1}Z$ Gate}
\label{sec:transgen}
The theorem below provides a general condition under which a subsystem product code (see Definition~\ref{def:subhomprod}) supports a transversal $C^{r-1}Z$ (and $U^r$) gate.

\begin{theorem}
  \label{thm:transgen}
  For integers $t,r\geq 2$, let $(Q^i=(Q_X^i,Q_Z^i))_{i\in[t]}$ be (non-subsystem) CSS codes over some field $\bF_q$, and let $Q=(Q_X,Q_Z)=\bigotimes_{i\in[t]}Q^i$ be the subsystem product. For each $i\in[t]$, fix a subspace $L_i\subseteq Q_Z^i$ with $L_i\cap{Q_X^i}^\perp=\{0\}$, and let
  \begin{align*}
    L &:= \bigotimes_{i\in[t]}L_i \\
    S &:= Q_Z\cap Q_X^\perp 
  \end{align*}
  satisfy
  \begin{equation}
    \label{eq:multprop}
    L^{*r} \cap (S*(L+S)^{*r-1}) = \{0\}.
  \end{equation}
  Then there exists an encoding map $\Enc:\bF_q^{\dim(L)}\xrightarrow{\sim}(L+S)/S\subseteq Q_Z/S$ for which $(Q,\Enc)$ supports a transversal $C^{r-1}Z$ (and $U^r$) gate (on $\dim(L)$ logical qudits).
\end{theorem}
\begin{proof}
  For each $i\in[t]$, if $Q^i$ has length $n_i$, let $A_i\subseteq[n_i]$ be a set of code components of size $|A_i|=\dim(L_i)$ such that $L_i|_{A_i}=\bF_q^{A_i}$. That is, the restriction map to components in $A_i$ induces an isomorphism $L_i\xrightarrow{\sim}\bF_q^{A_i}$. Define $\Enc_i:\bF_q^{A_i}\xrightarrow{\sim}L_i$ to be the inverse isomorphism, so that $\Enc_i(z_i)|_{A_i}=z_i$ for every $z_i\in\bF_q^{A_i}$. Then letting $A=A_1\times\cdots\times A_t$, define $\Enc:\bF_q^A\xrightarrow{\sim}(L+S)/S$ by
  \begin{equation*}
    \Enc(z) = (\Enc_1\otimes\cdots\otimes\Enc_t)(z)+S
  \end{equation*}
  for $z\in\bF_q^A$. Note that $\Enc$ is indeed an isomorphism because each $L_i\cap Q_X^i=\{0\}$, and hence $L\cap S=\{0\}$.

  It remains to prove the transversal $C^{r-1}Z$ property in Definition~\ref{def:transversal}. We first define the desired coefficients vector $a\in\bF_q^n$ for $n=n_1\cdots n_t$. For this purpose, by~(\ref{eq:multprop}), it holds that
  \begin{equation*}
    (L+S)^{*r} = L^{*r}+S*(L+S)^{*r-1} = L^{*r}\oplus S*(L+S)^{*r-1}.
  \end{equation*}
  Hence there is a well-defined linear projection map $\eta:(L+S)^{*r}\rightarrow L^{*r}$ with kernel $S*(L+S)^{*r-1}$, which we can extend (in an arbitrary way) to a linear map $\eta:\bF_q^n\rightarrow L^{*r}$, as $(L+S)^{*r}\subseteq\bF_q^n$. Then we define a linear functional $a:\bF_q^n\rightarrow\bF_q$ by $a(z')=\1_A\cdot\eta(z')=\sum_{j=(j_1,\dots,j_t)\in A}\eta(z')_j$, which we may view as a vector $a\in\bF_q^n$ where $a(z')=a\cdot z'$.

  Now for every $(z^h\in\bF_q^A)_{h\in[r]}$ and every $({z^h}'\in\Enc(z^h))_{h\in[r]}$, writing ${z^h}'=y^h+s^h$ for $y^h=(\Enc_1\otimes\cdots\otimes\Enc_t)z^h$ and $s^h\in S$, then it follows that
  \begin{align*}
    a \cdot ({z^1}'*\cdots*{z^r}')
    &= \1_A\cdot\eta({z^1}'*\cdots*{z^r}') \\
    &= \1_A\cdot\eta((y^1+s^1)*\cdots*(y^r+s^r)) \\
    &= \1_A\cdot(y^1*\cdots*y^r) \\
    &= \1_A\cdot(z^1*\cdots*z^r),
  \end{align*}
  as desired. Note that the first equality above holds by the definition of $a$, the second equality holds by the definition of the ${z^h}'$, the third equality holds by the definition of $\eta$, and the fourth equality holds because $y^h|_A=z^h$ by the definition of the $\Enc_i$'s.
\end{proof}


\subsection{Product of Two Reed-Solomon Codes}
\label{sec:transRS}
In this section, we instantiate the general construction in Theorem~\ref{thm:transgen} with the subsystem product of a pair of carefully chosen quantum Reed-Solomon codes. Below, we use the notation for polynomial evaluation codes defined in Definition~\ref{def:evlnotation}.

\begin{theorem}
  \label{thm:transRS}
  Fix any integer $r\geq 2$, and let $\epsilon=1/4r$. Let $\bF_q$ be a finite field with $q\geq 4r^2$, and for $i\in[2]$ let $Q^i=(Q^i_X=\textrm{RS}(q,k^i_X),\; Q^i_Z=\textrm{RS}(q,k^i_Z))$ with
  \begin{align*}
    k^1_X = q-\lfloor\epsilon q\rfloor, &\hspace{1em} k^1_Z = q-\lfloor\epsilon q\rfloor \\
    k^2_X = q-2\lfloor\epsilon q\rfloor, &\hspace{1em} k^2_Z = \lfloor q/r\rfloor.
  \end{align*}
  Then there exists some $\Delta=\Delta(r)>0$ such that the subsystem product $Q:=Q^1\otimes Q^2$ is a
  \begin{equation*}
    [[q^2,\; (k^1_X+k^1_Z-q)(k^2_X+k^2_Z-q),\; \Delta\cdot q^2]]_q
  \end{equation*}
  subsystem code of locality $2q$.
  
  Furthermore, for $i\in[2]$, define $L_i\subseteq Q^i_Z$ by $L_1=L_2=\evl(\bF_q[X]^{[\underline{\ell},\overline{\ell})})$ with
  \begin{align*}
    \underline{\ell} = \lceil q/r-q/2r^2\rceil, &\hspace{1em} \overline{\ell} = k^2_Z = \lfloor q/r\rfloor.
  \end{align*}
  Then defining $L$ and $S$ as in Theorem~\ref{thm:transgen}, there exists an encoding map $\Enc:\bF_q^{\dim(L)}\xrightarrow{\sim}(L+S)/S$ for which $(Q,\Enc)$ supports a transversal $C^{r-1}Z$ (and $U^r$) gate on $\dim(L) = (\overline{\ell}-\underline{\ell})^2$ logical qudits.
\end{theorem}
\begin{proof}
  The code parameters of $Q$ follow directly from Corollary~\ref{cor:subpeRS}. Thus to prove the theorem, it suffices to show that the subspaces $L,S$ satisfy the conditions in Theorem~\ref{thm:transgen}. For $i\in[2]$, by the definition of $L_i$ and $Q^i_X$ we have
  \begin{equation*}
    L_i\cap{Q^i_X}^\perp = \evl(\bF_q[X]^{[\underline{\ell},\overline{\ell})})\cap\evl(\bF_q[X]^{[0,q-k^i_X)}) = \evl(\bF_q[X]^{[\underline{\ell},\overline{\ell})\cap[0,q-k^i_X)}) = \{0\}
  \end{equation*}
  
  Therefore it only remains to show that~(\ref{eq:multprop}) holds. For this purpose, we may write
  \begin{align*}
    L
    &= \evl(\bF_q[X_1,X_2]^M)
  \end{align*}
  for
  \begin{align*}
    M &= [\underline{\ell},\overline{\ell})^2 \subseteq [q/r-q/2r^2,q/r)^2,
  \end{align*}
  so
  \begin{align*}
    L^{*r}
    &\subseteq \evl(\bF_q[X_1,X_2]^{rM}).
  \end{align*}
  Meanwhile,
  \begin{align*}
    S
    &= \evl(\bF_q[X_1,X_2]^T)
  \end{align*}
  for
  \begin{align*}
    T
    &= [0,k^1_Z)\times[0,q-k^2_X)\cup[0,q-k^1_X)\times[0,k^2_Z) \\
    &\subseteq [0,q)\times[0,2\epsilon q]\cup[0,\epsilon q]\times[0,q/r).
  \end{align*}
  Therefore
  \begin{align*}
    L^{*r}\cap(S*(L+S)^{*r-1}) \subseteq \evl(\bF_q[X_1,X_2]^A)
  \end{align*}
  for
  \begin{align*}
    A
    &= rM \cap (T+(M\cup T)^{+(r-1)}).
  \end{align*}

  Our goal is to show that $A=\emptyset$, as it then follows that $L^{*r}\cap(S*(L+S)^{*r-1})\subseteq\evl(\bF_q[X_1,X_2]^\emptyset)=\{0\}$, as desired. For this purpose, assume for contradiction that there is some point $a=(a_1,a_2)\in A$. Because $a\in rM$, we have $a_1,a_2\geq q-q/2r$. Because $a\in T+(M\cup T)^{+(r-1)}$, we may write $a=b^1+\cdots+b^r$ for some $b^1\in T$ and $b^2,\dots,b^r\in M\cup T$. As the second coordinate of every point in $M\cup T$ is $<q/r$, we must have each $b^i_2<q/r$ but $b^1_2+\cdots+b^r_2=a_2\geq q-q/2r$, and hence each $b^i>q/r-q/2r=q/2r$. But the only points in $M\cup T$ with second coordinate $>q/2r$ are those points in the set $P:=M\cup[0,\epsilon q]\times(q/2r,q/r)$. Hence each $b^i\in P$. But we also have that $b^1\in T$, so $b^1 \in P\cap T = [0,\epsilon q]\times(q/2r,q/r)$. Therefore $b^1\leq \epsilon q=q/4r$, and for $2\leq i\leq r$ we have $b^i_1<q/r$ because every point in $P$ has first coordinate $<q/r$. Thus $b^1_1+\cdots+b^r_1<q-3q/4r$, which contradicts the fact that $b^1_1+\cdots+b^r_1=a_1\geq q-q/2r$. Thus the assumption that $A$ is nonempty was false, as desired.
\end{proof}

As described in Section~\ref{sec:subprodccz}, the proof of Theorem~\ref{thm:transRS} required somewhat delicate parameter choices and analysis to ensure that the multiplication propety~(\ref{eq:multprop}) holds. The difficulty here was in part due to the fact that Corollary~\ref{cor:subpeRS} places certain constraints on the dimensions of the underlying Reed-Solomon codes; such constraints are ultimately needed because Reed-Solomon codes are only product-expanding when the sum of the rates is $<1$ (see Theorem~\ref{thm:pe2RS}). Similar, but even more delicate, challenges arise in establishing the multiplication property~(\ref{eq:multprop}) in Section~\ref{sec:tripleprod} below for the product-expanding codes from Section~\ref{sec:perp}.

\section{Product-Expansion of Randomly Punctured Tensor Reed-Solomon Codes}
\label{sec:perp}

In this section, we prove Theorem~\ref{thm:peptRS} below, which shows that randomly punctured tensor products of Reed-Solomon codes are product-expanding; in Section~\ref{sec:tripleprod} below, we will then apply the results in Section~\ref{sec:transversal} to show how to obtain a subsystem product of three such codes that supports a transversal $CCZ$ gate. Our proof builds on the techniques of \cite{kalachev_maximally_2025}, who showed that uniformly random classical codes are product-expanding (Theorem~\ref{thm:petrand}).

Note that in order to ultimately apply Theorem~\ref{thm:subpe} to the codes described below, we consider the product-expansion of a $t$-tuple of codes consisting of $t-1$ randomly punctured tensor products of Reed-Solomon codes, and one dual of such a code.

\begin{theorem}
  \label{thm:peptRS}
  Let $\bF_q$ be a finite field, and let $m,u,t\geq 1$ and $1\leq k_1,\dots,k_t<m$ be integers such that $q\geq m$. Let $n=m^u$. Sample $t$ uniformly random subsets $E_1,\dots,E_t\subseteq\bF_q^u$, each of size $|E_i|=n$. Then with probability $\geq 1-n^{t+1}2^{n^t+2}/q$, the $t$-tuple of codes
  \begin{equation}
    \label{eq:ptRS}
    \left(\evl_{E_1}(\bF_q[X_1,\dots,X_u]^{[0,k_1)^u}), \dots, \evl_{E_{t-1}}(\bF_q[X_1,\dots,X_u]^{[0,k_{t-1})^u}), \evl_{E_t}(\bF_q[X_1,\dots,X_u]^{[0,k_t)^u})^\perp\right)
  \end{equation}
  has product-expansion at least
  \begin{equation*}
    \left(\frac{1}{um}\right)^{O(t^3)} \cdot \left(\frac{k_t}{m}\right)^{ut} \cdot \prod_{i=1}^{t-1}\left(1-\frac{k_i}{m}\right)^{4ut^2}.
  \end{equation*}
\end{theorem}

The remainder of this section is dedicated to proving Theorem~\ref{thm:peptRS}. We begin with some necessary preliminary notions and prior results in Section~\ref{sec:peptprelim}, which go beyond those presented in Section~\ref{sec:prelim}; we then present our proof in Section~\ref{sec:peptproof}.

\subsection{Preliminaries}
\label{sec:peptprelim}
This section presents definitions and prior results that we will need to prove Theorem~\ref{thm:peptRS}.

\subsubsection{Classical Codes}
We begin with some necessary preliminaries on classical codes. Below, recall the definition of an MDS code from Definition~\ref{def:MDS}.

\begin{lemma}[Well-known]
  The following are equivalent conditions for a $k$-dimensional classical code $C\subseteq\bF_q^n$:
  \begin{enumerate}
  \item $C$ is MDS,
  \item $C^\perp$ is MDS,
  \item $C$ is a $[n,k,n-k+1]$ code.
  \end{enumerate}
\end{lemma}

The following property of locally testable codes was shown in \cite[Lemma~3.4]{kalachev_maximally_2025}.

\begin{lemma}[\cite{kalachev_maximally_2025}]
  \label{lem:colLTC}
  Let $C$ be a classical locally testable code of relative distance $\delta$, locality $w$, and soundness $\rho$, with parity-check matrix $H\in\bF_q^{m\times n}$. Then for every subset $S\subseteq[m]$, there exists a subset $E\subseteq[n]$ with $|E|\leq(6n/\delta\rho m)|S|$ such that for every $s\in\im H$ with $\supp(s)\subseteq S$, there exists $e\in\bF_q^n$ with $\supp(e)\subseteq R$ such that $He=s$.
\end{lemma}

\cite{viderman_combination_2015} showed the following result on local testability of tensor codes:

\begin{theorem}[\cite{viderman_combination_2015}]
  \label{thm:tensorltc}
  Let $C\subseteq\bF_q^n$ be a classical code of relative distance $\delta$. Then for every integer $t\geq 3$, the tensor code $C^{\otimes t}$ is locally testable of locality $n$ and soundness $\rho\geq \delta^{3t}/(tn)^{O(1)}$.
\end{theorem}

\begin{remark}
  In \cite{viderman_combination_2015}, the soundness lower bound is stated as $\delta^{3t}/t^{O(1)}$, whereas we have only stated the looser bound of $\delta^{3t}/(tn)^{O(1)}$. The difference arises because \cite{viderman_combination_2015} considers a definition of local testability in which multiple parity-checks on the same code components count as a single query; soundness under this definition may be larger than under our Definition~\ref{def:MDS} by a multiplicative factor growing polynomially in the locality.
\end{remark}

The local checks (i.e.~rows of the parity-check matrix) in Theorem~\ref{thm:tensorltc} simply check that for every $i\in[t]$, the restriction of the codeword (which is a $t$-dimensional tensor) to each direction-$i$ column lies in $C$. Therefore in particular, this parity-check matrix has $n^t$ columns and $t\cdot(n-\dim(C))\cdot n^{t-1}$ rows.

We will also use the well-known Schwartz-Zippel lemma:

\begin{lemma}[Schwartz-Zippel]
  \label{lem:sz}
  For a finite field $\bF_q$, an integer $t\in\bN$, and a nonzero polynomial $f(X_1,\dots,X_t)\in\bF_q[X_1,\dots,X_t]$ of total degree $\leq d$, then
  \begin{equation*}
    \Pr_{(x_1,\dots,x_t)\sim\Unif(\bF_q^t)}[f(x_1,\dots,x_t)=0] \leq \frac{d}{q}.
  \end{equation*}
\end{lemma}

\subsubsection{Product-Expansion}
We now present necessary preliminaries on product-expansion; many of the results here were shown in \cite{kalachev_maximally_2025}.

The following result on product-expansion of LTCs is shown in \cite[Lemma~3.6]{kalachev_maximally_2025}, though we have adapted the statement to our setting. Note that the statement in \cite{kalachev_maximally_2025} requires $t$ classical LTCs $C_1,\dots,C_t$, but the ultimate bound they prove does not depend on the testability of $C_t$ (and they note this property in a remark). This property that only $t-1$ of the codes need to be LTCs will be crucial in our application of this lemma.


\begin{lemma}[\cite{kalachev_maximally_2025}]
  \label{lem:peLTC}
  For $t\in\bN$, let $C_1,\dots,C_t\subseteq\bF_q^n$ be codes such that $C_i$ has relative distance $\delta_i$ for $i\in[t]$, and $C_i$ is locally testable of locality $w_i$ and soundness $\rho_i$ with parity-check matrix $H_i\in\bF_q^{m_i\times n}$ for $i\in[t-1]$. Then $(C_1,\dots,C_t)$ has product-expansion at least $\rho(t) = \rho(t,(\delta_i,w_i,\rho_i,n/m_i)_{i\in[t]})$ defined recursively by $\rho(1)=\delta_1$ and
  \begin{align*}
    \rho(t)
    &= \frac{\delta_t}{3}\cdot\left(\prod_{i=1}^{t-1}\frac{\delta_i\rho_im_i}{6w_in}\right)\cdot\rho(t-1) \hspace{1em}\forall\; t\geq 2.
  \end{align*}
\end{lemma}

The following basic lemma shows that product-expansion is preserved when the input codes are multiplied by arbitrary fixed coefficients. Below, recall from Section~\ref{sec:notation} that we use `$*$' to denote pointwise multiplication of vectors; we extend this notation to act on codes, so that for $\gamma\in\bF_q^n$ and $C\subseteq\bF_q^n$, then $\gamma*C=\{\gamma*c:c\in C\}\subseteq\bF_q^n$.

\begin{lemma}
  \label{lem:coeffsame}
  If the tuple $(C_i\subseteq\bF_q^n)_{i\in[t]}$ has product-expansion $\rho$, then for every tuple of vectors $(\gamma_i\in(\bF_q^*)^n)_{i\in[t]}$, the tuple $(\gamma_i*C_i\subseteq\bF_q^n)_{i\in[t]}$ also has product-expansion $\rho$.
\end{lemma}
\begin{proof}
  Let $\gamma^{-1}_i=(\gamma_{i,j}^{-1})_{j\in[n]}\in\bF_q^n$ denote the pointwise inverse, and let $\gamma=\gamma_1\otimes\cdots\otimes\gamma_t$ and $\gamma^{-1}=\gamma^{-1}_1\otimes\cdots\otimes\gamma^{-1}_t\in\bF_q^{n^t}$. For every $c=c_1+\cdots+c_t\in C_1\boxplus\cdots\boxplus C_t$ with each $c_i\in C^{(i)}$, then by definition $\gamma*c=\gamma*c_1+\cdots+\gamma*c_t\in(\gamma_1*C_1)\boxplus\cdots\boxplus(\gamma_t*C_t)$ with each $\gamma*c_i\in \gamma_i*C_i$. Similarly, for every $c'=c_1'+\cdots+c_t'\in(\gamma_1*C_1)\boxplus\cdots\boxplus(\gamma_t*C_t)$ with each $c_i'\in\gamma_i*C_i$, then $\gamma^{-1}*c'=\gamma^{-1}*c_1'+\cdots+\gamma^{-1}*c_t'\in C_1\boxplus\cdots\boxplus C_t$ with each $\gamma^{-1}*c_i'\in C^{(i)}$. Because all entries of each $\gamma_i$ are nonzero, multiplying component-wise by $\gamma_i$ or $\gamma_i^{-1}$ (and hence also by $\gamma$ or $\gamma^{-1}$) preserves Hamming weight. Therefore $(C_i)_{i\in[t]}$ and $(\gamma_i*C_i)_{i\in[t]}$ must have the same product-expansion.
\end{proof}

The following lemma shows that if a tuple of codes has product-expansion $\rho$, then every sub-tuple of this tuple has product-expansion $\geq\rho$.

\begin{lemma}[\cite{kalachev_two-sided_2023}]
  \label{lem:pesubtuple}
  For $t,n\in\bN$, if a tuple of codes $(C_i\subsetneq\bF_q^n)_{i\in[t]}$ has product-expansion at least $\rho>0$, then for every $I\subseteq[t]$, the tuple $(C_i)_{i\in I}$ also has product-expansion at least $\rho$.
\end{lemma}

We now present the notion of \textit{inner-generated} sets and \textit{$\epsilon$-closed sets} for dual tensor codes, which were introduced by \cite{kalachev_maximally_2025}. 

\begin{definition}[\cite{kalachev_maximally_2025}]
  \label{def:ig}
  For classical codes $C_1,\dots,C_t\subseteq\bF_q^n$, we say that a set $A\subseteq[n]^t$ is \textbf{inner generated} for the dual tensor code $C_1\boxplus\cdots\boxplus C_t$ if every $c\in C_1\boxplus\cdots\boxplus C_t$ with $\supp(c)\subseteq A$ can be expressed in the form $c=c_1+\cdots+c_t$ with each $c_i\in C^{(i)}$ supported inside direction-$i$ columns that are entirely contained within $A$.

  For $\epsilon>0$, the \textbf{$\epsilon$-closure of $A$}, denoted $[A]_\epsilon\subseteq[n]^t$, is the minimal set such that for every $i\in[t]$, every direction-$i$ column is either entirely contained inside $A$, or intersects $A$ at $<\epsilon n$ components. We say that $A$ is \textbf{$\epsilon$-closed} if $[A]_\epsilon=A$.
\end{definition}

\cite{kalachev_maximally_2025} showed the following results regarding product-expansion and inner-generated sets, which we will apply in Section~\ref{sec:perp} to show that randomly punctured tensor products of Reed-Solomon codes are product-expanding.

\begin{lemma}[\cite{kalachev_maximally_2025}]
  \label{lem:petoig}
  Let $C_1,\dots,C_t\subsetneq\bF_q^n$ have product-expansion $\rho$. Then every $\rho$-closed subset $A\subseteq[n]^t$ is inner-generated for $C_1\boxplus\cdots\boxplus C_t$.
\end{lemma}

\begin{lemma}[\cite{kalachev_maximally_2025}]
  For every $t,n\in\bN$, $\epsilon>0$, and $A\subseteq[n]^t$, it holds that
  \begin{equation*}
    |[A]_\epsilon| \leq \left(\frac{2^t+1}{\epsilon}\right)^t \cdot |A|.
  \end{equation*}
\end{lemma}

\begin{lemma}[\cite{kalachev_maximally_2025}]
  \label{lem:igtope}
  Let $C_1,\dots,C_t\subsetneq\bF_q^n$ and $\epsilon>0$ be such that every $\epsilon$-closed subset of $[n]^t$ is inner-generated for $C_1\boxplus\cdots\boxplus C_t$. Then $C_1,\dots,C_t$ have product-expansion
  \begin{equation*}
    \rho \geq \frac{\epsilon^t}{t(2^t+1)^t}.
  \end{equation*}
\end{lemma}

The following definition describes the natural way of writing a generator matrix for a dual tensor code $C_1\boxplus\cdots\boxplus C_t=C^{(1)}+\cdots+C^{(t)}$, by simply stacking generator matrices for $C^{(1)},\dots,C^{(t)}$.

\begin{definition}
  For $C_1,\dots,C_t\subseteq\bF^n$ with full-rank generator matrices $G_i\in\bF^{k_i\times n}$ for $i\in[t]$, the associated \textbf{canonical generator matrix} for $C_1\boxplus\cdots\boxplus C_t$ is the matrix $G\in\bF^{R\times[n]^t}$ defined as follows. The rows of $G$ are labeled by the set $R=\bigsqcup_{i\in[t]}[k_i]\times[n]^{t-1}$. For every $i\in[t]$, every direction-$i$ column in $[n]^t$ (specified by a tuple $j=(j_\ell)_{\ell\in[t]\setminus\{i\}}\in[n]^{t-1}$) naturally induces a $k_i\times n$ submatrix of $G$ (with row labels in $[k_i]\times\{j\}$ and column labels in $\{(j_\ell)_{\ell\in[t]}:j_i\in[n]\}$), which we require to equal $G_i$; all entries of $G$ not belonging to any of these submatrices are set to $0$.
\end{definition}

\begin{lemma}[\cite{kalachev_maximally_2025}]
  \label{lem:igrank}
  For $C_1,\dots,C_t\subsetneq\bF_q^n$, let $C_1\boxplus\cdots\boxplus C_t$ have canonical generator matrix $G\in\bF_q^{R\times[n]^t}$. For every subset $A\subseteq[n]^t$, letting $B=B(A)\subseteq R$ denote the set of rows of $R$ associated to direction-$i$ columns (across all $i\in[t]$) that lie entirely inside $A$, then it holds that
  \begin{equation*}
    \rank(G|_{B\times A}) \leq \rank(G)-\rank(G|_{R\times([n]^t\setminus A)}),
  \end{equation*}
  with equality if and only if $A$ is inner-generated for $C_1\boxplus\cdots\boxplus C_t$.
\end{lemma}

\begin{remark}
  \label{remark:igrank}
  In Lemma~\ref{lem:igrank}, letting $k_i=\dim(C_i)$, then
  \begin{equation*}
    \rank(G) = \dim(C_1\boxplus\cdots\boxplus C_t) = n^t-\dim(C_1^\perp\otimes\cdots\otimes C_t^\perp) = n^t-\prod_{i\in[t]}(n-k_i).
  \end{equation*}
\end{remark}

\subsection{Proof of Product-Expansion}
\label{sec:peptproof}
In this section, we prove Theorem~\ref{thm:peptRS}. Below, we let all variables be defined as in the theorem.

We begin with the following lemma (see Definition~\ref{def:MDS} for the definition of MDS codes).

\begin{lemma}
  \label{lem:randommds}
  For $i\in[t]$, the code $\evl_{E_i}(\bF_q[X_1,\dots,X_u]^{[0,k_i)^u})$ is MDS of dimension $k_i^u$ with probability $\geq 1-n^22^n/q$.
\end{lemma}
\begin{proof}
  Let $Y=(Y_{e,j})_{e\in[n],j\in[u]}$ be a tuple of $=nu$ free variables, so that $\bF_q[Y]$ is the space of $nu$-variate polynomials over $\bF_q$. We define a polynomial-valued matrix $G(Y)\in\bF_q[Y]^{k_i^u\times n}$ with rows indexed by $\{0,\dots,k_i-1\}^u$ and columns indexed by $[n]$, where the row with index $(\ell_1,\dots,\ell_u)$ is equal to $(Y_{e,1}^{\ell_1}\cdots Y_{e,u}^{\ell_u})_{e\in[n]}\in\bF_q[Y]^n$. Fix an arbitrary bijection $E_i\cong[n]$, and let $G(E_i)$ denote the evaluation of $G(Y)$ at $(Y_{e,j}=e_j)_{e\in E_i\cong[n],j\in[u]}$. Then by definition $G(E_i)$ is a generator matrix for the code $\evl_{E_i}(\bF_q[X_1,\dots,X_u]^{[0,k_i)^u})$.

  We first argue that every $k_i^u\times k_i^u$ submatrix of $G(Y)$ has nonzero determinant. Specifically, let $C\subseteq[n]$ be a subset of $|C|=k_i^u$ columns of $G(Y)$, and let $G(Y)|_C\in\bF_q[Y]^{k_i^u\times k_i^u}$ denote the restriction of $G(Y)$ to columns in $C$; note that $G(Y)|_C$ only depends on the variables $Y_{e,j}$ for $e\in C$. Let $\phi:C\xrightarrow{\sim}[k_i]^u$ be an arbitrary bijection, let $z_1,\dots,z_{k_i}\in\bF_q$ be arbitrary distinct points, and define $y=(y_{e,j})_{e\in C,j\in[u]}\in\bF_q^{k_i^uu}$ by $y_{e,j}=z_{\phi(e)_j}$. Then by definition the row indexed by $(\ell_1,\dots,\ell_u)$ in $G(y)|_C\in\bF_q^{k_i^u\times k_i^u}$ is equal to $(z_{\phi(e)_1}^{\ell_1}\cdots z_{\phi(e)_u}^{\ell_u})_{e\in C}$, where $\{\phi(e):e\in C\}$ is the set of all tuples in $[k_i]^u$. Thus $G(y)|_C$ equals (up to a permutation of the columns) a generator matrix for the code $\evl_{\{z_1,\dots,z_{k_i}\}^u}(\bF_q[Z_1,\dots,Z_u]^{[0,k_i)^u})=\evl_{\{z_1,\dots,z_{k_i}\}}(\bF_q[Z]^{[0,k_i)})^{\otimes u}$, which is simply the entire space $\bF_q^{k_i^u}$, as the evaluations of univariate polynomials of degree $<k_i$ on any $k_i$ distinct points form a basis for functions on those $k_i$ points. Hence $G(y)|_C$ is a full-rank matrix, so $\det(G(y)|_C)$ is a nonzero element of $\bF_q$, and therefore $\det(G(Y)|_C)$ is a nonzero polynomial.

  Now by definition for each $C\subseteq[n]$ with $|C|=k_i^u$, the polynomial $\det(G(Y)|_C)$ has total degree $\leq k_i^u\cdot(k_i-1)\cdot u\leq n^2$, and hence the Schwartz-Zippel lemma (Lemma~\ref{lem:sz}) implies that
  \begin{equation*}
    \Pr_{E_i}[\det(G(E_i)|_C)=0] \leq \frac{n^2}{q}.
  \end{equation*}
  Union-bounding over all ${n\choose k_i^u}\leq 2^n$ possible choices of $C$ then yields the desired bound.
\end{proof}

We next show that for certain choices $E_1',\dots,E_t'$ of the sets $E_1,\dots,E_t$, the codes in Theorem~\ref{thm:peptRS} have good product-expansion.

\begin{lemma}
  \label{lem:tRSpe}
  Fix $m$ arbitrary distinct points $z_1,\dots,z_m\in\bF_q$. Let $E_1'=\cdots=E_{t-1}'=\{z_1,\dots,z_m\}^u$, and let $E_t'\subseteq\bF_q^u$ be any subset of size $|E_t'|=n$ for which $\evl_{E_t'}(\bF_q[X_1,\dots,X_u]^{[0,k_t)^u})$ is MDS. Then the $t$-tuple of codes
  \begin{equation}
    \label{eq:applypeLTC}
    \left(\evl_{E_1'}(\bF_q[X_1,\dots,X_u]^{[0,k_1)^u}), \dots, \evl_{E_{t-1}'}(\bF_q[X_1,\dots,X_u]^{[0,k_{t-1})^u}), \evl_{E_t'}(\bF_q[X_1,\dots,X_u]^{[0,k_t)^u})^\perp\right)
  \end{equation}
  has product-expansion at least
  \begin{equation}
    \label{eq:applypeLTCbound}
    \rho' := \left(\frac{1}{um}\right)^{O(t^2)} \cdot \left(\frac{k_t}{m}\right)^u \cdot \prod_{i=1}^{t-1}\left(1-\frac{k_i}{m}\right)^{4ut}.
  \end{equation}
\end{lemma}
\begin{proof}
  Let $(C_1,\dots,C_t)$ denote the tuple in~(\ref{eq:applypeLTC}). For $i\in[t-1]$, Theorem~\ref{thm:tensorltc} implies that $C_i=\evl_{\{z_1,\dots,z_m\}}(\bF_q[X]^{[0,k_i)})^{\otimes u}$ is locally testable of locality $m$ and soundness $\rho_i\geq(1-k_i/m)^{3u}/(um)^{O(1)}$. Then Lemma~\ref{lem:peLTC} implies that $(C_1,\dots,C_t)$ has product-expansion at least
  \begin{align*}
    \rho(t)
    &\geq \left(\frac{(k_t/m)^u}{3}\prod_{t'=1}^{t-1}\frac{(1-k_{t'}/m)^u}{3}\right)\cdot\left(\prod_{t'=1}^{t-1}\prod_{i=1}^{t'}\frac{(1-k_i/m)^{4u}}{(um)^{O(1)}}\right) \\
    &\geq \left(\frac{1}{um}\right)^{O(t^2)} \cdot \left(\frac{k_t}{m}\right)^u \cdot \prod_{i=1}^{t-1}\left(1-\frac{k_i}{m}\right)^{4ut},
  \end{align*}
  where here we have used the fact that for $i\in[t-1]$ the code $C_i$ is a $u$th tensor power of a $k_i$-dimensional (punctured) Reed-Solomon code and hence has distance $(m-k_i+1)^u$, while $C_t$ is MDS and hence has distance $k_t^u+1$.
\end{proof}

We prove Theorem~\ref{thm:peptRS} by arguing that product-expansion is preserved with high probability when $E_1',\dots,E_t'$ in Lemma~\ref{lem:tRSpe} are replaced by random sets $E_1,\dots,E_t$. Our proof relies heavily on the results of \cite{kalachev_maximally_2025} stated in Section~\ref{sec:petech} relating product-expansion to \textit{inner-generated} sets and \textit{$\epsilon$-closed} sets (see Definition~\ref{def:ig}).

\begin{proof}[Proof of Theorem~\ref{thm:peptRS}]
  Let $(C_1,\dots,C_t)$ denote the tuple of codes in~(\ref{eq:ptRS}). By Lemma~\ref{lem:randommds}, with probability $\geq 1-tn^22^n/q$ it holds that every $C_i$ is MDS, of dimension $k_i^u$ if $i\in[t-1]$ and of dimension $n-k_i^u$ if $i=t$; we condition on this event for the remainder of the proof.

  Define $E_1',\dots,E_{t-1}'$ as in Lemma~\ref{lem:tRSpe}, let $E_t'=E_t$, and then let $(C_1',\dots,C_t')$ denote the tuple of codes in~(\ref{eq:applypeLTC}), with product-expansion at least $\rho'$ as defined in~(\ref{eq:applypeLTCbound}).

  Let $Y=(Y_{i,e,j})_{i\in[t],e\in[n],j\in[u]}$ be a tuple of $tnu$ free variables, so that $\bF_q[Y]$ (resp.~$\bF_q(Y)$) is the ring of $tnu$-variate polynomials (resp.~field of $tnu$-variate rational functions) over $\bF_q$. We also write $Y_i=(Y_{i,e,j})_{e\in[n],j\in[u]}$.

  For $i\in[t]$, similarly as in the proof of Lemma~\ref{lem:randommds}, we define a polynomial-valued matrix $G_i(Y_i)\in\bF_q[Y_i]^{k_i^u\times n}$ with rows indexed by $\{0,\dots,k_i-1\}^u$ and columns indexed by $[n]$, where the row with index $(\ell_1,\dots,\ell_u)$ is equal to $(Y_{i,e,1}^{\ell_1}\cdots Y_{i,e,u}^{\ell_u})_{e\in[n]}\in\bF_q[Y]^n$. Fix arbitrary bijections $E_i\cong[n]$ for $i\in[t]$, and let $G_i(E_i)$ denote the evaluation of $G_i(Y_i)$ at $(Y_{i,e,j}=e_j)_{e\in E_i\cong[n],j\in[u]}$. Then by definition $G_i(E_i)$ is a generator matrix for the code $\evl_{E_i}(\bF_q[X_1,\dots,X_u]^{[0,k_i)^u})$.

  Viewing $Y_i$ as a set of $n$ points $\{(Y_{i,e,1},\dots,Y_{i,e,u}):e\in[n]\}$, then we may similarly view $G_i(Y_i)$ as a generator matrix for the code\footnote{In this paper we have typically presented coding theory definitions with respect to finite fields, but codes are naturally also well-defined over infinite fields, as considered here.} $\evl_{Y_i}(\bF_q(Y_i)[X_1,\dots,X_u]^{[0,k_i)^u})\subseteq\bF_q(Y_i)^n$ over the rational function field $\bF_q(Y_i)\subseteq\bF_q(Y)$. Let $G(Y)$ be the canonical generator matrix for $\im(G_1(Y_1))\boxplus\cdots\boxplus\im(G_t(Y_t))$. Then letting $E=(E_1,\dots,E_t)$ and $E'=(E_1',\dots,E_t')$, by definition $G(E)$ is the canonical generator matrix for $C_1\boxplus\cdots\boxplus C_t$, and $G(E')$ is the canonical generator matrix for $C_1'\boxplus\cdots\boxplus C_t'$.

  Now by Lemma~\ref{lem:petoig}, every $\rho'$-closed subset $A\subseteq[n]^t$ is inner-generated for $C_1'\boxplus\cdots\boxplus C_t'$. Let
  \begin{equation*}
    r = \rank(G(Y)) = \rank(G(E)) = \rank(G(E')) = n^t-k_t^u\cdot\prod_{i=1}^{t-1}(n-k_i^u).
  \end{equation*}
  Note that the equalities above hold because by assumption for each $i\in[t]$, the three matrices $G_i(Y),G_i(E_i),G_i(E_i')$ are full rank of the same dimensions. Then by Lemma~\ref{lem:igrank} and Remark~\ref{remark:igrank}, for every such $\rho'$-closed set $A$, defining $B(A)$ as in Lemma~\ref{lem:igrank}, we have that
  \begin{equation}
    \label{eq:rankbound}
    \rank(G(E')|_{B(A)\times A}) = r-\rank(G(E')|_{R\times([n]^t\setminus A)}),
  \end{equation}
  where $R$ denotes the set of all row labels of $G(E')$ (or equivalently, of $G(Y)$).

  Therefore there exist subsets $B_1(A)\subseteq B(A)$, $A_1(A)\subseteq A$, $B_2(A)\subseteq R$, $A_2(A)\subseteq[n]^t\setminus A$ with
  \begin{align*}
    |B_1(A)| &= |A_1(A)| = \rank(G(E')|_{B(A)\times A}) \\
    |B_2(A)| &= |A_2(A)| = \rank(G(E')|_{R\times([n]^t\setminus A)})
  \end{align*}
  such that the square submatrices $G(E')|_{B_1(A)\times A_1(A)}$ and $G(E')|_{B_2(A)\times A_2(A)}$ have full rank, and hence have nonzero determinant. Letting $Y_{<t}=(Y_1,\dots,Y_{t-1})$, and recalling that $E_t'=E_t$, it follows that
  \begin{align*}
    \det(G(Y_{<t},E_t)|_{B_1(A)\times A_1(A)}) &\hspace{1em}\text{ and }\hspace{1em} \det(G(Y_{<t},E_t)|_{B_2(A)\times A_2(A)})
  \end{align*}
  are nonzero polynomials in $\bF_q[Y_{<t}]$, as both of these polynomials have a nonzero evaluation at $Y_{<t}=E_{<t}'$. Note that both of these polynomials have total degree at most $n^t\cdot(m-1)\cdot u\leq n^{t+1}$.

  Let $\cF$ denote the ``bad'' event that $C_1,\dots,C_t$ are not all MDS, or that for some $\rho'$-closed subset $A\subseteq[n]^t$ either
  \begin{align}
    \label{eq:rankdets}
    \det(G(E)|_{B_1(A)\times A_1(A)}) &\hspace{1em}\text{ or }\hspace{1em} \det(G(E)|_{B_2(A)\times A_2(A)}).
  \end{align}
  equals $0$. Then by Lemma~\ref{lem:randommds} along with the Schwartz-Zippel lemma (Lemma~\ref{lem:sz}), and because there are at most $2^{n^t}$ distinct choices of $A$, we have
  \begin{align*}
    \Pr_E[\cF]
    &\leq t\cdot\frac{n^22^n}{q} + 2^{n^t+1}\cdot\frac{n^{t+1}}{q} \leq \frac{n^{t+1}2^{n^t+2}}{q}.
  \end{align*}

  Now assume that the choice of $E$ is such that $\cF$ does not occur. Then by definition for every $\rho'$-closed subset $A\subseteq[n]^t$, both determinants in~(\ref{eq:rankdets}) are nonzero, so
  \begin{align*}
    \rank(G(E))|_{B(A)\times A} &\geq |B_1(A)| = \rank(G(E')|_{B(A)\times A}) \\
    \rank(G(E)|_{R\times([n]^t\setminus A)}) &\geq |B_2(A)| = \rank(G(E')|_{R\times([n]^t\setminus A)}).
  \end{align*}
  Hence it follows from~(\ref{eq:rankbound}) that both inequalities above must be equalities, and
  \begin{equation*}
    \rank(G(E)|_{B(A)\times A}) = r-\rank(G(E)|_{R\times([n]^t\setminus A)}).
  \end{equation*}
  Then Lemma~\ref{lem:igrank} implies that $A$ is inner-generated for $C_1\boxplus\cdots\boxplus C_t$. As this conclusion applies for every $\rho'$-closed $A\subseteq[n]^t$, Lemma~\ref{lem:igtope} implies that $(C_1,\dots,C_t)$ has product-expansion at least
  \begin{align*}
    \frac{{\rho'}^t}{t(2^t+1)^t}
    &\geq \left(\frac{1}{um}\right)^{O(t^3)} \cdot \left(\frac{k_t}{m}\right)^{ut} \cdot \prod_{i=1}^{t-1}\left(1-\frac{k_i}{m}\right)^{4ut^2},
  \end{align*}
  as desired, where above we recall the definition of $\rho'$ in~(\ref{eq:applypeLTCbound}).
\end{proof}

\section{Transversal $CCZ$ Gates on Products of Punctured Tensor Reed-Solomon Codes}
\label{sec:tripleprod}
In this section, we apply Theorem~\ref{thm:transgen} to obtain a transversal $CCZ$ gate on the subsystem product of three quantum codes arising from randomly punctured tensor products of Reed-Solomon codes. We apply our product-expansion result in Section~\ref{sec:perp} to show that these subsystem products have good (almost-linear) distance. Specifically, we will prove the following:

\begin{theorem}
  \label{thm:tripleprodparam}
  For every fixed $\epsilon'>0$, there exists an infinites families of:
  \begin{enumerate}
  \item $[[N,\Theta(N),\Theta(N^{1-\epsilon'})]]_q$ subsystem codes supporting a transversal $CCZ$ (and $U^3$) gate on $\Theta(N)$ logical qudits,
  \item $[[N,N^{1-o(1)},N^{1-o(1)}]]_q$ subsystem codes supporting a transversal $CCZ$ (and $U^3$) gate on $N^{1-o(1)}$ logical qudits,
  \end{enumerate}
  both of which have locality $3N^{1/3}$ and alphabet size $q\leq N^{5/3}2^{N^{4/3}+8}$.
\end{theorem}

As stated, the alphabet size (i.e.~local qudit dimension) $q$ in Theorem~\ref{thm:tripleprodparam} will be a power of $2$, but as described in Remark~\ref{remark:fieldsize}, the result extends to allow for powers of arbtrary primes, with just a slight blowup in alphabet size.

\subsection{Construction}
\label{sec:cczinstan}
This section presents our specific construction. We give explicit constants below for concreteness, but we make no attempt to optimize them.

Let $\bF_q$ be a finite field, and let $100\leq m\leq q$ and $u\geq 1$ be integers. Let $n=m^u$. Sample uniformly random subsets $E_1,E_2,E_3\subseteq\bF_q^u$, each of size $|E_i|=n$. Set $k_0=\lfloor m/4\rfloor$ and $\epsilon=1/100$.

For $i\in[3]$, let $\gamma_i\in(\bF_q^*)^{E_i}\cong(\bF_q^*)^n$ be a vector of nonzero field elements such that
\begin{align}
  \label{eq:gammareq}
  \begin{split}
    \gamma_i \cdot \evl_{E_i}(X_1^{k_0}\cdots X_u^{k_0}) &= 1 \\
    \gamma_i \cdot \evl_{E_i}(X_1^{\ell_1}\cdots X_u^{\ell_u}) &= 0 \hspace{1em}\forall\; (\ell_1,\dots,\ell_u)\in\{0,\dots,2k_0\}^u\setminus\{(k_0,\dots,k_0)\}.
  \end{split}
\end{align}
We will show in Section~\ref{sec:cczanalys} below that such $\gamma_i$ exist with high probability.

We now define length-$n$ CSS codes $(Q^i=(Q^i_X,Q^i_Z))_{i\in[3]}$ by\footnote{Recall here that `$*$' denotes pointwise multiplication; see also Lemma~\ref{lem:coeffsame}.}
\begin{align*}
  &Q^1_X = \evl_{E_1}(\bF_q[X_1,\dots,X_u]^{[0,k^1_X)^u})^\perp, \hspace{1em} Q^1_Z = (\gamma_1*\evl_{E_1}(\bF_q[X_1,\dots,X_u]^{[0,k^1_Z)^u}))^\perp \\
  &Q^2_X = \evl_{E_2}(\bF_q[X_1,\dots,X_u]^{[0,k^2_X)^u})^\perp, \hspace{1em} Q^2_Z = (\gamma_2*\evl_{E_2}(\bF_q[X_1,\dots,X_u]^{[0,k^2_Z)^u}))^\perp \\
  &Q^3_X = \evl_{E_3}(\bF_q[X_1,\dots,X_u]^{[0,k^3_X)^u})^\perp, \hspace{1em} Q^3_Z = \evl_{E_3}(\bF_q[X_1,\dots,X_u]^{[0,k^3_Z)^u}),
\end{align*}
where
\begin{align*}
  k^1_X = k^2_X &= \lfloor\epsilon k_0\rfloor, \hspace{1.5em} k^1_Z = k^2_Z = \left\lfloor\frac23k_0\right\rfloor \\
  k^3_X &= \lfloor\epsilon k_0\rfloor, \hspace{1.5em} k^3_Z = \left\lfloor\frac13k_0\right\rfloor.
\end{align*}
We will show below that the subsystem product (see Definition~\ref{def:subhomprod})
\begin{equation}
  \label{eq:cczQdef}
  Q = (Q_X,Q_Z) = Q^1 \otimes Q^2 \otimes Q^3
\end{equation}
supports a transversal $CCZ$ gate.

In order to apply Theorem~\ref{thm:transgen}, we also define subspaces $(L_i\subseteq Q^i_Z)_{i\in[3]}$ by
\begin{align*}
  L_i &= \evl_{E_i}(\bF_q[X_1,\dots,X_u]^{[\underline{\ell},\overline{\ell})^u}),
\end{align*}
where
\begin{align*}
  \underline{\ell} &= \left\lceil\left(\frac13-\epsilon\right)k_0\right\rceil, \hspace{1em} \overline{\ell} = \left\lfloor\frac13k_0\right\rfloor.
\end{align*}
Note that~(\ref{eq:gammareq}) ensures that $L_1\subseteq Q^1_Z$ and that $L_2\subseteq Q^2_Z$, while by definition $L_3\subseteq Q^3_Z$.

\subsection{Analysis}
\label{sec:cczanalys}
We now analyze the construction in Section~\ref{sec:cczinstan} by showing the following:

\begin{theorem}
  \label{thm:tripleprodccz}
  Define all variables as in Section~\ref{sec:cczinstan}. Then with probability $> 1-n^52^{n^4+6}/q$ over the choice of $E_1,E_2,E_3$, all of the following hold:
  \begin{enumerate}
  \item There exist $\gamma_1,\gamma_2,\gamma_3\in(\bF_q^*)^n$ satisfying~(\ref{eq:gammareq}).
  \item The resulting code $Q$ in~(\ref{eq:cczQdef}) is a $[[N,K,D]]_q$ subsystem code of locality $3n$ for
    \begin{align*}
      N &= n^3 \\
      K &\geq n^3 \cdot \frac{1}{20^{u+1}} \\
      D &\geq n^3 \cdot \frac{1}{(um)^{O(1)}\cdot 2^{O(u)}}.
    \end{align*}
  \item Defining $L$ and $S$ as in Theorem~\ref{thm:transgen}, there exists an encoding map $\Enc:\bF_q^{\dim(L)}\xrightarrow{\sim}(L+S)/S$ for which $(Q,\Enc)$ supports a transversal $CCZ$ (and $U^3$) gate on
    \begin{equation*}
      \dim(L) = (\overline{\ell}-\underline{\ell})^u \geq (\epsilon k_0-2)^u \geq \Omega(m)^u
    \end{equation*}
    logical qudits.
  \end{enumerate}
\end{theorem}
\begin{proof}
  By Lemma~\ref{lem:randommds}, Theorem~\ref{thm:peptRS}, and Lemma~\ref{lem:coeffsame}, with probability $> 1-n^52^{n^4+6}/q$, all of the following hold:
  \begin{enumerate}
  \item\label{it:eventMDS} For every $i\in[3]$ and every $0\leq k\leq 2k_0+1$, the code $\evl_{E_i}(\bF_q[X_1,\dots,X_u]^{[0,k)^u})$ is MDS of dimension $k^u$.
  \item\label{it:eventpe} All of the tuples
    \begin{align*}
      &(Q^1_X), ({Q^1_Z}^\perp,Q^2_X), ({Q^1_Z}^\perp,{Q^2_Z}^\perp,Q^3_X) \\
      &(Q^1_Z), ({Q^1_X}^\perp,Q^2_Z), ({Q^1_X}^\perp,{Q^2_X}^\perp,Q^3_Z)
    \end{align*}
    have product-expansion at least
    \begin{equation}
      \label{eq:hppe}
      \left(\frac{1}{um}\right)^{O(1)} \cdot \left(\frac{\epsilon k_0-1}{m}\right)^{4u} \cdot \left(1-\frac{k_0}{m}\right)^{192u} \geq \frac{1}{(um)^{O(1)}\cdot 2^{O(u)}}.
    \end{equation}
  \end{enumerate}
  Specifically, the product-expansion bound on the LHS of~(\ref{eq:hppe}) follows by applying Theorem~\ref{thm:peptRS} with $t=1$ for $(Q^1_X),(Q^1_Z)$, with $t=2$ for $({Q^1_Z}^\perp,Q^2_X),({Q^1_X}^\perp,Q^2_Z)$, with $t=3$ for $({Q^1_Z}^\perp,{Q^2_Z}^\perp,Q^3_X)$, and with $t=4$ for $({Q^1_X}^\perp,{Q^2_X}^\perp,Q^3_Z)$. In particular, for this final case, because all three codes ${Q^1_X}^\perp,{Q^2_X}^\perp,Q^3_Z$ are of the form $\evl_{E_i}(\bF_q[X_1,\dots,X_u]^{[0,k^i)})$ for some $k^i\leq k_0$, we apply Theorem~\ref{thm:peptRS} to the 4-tuple of codes $({Q^1_X}^\perp,{Q^2_X}^\perp,Q^3_Z,C^4)$ for $C^4=\evl_{E_4}(\bF_q[X_1,\dots,X_u]^{[0,k_0)})^\perp$ with $E_4\subseteq\bF_q^u$ a uniformly random subset of size $|E_4|=n$. The resulting lower bound from Theorem~\ref{thm:peptRS} on the product-expansion of the 4-tuple $({Q^1_X}^\perp,{Q^2_X}^\perp,Q^3_Z,C^4)$ is also a lower bound on the product-expansion of the triple $({Q^1_X}^\perp,{Q^2_X}^\perp,Q^3_Z)$ by Lemma~\ref{lem:pesubtuple}.

  Conditioning on this event that item~\ref{it:eventMDS} and item~\ref{it:eventpe} above hold, we now prove that the three desired properties hold:
  \begin{enumerate}
  \item For $i\in[3]$, by assumption the code $\evl_{E_i}(\bF_q[X_1,\dots,X_u]^{[0,2k_0]^u})$ is MDS of dimension $(2k_0+1)^u$, so in particular the generator matrix rows $\evl_{E_i}(X_1^{\ell_1}\cdots X_u^{\ell_u})$ for $(\ell_1,\dots,\ell_u)\in\{0,\dots,2k_0\}^u$ are all linearly independent. Therefore there exists some $\gamma_i^0\in\bF_q^{E_i}$ with the desired inner product values in~(\ref{eq:gammareq}). Furthermore, because $\evl_{E_i}(\bF_q[X_1,\dots,X_u]^{[0,2k_0]^u})$ is MDS, (the marginal of) each component of a uniformly random dual codeword $\gamma_i'\in\evl_{E_i}(\bF_q[X_1,\dots,X_u]^{[0,2k_0]^u})^\perp$ is uniformly distributed over $\bF_q$. Therefore union-bounding over all $n$ code components implies that with probability $\geq 1-n/q$, all $n$ components of $\gamma_i^0+\gamma_i'$ are nonzero. As we may assume that $n<q$ (for otherwise the theorem statement has a negative probability lower bound and hence is vacuous), it follows that there exists a choice of $\gamma_i'$ for which $\gamma_i=\gamma_i^0+\gamma_i'$ has all $n$ components nonzero. By definition such $\gamma_i$ also satisfies~(\ref{eq:gammareq}), as desired.
  \item By definition for $i=1,2$, the code $Q^i$ has dimension at least $n-(k^i_X)^u-(k^i_Z)^u\geq n/2$, while $Q^3$ has dimension at least $(k^3_Z)^u-(k^3_X)^u\geq(m/12-2)^u-(m/400)^u\geq n/20^u$, where here we use the assumption that $m\geq 100$. Hence $Q$ has dimension at least $n^3/20^{u+1}$. Meanwhile, the distance bound in the theorem statement follows directly from Theorem~\ref{thm:subpe} along with the product-expansion bound in~(\ref{eq:hppe}).
  \item It suffices to show that the subspaces $L,S$ satisfy the conditions in Theorem~\ref{thm:transgen}. For each $i\in[3]$, recall that by the assumption that $\evl_{E_i}(\bF_q[X_1,\dots,X_u]^{[0,2k_0]^u})$ is MDS of dimension $(2k_0+1)^u$, the rows $\evl_{E_i}(X_1^{\ell_1}\cdots X_u^{\ell_u})$ for $(\ell_1,\dots,\ell_u)\in\{0,\dots,2k_0\}^u$ are all linearly independent. Because $L_i$ equals the span of the rows with $(\ell_1,\dots,\ell_u)\in[\underline{\ell},\overline{\ell})^u$, while ${Q^i_X}^\perp$ equals the span of the rows with $(\ell_1,\dots,\ell_u)\in[0,k^i_X)^u$, and by definition $k^i_X<\underline{\ell}$, it follows that $L_i\cap{Q^i_X}^\perp=\{0\}$.

    Thus it remains to be shown that~(\ref{eq:multprop}) holds with $r=3$. That is, we want to show that
    \begin{equation}
      \label{eq:multgoal}
      L^{*3}\cap(S*(L+S)^{*2}) = \{0\}.
    \end{equation}
    We prove~(\ref{eq:multgoal}) through the following two claims. Below, we let $E=E_1\times E_2\times E_3\subseteq\bF_q^{3u}$, so that $|E|=n^3$, and we let $\gamma=\gamma_1\otimes\gamma_2\otimes\gamma_3\in\bF_q^E$. We also denote elements of $(\{0,\dots,m\}^u)^3=\{0,\dots,m\}^{3u}$ by $\ell_{\filler,\filler}=(\ell_{i,\filler})_{i\in[3]}=(\ell_{i,j})_{i\in[3],j\in[u]}$. Similarly, we let $X_{\filler,\filler}=(X_{i,\filler})_{i\in[3]}=(X_{i,j})_{i\in[3],j\in[u]}$ denote the tuple of $3u$ free variables.
    \begin{claim}
      \label{claim:L3prop}
      For every nonzero $c\in L^{*3}$, there exists some $\ell_{\filler,\filler}\in\{0,\dots,3\epsilon k_0\}^{3u}$ such that
      \begin{equation*}
        \gamma \cdot (\evl_E(X_{\filler,\filler}^{\ell_{\filler,\filler}})*c) \neq 0.
      \end{equation*}
    \end{claim}
    \begin{proof}
      By definition
      \begin{equation*}
        L^{*3} \subseteq \evl_E(\bF_q[X_{\filler,\filler}]^{[(1/3-\epsilon)k_0,k_0/3)^{3u}})^{*3} \subseteq \evl_E(\bF_q[X_{\filler,\filler}]^{[(1-3\epsilon)k_0,k_0)^{3u}}).
      \end{equation*}
      Hence any nonzero $c\in L^{*3}$ can be expressed as
      \begin{equation*}
        c = \sum_{\ell_{\filler,\filler}''\in(\bZ\cap[(1-3\epsilon)k_0,k_0))^{3u}}c_{\ell_{\filler,\filler}''} \cdot \evl_E(X_{\filler,\filler}^{\ell_{\filler,\filler}''}),
      \end{equation*}
      for some coefficients $c_{\ell_{\filler,\filler}}''\in\bF_q$. Fix some $\ell_{\filler,\filler}'\in(\bZ\cap[(1-3\epsilon)k_0,k_0))^{3u}$ for which $c_{\ell_{\filler,\filler}'}\neq 0$. Then $\ell_{\filler,\filler}:=(k_0,\dots,k_0)-\ell_{\filler,\filler}'$ by definition lies inside $\{0,\dots,3\epsilon k_0\}^{3u}$, and satisfies
      \begin{align*}
        \gamma \cdot (\evl_E(X_{\filler,\filler}^{\ell_{\filler,\filler}}) * c)
        &= \sum_{\ell_{\filler,\filler}''\in(\bZ\cap[(1-3\epsilon)k_0,k_0))^{3u}} c_{\ell_{\filler,\filler}''} (\gamma\cdot\evl_E(X_{\filler,\filler}^{\ell_{\filler,\filler}+\ell_{\filler,\filler}''})) \\
        &= \sum_{\ell_{\filler,\filler}''\in(\bZ\cap[(1-3\epsilon)k_0,k_0))^{3u}} c_{\ell_{\filler,\filler}''} \prod_{i=1}^3(\gamma_i\cdot\evl_E(X_{i,\filler}^{\ell_{i,\filler}+\ell_{i,\filler}''})) \\
        &= c_{\ell_{\filler,\filler}'} \\
        &\neq 0,
      \end{align*}
      as desired, where the third equality above holds by~(\ref{eq:gammareq}).
    \end{proof}

    \begin{claim}
      \label{claim:SLSprop}
      For every $c\in S*(L+S)^{*2}$ and every $\ell_{\filler,\filler}\in\{0,\dots,3\epsilon k_0\}^{3u}$, we have
      \begin{equation*}
        \gamma \cdot (\evl_E(X_{\filler,\filler}^{\ell_{\filler,\filler}})*c) = 0.
      \end{equation*}
    \end{claim}
    \begin{proof}
      It suffices to show the result for $c$ ranging over a set of vectors that span $S*(L+S)^{*2}$, as then the desired result follows by linearity. Thus we may restrict attention to $c$ of the form $c=c_1*c_2*c_3$ for $c_1\in S$ and $c_2,c_3\in L+S$. Recalling that
      \begin{equation*}
        S = Q_Z\cap Q_X^\perp = S_1+S_2+S_3
      \end{equation*}
      for
      \begin{align*}
        S_1 &= {Q^1_X}^\perp\otimes Q^2_Z\otimes Q^3_Z \\
        S_2 &= Q^1_Z\otimes{Q^2_X}^\perp\otimes Q^3_Z \\
        S_3 &= Q^1_Z\otimes Q^2_Z\otimes{Q^3_X}^\perp,
      \end{align*}
      we may further assume that $c_1$ lies inside one of $S_1,S_2,S_3$, and that $c_2,c_3$ each lie inside one of $L,S_1,S_2,S_3$.

      If any of $c_1,c_2,c_3$ lie inside $S_3$, then
      \begin{equation*}
        c=c_1*c_2*c_3\in\bF_q^{E_1}\otimes\bF_q^{E_2}\otimes({Q^3_X}^\perp*(Q^3_Z)^{*2}),
      \end{equation*}
      where
      \begin{equation*}
        {Q^3_X}^\perp*(Q^3_Z)^{*2} \subseteq \evl_{E_3}(\bF_q[X_{3,\filler}]^{[0,(2/3+\epsilon)k_0)^u}).
      \end{equation*}
      Hence for every $\ell_{\filler,\filler}\in\{0,\dots,3\epsilon k_0\}^{3u}$,
      \begin{equation*}
        \evl_E(X_{\cdot,\cdot}^{\ell_{\cdot,\cdot}})*c \in \bF_q^{E_1} \otimes \bF_q^{E_2} \otimes \evl_{E_3}(\bF_q[X_{3,\filler}]^{[0,(2/3+4\epsilon)k_0)^u})
      \end{equation*}
      It follows by~(\ref{eq:gammareq}) that the RHS above lies in $\ker(I\otimes I\otimes\gamma_3^\top)$, which in turn implies that $\gamma\cdot(\evl_E(X_{\cdot,\cdot}^{\ell_{\cdot,\cdot}})*c)=0$, as desired.

      Thus it remains to consider the case where none of $c_1,c_2,c_3$ lie inside $S_3$. That is, we now assume that $c_1$ lies inside $S_1$ or $S_2$, and $c_2,c_3$ each lie inside $L$, $S_1$, or $S_2$. If both $c_2,c_3\in L$, then assuming that $c_1\in S_1$ (the proof for $c_1\in S_2$ is analogous), we have
      \begin{equation*}
        c=c_1*c_2*c_3 \in ({Q^1_X}^\perp*L_1^{*2})\otimes\bF_q^{E_2}\otimes\bF_q^{E_3},
      \end{equation*}
      where
      \begin{equation*}
        {Q^1_X}^\perp*L_1^{*2} \subseteq \evl_{E_1}(\bF_q[X_{1,\filler}]^{[(2/3+\epsilon)k_0)^u}).
      \end{equation*}
      Hence for every $\ell_{\filler,\filler}\in\{0,\dots,3\epsilon k_0\}^{3u}$,
      \begin{equation*}
        \evl_E(X_{\cdot,\cdot}^{\ell_{\cdot,\cdot}})*c \in \evl_{E_1}(\bF_q[X_{1,\filler}]^{[(2/3+4\epsilon)k_0)^u}) \otimes \bF_q^{E_2} \otimes \bF_q^{E_3}.
      \end{equation*}
      It follows by~(\ref{eq:gammareq}) that the RHS above lies in $\ker(I\otimes I\otimes\gamma_3^\top)$, which in turn implies that $\gamma\cdot(\evl_E(X_{\cdot,\cdot}^{\ell_{\cdot,\cdot}})*c)=0$, as desired.

      Therefore it remains to consider the case where either $c_2\notin L$ or $c_3\notin L$. By definition either at most one of $c_1,c_2,c_3$ lies inside $S_1$, or at most one lies inside $S_2$; assume the former, as the proof for the latter is analogous. Thus we have that one of $c_1,c_2,c_3$ lies in $S_1$, one lies inside $S_2$ (as above we covered the case where two lie inside $L$), and the third lies inside $S_2$ or $L$, and hence inside $S_2+L$. Then
      \begin{equation*}
        c=c_1*c_2*c_3 \in \bF_q^{E_1}\otimes(Q^2_Z*{Q^2_X}^\perp*({Q^2_X}^\perp+L_2))\otimes\bF_q^{E_3},
      \end{equation*}
      where
      \begin{align*}
        Q^2_Z*{Q^2_X}^\perp*({Q^2_X}^\perp+L_2) &\subseteq (\gamma_2*\evl_{E_2}(\bF_q[X_{2,\filler}]^{[0,2k_0/3)^u}))^\perp * \evl_{E_2}(\bF_q[X_{2,\filler}]^{[0,(1/3+\epsilon)k_0)^u}).
      \end{align*}
      Therefore for every $\ell_{\filler,\filler}\in\{0,\dots,3\epsilon k_0\}^{3u}$,
      \begin{align*}
        \evl_E(X_{\filler,\filler}^{\ell_{\filler,\filler}})*c
        &\in \bF_q^{E_1} \otimes \left((\gamma_2*\evl_{E_2}(\bF_q[X_{2,\filler}]^{[0,2k_0/3)^u}))^\perp * \evl_{E_2}(\bF_q[X_{2,\filler}]^{[0,(1/3+4\epsilon)k_0)^u})\right) \otimes \bF_q^{E_3}.
      \end{align*}
      But the RHS above lies within $\ker(I\otimes\gamma_2^\top\otimes I)$, as for every $a\in(\gamma_2*\evl_{E_2}(\bF_q[X_{2,\filler}]^{[0,2k_0/3)^u}))^\perp$ and $b\in\evl_{E_2}(\bF_q[X_{2,\filler}]^{[0,(1/3+4\epsilon)k_0)^u})$, then
      \begin{align*}
        \gamma_2^\top(a*b)
        &= a\cdot(\gamma_2*b) = 0,
      \end{align*}
      where the second equality above holds because by definition $a$ and $\gamma_2*b$ lie inside codes that are dual to each other. Thus we have shown that $(I\otimes\gamma_2^\top\otimes I)(\evl_E(X_{\filler,\filler}^{\ell_{\filler,\filler}})*c)=0$, and hence $\gamma\cdot(\evl_E(X_{\filler,\filler}^{\ell_{\filler,\filler}})*c)=0$, as desired.
    \end{proof}
        Claim~\ref{claim:L3prop} and Claim~\ref{claim:SLSprop} together imply the desired equality~(\ref{eq:multgoal}). \qedhere
  \end{enumerate}
\end{proof}

\subsection{Choosing Parameters}
In this section, we show how to choose parameters in Theorem~\ref{thm:tripleprodccz} to obtain the subsystem codes described in Theorem~\ref{thm:tripleprodparam}, which have almost-linear dimension and distance, have cube-root locality, and support transversal $CCZ$ gates.

\begin{proof}[Proof of Theorem~\ref{thm:tripleprodparam}]
  Both families in the theorem statement are instantiations of the construction from Section~\ref{sec:cczinstan}, and the code properties follow from Theorem~\ref{thm:tripleprodccz}. Specifically, let $\beta=O(1)$ be a sufficiently large constant such that the distance bound in Theorem~\ref{thm:tripleprodccz} is $D\geq n^3/(um\cdot 2^u)^\beta$. Then we obtain the desired families by instantiating the parameters in Theorem~\ref{thm:tripleprodccz} as follows, where in both cases we let $q$ be the greatest power of $2$ less than $N^{5/3}2^{N^{4/3}+8}$, and we recall that $n=m^u$:
  \begin{enumerate}
  \item Set $u=\lceil\beta/3\epsilon'\rceil$, so that $m=N^{1/3u}\leq N^{\epsilon'/\beta}$. Then for every fixed constant $0<\epsilon'=\Omega(1)$, applying Theorem~\ref{thm:tripleprodccz} with $m\rightarrow\infty$ gives a family of $[[N,\; \Omega(N),\; \Omega(N/m^\beta)\geq\Omega(N^{1-\epsilon'})]]_q$ codes supporting transversal $CCZ$ (and $U^3$).
  \item Set $u=\lceil\beta\log\log(N)/3\rceil$, so that $m=N^{1/3u}\leq N^{1/\beta\log\log(N)}$. Then applying Theorem~\ref{thm:tripleprodccz} with $m\rightarrow\infty$ gives a family of $[[N,\; N/\poly(\log N),\; N^{1-1/\log\log(N)}/\poly(\log N)\geq N^{1-o(1)}]]_q$ codes supporting transversal $CCZ$ (and $U^3$).\qedhere
  \end{enumerate}
\end{proof}

\begin{remark}
  \label{remark:fieldsize}
  In the proof of Theorem~\ref{thm:tripleprodparam}, for simplicity we chose $q$ to be a power of $2$. However, we could similarly choose $\bF_q$ to be a field of any given prime characteristic $p$, by in particular choosing $q$ to be the least power of $p$ greater than $N^{5/3}2^{N^{4/3}+6}$, and otherwise choosing the same parameters as in Theorem~\ref{thm:tripleprodparam}.
\end{remark}

\section{Acknowledgments}
This work was done in part while the authors were attending the Spring 2024 programs at the Simons Institute for the Theory of Computing.

\bibliographystyle{alpha}
\bibliography{library}

\newcommand{\etalchar}[1]{$^{#1}$}
\begin{thebibliography}{KPRvdO16}

\bibitem[BDET24]{breuckmann_cups_2024}
Nikolas~P. Breuckmann, Margarita Davydova, Jens~N. Eberhardt, and Nathanan
  Tantivasadakarn.
\newblock Cups and {Gates} {I}: {Cohomology} invariants and logical quantum
  operations, October 2024.
\newblock arXiv:2410.16250.

\bibitem[BE21]{breuckmann_balanced_2021}
Nikolas~P. Breuckmann and Jens~N. Eberhardt.
\newblock Balanced {Product} {Quantum} {Codes}.
\newblock {\em IEEE Transactions on Information Theory}, 67(10):6653--6674,
  October 2021.
\newblock Conference Name: IEEE Transactions on Information Theory.

\bibitem[BH12]{bravyi_magic-state_2012}
Sergey Bravyi and Jeongwan Haah.
\newblock Magic-state distillation with low overhead.
\newblock {\em Physical Review A}, 86(5):052329, November 2012.
\newblock Publisher: American Physical Society.

\bibitem[BH14]{bravyi_homological_2014}
Sergey Bravyi and Matthew~B. Hastings.
\newblock Homological product codes.
\newblock In {\em Proceedings of the forty-sixth annual {ACM} symposium on
  {Theory} of computing}, {STOC} '14, pages 273--282, New York, NY, USA, May
  2014. Association for Computing Machinery.

\bibitem[BMD07a]{bombin_exact_2007}
H.~Bombin and M.~A. Martin-Delgado.
\newblock Exact topological quantum order in \${D}=3\$ and beyond: {Branyons}
  and brane-net condensates.
\newblock {\em Physical Review B}, 75(7):075103, February 2007.
\newblock Publisher: American Physical Society.

\bibitem[BMD07b]{bombin_topological_2007}
H.~Bombin and M.~A. Martin-Delgado.
\newblock Topological {Computation} without {Braiding}.
\newblock {\em Physical Review Letters}, 98(16):160502, April 2007.
\newblock Publisher: American Physical Society.

\bibitem[Bom15]{bombin_gauge_2015}
Héctor Bombín.
\newblock Gauge color codes: optimal transversal gates and gauge fixing in
  topological stabilizer codes.
\newblock {\em New Journal of Physics}, 17(8):083002, August 2015.
\newblock Publisher: IOP Publishing.

\bibitem[CFG13]{candes_super-resolution_2013}
Emmanuel~J. Candès and Carlos Fernandez-Granda.
\newblock Super-{Resolution} from {Noisy} {Data}.
\newblock {\em Journal of Fourier Analysis and Applications}, 19(6):1229--1254,
  December 2013.

\bibitem[CKPS16]{chen_fourier-sparse_2016}
Xue Chen, Daniel~M. Kane, Eric Price, and Zhao Song.
\newblock Fourier-{Sparse} {Interpolation} without a {Frequency} {Gap}.
\newblock In {\em 2016 {IEEE} 57th {Annual} {Symposium} on {Foundations} of
  {Computer} {Science} ({FOCS})}, pages 741--750, New Brunswick, NJ, USA,
  October 2016. IEEE.

\bibitem[CL18]{cuyt_multivariate_2018}
Annie Cuyt and Wen-shin Lee.
\newblock Multivariate exponential analysis from the minimal number of samples.
\newblock {\em Advances in Computational Mathematics}, 44(4):987--1002, August
  2018.

\bibitem[CRT06]{candes_robust_2006}
E.J. Candes, J.~Romberg, and T.~Tao.
\newblock Robust uncertainty principles: exact signal reconstruction from
  highly incomplete frequency information.
\newblock {\em IEEE Transactions on Information Theory}, 52(2):489--509,
  February 2006.
\newblock Conference Name: IEEE Transactions on Information Theory.

\bibitem[DHLV23]{dinur_good_2023}
Irit Dinur, Min-Hsiu Hsieh, Ting-Chun Lin, and Thomas Vidick.
\newblock Good {Quantum} {LDPC} {Codes} with {Linear} {Time} {Decoders}.
\newblock In {\em Proceedings of the 55th {Annual} {ACM} {Symposium} on
  {Theory} of {Computing}}, {STOC} 2023, pages 905--918, New York, NY, USA,
  June 2023. Association for Computing Machinery.

\bibitem[DKP23]{diederichs_how_2023}
Benedikt Diederichs, Mihail~N. Kolountzakis, and Effie Papageorgiou.
\newblock How many {Fourier} coefficients are needed?
\newblock {\em Monatshefte für Mathematik}, 200(1):23--42, January 2023.

\bibitem[DLV24]{dinur_expansion_2024}
Irit Dinur, Ting-Chun Lin, and Thomas Vidick.
\newblock Expansion of higher-dimensional cubical complexes with application to
  quantum locally testable codes, April 2024.
\newblock arXiv:2402.07476 [quant-ph].

\bibitem[dP95]{de_prony_essai_1795}
B.~G.~R. de~Prony.
\newblock Essai experimental et analytique sur les lois de la dilatabilite de
  fluides elastiques et sur celles da la force expansion de la vapeur de
  l'alcool, a differentes temperatures.
\newblock {\em Journal de l'Ecole Polytechnique}, 1(2), 1795.

\bibitem[GG24]{golowich_quantum_2024}
Louis Golowich and Venkatesan Guruswami.
\newblock Quantum {LDPC} {Codes} of {Almost} {Linear} {Distance} via
  {Homological} {Products}, November 2024.
\newblock arXiv:2411.03646 [quant-ph].

\bibitem[GG25]{golowich_asymptotically_2025}
Louis Golowich and Venkatesan Guruswami.
\newblock Asymptotically {Good} {Quantum} {Codes} with {Transversal}
  {Non}-{Clifford} {Gates}.
\newblock In {\em Proceedings of the 57th {Annual} {ACM} {Symposium} on
  {Theory} of {Computing}}, {STOC} '25, pages 707--717, New York, NY, USA, June
  2025. Association for Computing Machinery.

\bibitem[GHJY14]{gopalan_explicit_2014}
Parikshit Gopalan, Cheng Huang, Bob Jenkins, and Sergey Yekhanin.
\newblock Explicit {Maximally} {Recoverable} {Codes} {With} {Locality}.
\newblock {\em IEEE Transactions on Information Theory}, 60(9):5245--5256,
  September 2014.
\newblock Conference Name: IEEE Transactions on Information Theory.

\bibitem[GRS22]{guruswami_essential_2022}
Venkatesan Guruswami, Atri Rudra, and Madhu Sudan.
\newblock Essential coding theory.
\newblock {\em Draft available at http://www. cse. buffalo. edu/
  atri/courses/coding-theory/book}, 2022.

\bibitem[Has23]{hastings_quantum_2023}
M.~B. Hastings.
\newblock On {Quantum} {Weight} {Reduction}, July 2023.
\newblock arXiv:2102.10030 [quant-ph].

\bibitem[HVWZ25]{he_quantum_2025}
Zhiyang He, Vinod Vaikuntanathan, Adam Wills, and Rachel~Yun Zhang.
\newblock Quantum {Codes} with {Addressable} and {Transversal} {Non}-{Clifford}
  {Gates}, February 2025.
\newblock arXiv:2502.01864 [quant-ph].

\bibitem[KP23]{kalachev_two-sided_2023}
Gleb Kalachev and Pavel Panteleev.
\newblock Two-sided {Robustly} {Testable} {Codes}, August 2023.
\newblock arXiv:2206.09973 [cs, math].

\bibitem[KP25]{kalachev_maximally_2025}
Gleb Kalachev and Pavel Panteleev.
\newblock Maximally {Extendable} {Product} {Codes} are {Good} {Coboundary}
  {Expanders}, January 2025.
\newblock arXiv:2501.01411 [cs].

\bibitem[KPRvdO16]{kunis_multivariate_2016}
Stefan Kunis, Thomas Peter, Tim Römer, and Ulrich von~der Ohe.
\newblock A multivariate generalization of {Prony}'s method.
\newblock {\em Linear Algebra and its Applications}, 490:31--47, February 2016.

\bibitem[KT19]{krishna_towards_2019}
Anirudh Krishna and Jean-Pierre Tillich.
\newblock Towards {Low} {Overhead} {Magic} {State} {Distillation}.
\newblock {\em Physical Review Letters}, 123(7):070507, August 2019.

\bibitem[Lin24]{lin_transversal_2024}
Ting-Chun Lin.
\newblock Transversal non-{Clifford} gates for quantum {LDPC} codes on sheaves,
  October 2024.
\newblock arXiv:2410.14631 [quant-ph].

\bibitem[LZ22]{leverrier_quantum_2022-1}
Anthony Leverrier and Gilles Zémor.
\newblock Quantum {Tanner} codes.
\newblock In {\em 2022 {IEEE} 63rd {Annual} {Symposium} on {Foundations} of
  {Computer} {Science} ({FOCS})}, pages 872--883. IEEE Computer Society,
  October 2022.

\bibitem[Moi15]{moitra_super-resolution_2015}
Ankur Moitra.
\newblock Super-resolution, {Extremal} {Functions} and the {Condition} {Number}
  of {Vandermonde} {Matrices}.
\newblock In {\em Proceedings of the forty-seventh annual {ACM} symposium on
  {Theory} of {Computing}}, {STOC} '15, pages 821--830, New York, NY, USA, June
  2015. Association for Computing Machinery.

\bibitem[NC10]{nielsen_quantum_2010}
Michael~A. Nielsen and Isaac~L. Chuang.
\newblock {\em Quantum {Computation} and {Quantum} {Information}: 10th
  {Anniversary} {Edition}}.
\newblock Cambridge University Press, December 2010.

\bibitem[Ngu25]{nguyen_good_2025}
Quynh~T. Nguyen.
\newblock Good {Binary} {Quantum} {Codes} with {Transversal} {CCZ} {Gate}.
\newblock In {\em Proceedings of the 57th {Annual} {ACM} {Symposium} on
  {Theory} of {Computing}}, {STOC} '25, pages 697--706, New York, NY, USA, June
  2025. Association for Computing Machinery.

\bibitem[NP24]{nguyen_quantum_2024}
Quynh~T. Nguyen and Christopher~A. Pattison.
\newblock Quantum fault tolerance with constant-space and logarithmic-time
  overheads, November 2024.
\newblock arXiv:2411.03632 [quant-ph].

\bibitem[PK22a]{panteleev_asymptotically_2022}
Pavel Panteleev and Gleb Kalachev.
\newblock Asymptotically good {Quantum} and locally testable classical {LDPC}
  codes.
\newblock In {\em Proceedings of the 54th {Annual} {ACM} {SIGACT} {Symposium}
  on {Theory} of {Computing}}, {STOC} 2022, pages 375--388, New York, NY, USA,
  June 2022. Association for Computing Machinery.

\bibitem[PK22b]{panteleev_quantum_2022}
Pavel Panteleev and Gleb Kalachev.
\newblock Quantum {LDPC} {Codes} {With} {Almost} {Linear} {Minimum} {Distance}.
\newblock {\em IEEE Transactions on Information Theory}, 68(1):213--229,
  January 2022.
\newblock Conference Name: IEEE Transactions on Information Theory.

\bibitem[PK24]{panteleev_maximally_2024}
Pavel Panteleev and Gleb Kalachev.
\newblock Maximally {Extendable} {Sheaf} {Codes}, March 2024.
\newblock arXiv:2403.03651 [quant-ph].

\bibitem[PS94]{polishchuk_nearly-linear_1994}
Alexander Polishchuk and Daniel~A. Spielman.
\newblock Nearly-linear size holographic proofs.
\newblock In {\em Proceedings of the twenty-sixth annual {ACM} symposium on
  {Theory} of {Computing}}, {STOC} '94, pages 194--203, New York, NY, USA, May
  1994. Association for Computing Machinery.

\bibitem[PS15]{price_robust_2015}
Eric Price and Zhao Song.
\newblock A {Robust} {Sparse} {Fourier} {Transform} in the {Continuous}
  {Setting}.
\newblock In {\em 2015 {IEEE} 56th {Annual} {Symposium} on {Foundations} of
  {Computer} {Science}}, pages 583--600, Berkeley, CA, USA, October 2015. IEEE.

\bibitem[Sau18]{sauer_pronys_2018}
Tomas Sauer.
\newblock Prony's method in several variables: {Symbolic} solutions by
  universal interpolation.
\newblock {\em Journal of Symbolic Computation}, 84:95--112, January 2018.

\bibitem[SPW24]{scruby_quantum_2024}
Thomas~R. Scruby, Arthur Pesah, and Mark Webster.
\newblock Quantum {Rainbow} {Codes}, August 2024.
\newblock arXiv:2408.13130 [quant-ph].

\bibitem[Vid15]{viderman_combination_2015}
Michael Viderman.
\newblock A combination of testability and decodability by tensor products.
\newblock {\em Random Structures \& Algorithms}, 46(3):572--598, 2015.
\newblock \_eprint: https://onlinelibrary.wiley.com/doi/pdf/10.1002/rsa.20498.

\bibitem[WHY24]{wills_constant-overhead_2024}
Adam Wills, Min-Hsiu Hsieh, and Hayata Yamasaki.
\newblock Constant-{Overhead} {Magic} {State} {Distillation}, August 2024.
\newblock arXiv:2408.07764 [quant-ph].

\bibitem[WLH24]{wills_tradeoff_2024}
Adam Wills, Ting-Chun Lin, and Min-Hsiu Hsieh.
\newblock Tradeoff {Constructions} for {Quantum} {Locally} {Testable} {Codes},
  January 2024.
\newblock arXiv:2309.05541 [quant-ph].

\bibitem[Zhu25a]{zhu_topological_2025}
Guanyu Zhu.
\newblock A topological theory for {qLDPC}: non-{Clifford} gates and magic
  state fountain on homological product codes with constant rate and beyond the
  \${N}{\textasciicircum}\{1/3\}\$ distance barrier, January 2025.
\newblock arXiv:2501.19375 [quant-ph].

\bibitem[Zhu25b]{zhu_transversal_2025}
Guanyu Zhu.
\newblock Transversal non-{Clifford} gates on {qLDPC} codes breaking the
  \${\textbackslash}sqrt\{{N}\}\$ distance barrier and quantum-inspired
  geometry with \${\textbackslash}mathbb\{{Z}\}\_2\$ systolic freedom, July
  2025.
\newblock arXiv:2507.15056 [quant-ph].

\bibitem[ZP20]{zeng_minimal_2020}
Weilei Zeng and Leonid~P. Pryadko.
\newblock Minimal distances for certain quantum product codes and tensor
  products of chain complexes.
\newblock {\em Physical Review A}, 102(6):062402, December 2020.
\newblock arXiv:2007.12152 [math-ph, physics:quant-ph].

\bibitem[ZSP{\etalchar{+}}23]{zhu_non-clifford_2023}
Guanyu Zhu, Shehryar Sikander, Elia Portnoy, Andrew~W. Cross, and Benjamin~J.
  Brown.
\newblock Non-{Clifford} and parallelizable fault-tolerant logical gates on
  constant and almost-constant rate homological quantum {LDPC} codes via higher
  symmetries, October 2023.
\newblock arXiv:2310.16982 [cond-mat, physics:hep-th, physics:quant-ph].

\end{thebibliography}

\appendix

\section{Omitted Proofs from Preliminaries}
\label{sec:peproofs}
This appendix provides proofs that were omitted from Section~\ref{sec:prelim}.

Below, we present the proof of Lemma~\ref{lem:homvanexp}. Here, we carry over the notation $C^{(i)}$, $C^{(i,j)}$, $|\cdot|_i$ from Section~\ref{sec:pe}.

\begin{proof}[Proof of Lemma~\ref{lem:homvanexp}]
  The lemma essentially follows from the discussion in \cite[Appendix~B]{kalachev_two-sided_2023}, though we include a proof for completeness. At a high level, \cite{kalachev_two-sided_2023} provide an interpretation of $C_1\boxplus\cdots\boxplus C_t$ using a tensor product of chain complexes associated to the codes $C_1,\dots,C_t$. The lemma follows by applying the K\"{u}nneth formula to this product, and then tracing through the definitions.

  Note that here we use ordinary (multi-term) chain complexes as opposed to single-sector complexes; see \cite[Appendix~B]{kalachev_two-sided_2023} for the relevant background and definitions, specifically pertaining to tensor products and the K\"{u}nneth formula for such complexes.
  
  For $i\in[t]$, let $k_i=\dim C_i$, and define a 2-term cochain complex
  \begin{equation*}
    \cC_i=(\bF_q^{k_i}\xrightarrow{G_i}\bF_q^{n_i}),
  \end{equation*}
  where the coboundary map is a generator matrix $G_i\in\bF_q^{n_i\times k_i}$ for $C_i$, meaning that $C_i=\im G_i$. Then let
  \begin{equation*}
    \cA^* = (A^0 \xrightarrow{\delta_0} A^1 \xrightarrow{\delta_1} \cdots \xrightarrow{\delta_{t-1}} A^t)
  \end{equation*}
  be the $(t+1)$-term cochain complex given by the tensor product (see e.g.~\cite[Appendix~B]{kalachev_two-sided_2023}) of $\cC_1,\dots,\cC_t$, that is
  \begin{equation*}
    \cA = \cC_1 \otimes \cdots \otimes \cC_t.
  \end{equation*}
  The well-known K\"{u}nneth formula implies that the cohomology $H^i(\cA)$ vanishes for all $0\leq i\leq t-1$, so in particular $H^{t-1}(\cA)=0$.

  For $0\leq i\leq t$, the space $A^i$ is the direct sum of ${t\choose i}$ spaces of $t$-dimensional tensors, that is,
  \begin{equation}
    \label{eq:Aistructure}
    A^i \cong \bigoplus_{S\subseteq[t]:|S|=i}\left(\bigotimes_{j\in S}\bF_q^{n_j}\right)\otimes\left(\bigotimes_{j\in[t]\setminus S}\bF_q^{k_j}\right),
  \end{equation}
  where we write $\cong$ instead of $=$ above because the $t$ factors in the tensor products above should be reordered to go in increasing order by $j$. Letting $G^{(i)}=I^{\otimes i-1}\otimes G_i\otimes I^{t-i}$ denote the map that applies $G_i$ to all the direction-$i$ vectors in a $t$-dimensional tensor, then the coboundary maps of $\cA$ by definition consist of sums of maps $G^{(i)}$ with appropriate signs. More details on the chain complex $\cA$ can be found in Appendix~B of~\cite{kalachev_two-sided_2023} and in Section~4 of~\cite{dinur_expansion_2024} (though note that these works use different signing conventions for the coboundary maps).

  Now for $i\in[t]$, let $a_i={G^{(i)}}^{-1}(-1)^{i-1}(c_i-c_i')$ be the $t$-dimensional tensor obtained by applying $G_i^{-1}$ to every direction-$i$ column of $(-1)^{i-1}(c_i-c_i')$; note that as every direction-$i$ column of $c_i-c_i'$ lies in $C_i=\im G_i$, these inverses are (uniqely) well-defined. Then by definition $a:=(a_1,\dots,a_t)$ is an element of $A^{t-1}$, and
  \begin{equation*}
    \delta_{t-1}a = \sum_{i\in[t]}(-1)^{i-1}(-1)^{i-1}(c_i-c_i') = c-c = 0.
  \end{equation*}
  Thus because $H^{t-1}(\cA)=0$ as shown above, there must exist some $b\in A^{t-2}$ with $a=\delta_{t-2}b$. It follows from~(\ref{eq:Aistructure}) that $b$ is a collection of tensors $b=(b_{i,j})_{1\leq i<j\leq t}$, where $b_{i,j}$ has length $k_i$ and $k_j$ in directions $i$ and $j$ respectively, and length $n_\ell$ in all other directions $\ell\in[t]\setminus\{i,j\}$. Furthermore, for $i\in[t]$, by definition
  \begin{align*}
    a_i &= (\delta_{t-2}b)_i = \sum_{j=1}^{i-1}(-1)^{j-1}G^{(j)}b_{j,i} + \sum_{j=i+1}^t(-1)^{j-2}G^{(j)}b_{i,j}.
  \end{align*}
  Applying $(-1)^{i-1}G^{(i)}$ to both sides of the above equation, we conclude that
  \begin{equation*}
    c_i-c_i' = (-1)^{i-1}G^{(i)}a_i = \sum_{j=1}^{i-1}(-1)^{(j-1)+(i-1)}G^{(j)}G^{(i)}b_{j,i} + \sum_{j=i+1}^t(-1)^{(i-1)+(j-2)}G^{(i)}G^{(j)}b_{i,j}.
  \end{equation*}
  Thus for all $1\leq i<j\leq t$, if we let
  \begin{equation*}
    c_{i,j} = (-1)^{(i-1)+(j-1)}G^{(i)}G^{(j)}b_{i,j},
  \end{equation*}
  then $c_{i,j}\in C^{(i,j)}$ and~(\ref{eq:homvanexp}) holds, as desired.
\end{proof}

Below, we present the proof of Lemma~\ref{lem:transdef}.

\begin{proof}[Proof of Lemma~\ref{lem:transdef}]
  We prove the result for the $CCZ_q$ gate; the proof for $U_q$ is analogous, so we omit the details to avoid redundancy. Letting $p$ denote the characteristic of $q$, then for every $z^1,\dots,z^r\in\bF_q^\ell$, we have
  \begin{align*}
    \hspace{1em}&\hspace{-1em} \bigotimes_{j\in[n]}C^{r-1}Z_q^{a_j}(\Enc_{\bC}^1\otimes\cdots\Enc_{\bC}^r)\ket{z^1,\dots,z^r} \\
    &= \left(\prod_{h\in[r]}\frac{1}{\sqrt{|Q^h_Z\cap{Q^h_X}^\perp|}}\right) \sum_{({z^{(h)}}'\in\Enc_{\bC}^h(z^h))_{h\in[r]}} e^{2\pi i\tr_{\bF_q/\bF_p}(\sum_{j\in[n]}a_j\cdot{z_j^1}'\cdots{z_j^r}')/p} \ket{{z^1}',\dots,{z^r}'} \\
    &= \left(\prod_{h\in[r]}\frac{1}{\sqrt{|Q^h_Z\cap{Q^h_X}^\perp|}}\right) \sum_{({z^{(h)}}'\in\Enc_{\bC}^h(z^h))_{h\in[r]}} e^{2\pi i\tr_{\bF_q/\bF_p}(\sum_{j\in[\ell]}z_j^1\cdots z_j^r)/p} \ket{{z^1}',\dots,{z^r}'} \\
    &= e^{2\pi i\tr_{\bF_q/\bF_p}(\sum_{j\in[\ell]}z_j^1\cdots z_j^r)/p} (\Enc_{\bC}^1\otimes\cdots\Enc_{\bC}^r) \ket{z^1,\dots,z^r} \\
    &= (\Enc_{\bC}^1\otimes\cdots\Enc_{\bC}^r) \bigotimes_{j\in[\ell]}\left(e^{2\pi i\tr_{\bF_q/\bF_p}(z_j^1\cdots z_j^r)/p}\ket{z_j^1,\dots,z_j^r}\right) \\
    &= (\Enc_{\bC}^1\otimes\cdots\Enc_{\bC}^r)\left(\bigotimes_{j\in[\ell]}C^{r-1}Z_q\right)\ket{z^1,\dots,z^r},
  \end{align*}
  as desired, where the second equality above holds by~(\ref{eq:transdef}).
\end{proof}

Below we sketch the proof of Lemma~\ref{lem:alphred}; similar proofs can be found in \cite{nguyen_good_2025,golowich_quantum_2024}.

\begin{proof}[Proof sketch of Lemma~\ref{lem:alphred}]
  If we restrict attention to the $\bF_{q'}$-structure within $\bF_q$, and view $Q^h_X,Q^h_Z$ as subspaces of $\bF_{q'}^{en}$, then $Q^h=(Q^h_X,Q^h_Z)$ is by definition an $[[e\cdot n,e\cdot k_h,d_h]]_{q'}$ subsystem code of locality $\leq e\cdot w_h$. To obtain a transversal $C^{r-1}Z_{q'}$ gate, for $h\in[r]$ we concatenate this code with a length-$e^{r-1}$ classical repetition code (in the computational basis), and then permute the resulting physical qudits by an appropriate permutation $\pi_h:[ne^r]\rightarrow[ne^r]$ (defined below) to obtain our desired $[[n'=e^r\cdot n,\;k_h'=e\cdot k_h,\;d_h'\geq d_h]]_{q'}$ subsystem code ${Q'}^h$ of locality $w_h'\leq e^r\cdot w_h\cdot$.

  Let $(\Enc^h)_{h\in[r]}$ and $a\in\bF_q^n$ denote the encoding functions and coefficients vector respectively for the transversal $C^{r-1}Z$ gate on $(Q^h)_{h\in[r]}$. We then define our encoding function ${\Enc'}^h$ for ${Q'}^h$ to first apply $\Enc^h$ to the input $z^h\in\bF_{q'}^\ell$ (viewing $z^h$ as a vector in $\bF_q^\ell$ via the natural inclusion $\bF_{q'}\subseteq\bF_q$), repeat each component of every vector in the resulting coset $e^{r-1}$ times, and then apply the permutation $\pi_h$.
  
  The key observation now is that that multiplication over $\bF_q$ is an $e$-multilinear operation over $\bF_{q'}$. Therefore for every $z^1,\dots,z^r\in\bF_{q'}^\ell$, the transversal $C^{r-1}Z_q$ condition~(\ref{eq:transdef}) implies that $\sum_{j\in[\ell]}z_j^i$ equals a degree-$r$ multilinear polynomial in the components of ${z^h}'\in\Enc^h(z^h)$ for $h\in[r]$, where we now view ${z^h}'\in\bF_q^n$ as a length-$ne$ vector in $\bF_{q'}^{ne}$. Furthermore, by definition each variable ${z_j^h}'$ for $j\in[n]\times[e]$ participates in $\leq e^{r-1}$ monomials in this polynomial, as ${z_j^h}'$ for $j=(j_1,j_2)\in[n]\times[e]$ can only participate in monomials involving other variables with the same value of $j_1$, and each monomial contains exactly one variable with each value of $h\in[r]$.

  Hence because the encoding ${\Enc'}^h$ simply repeats the value of each variable ${z_j^h}'$ a total of $e^{r-1}$ times, if we instead consider ${z^h}'\in{\Enc'}^h(z^h)$, then we can use a different copy of each variable in each monomial, and thereby express $\sum_{j\in[\ell]}z_j^1\cdots z_j^r$ as a degree-$r$ multilinear polynomial in the components of ${z_j^h}'$ such that each variable appears in at most one monomial. Then by choosing the permutations $\pi_h$ appropriately, we can express this multilinear polynomial in the form $\sum_{j\in[ne^r]}a_j'\cdot{z_j^1}'\cdots{z_j^r}'$ for some appropriate coefficients $a_j'\in\bF_{q'}$, and hence we obtain the desired equality $\sum_{j\in[\ell]}z_j^1\cdots z_j^r=\sum_{j\in[ne^r]}a_j'\cdot{z_j^1}'\cdots{z_j^r}'$.
\end{proof}

\section{Omitted Decoder Proof}
\label{sec:omitdec}

\begin{algorithm}[!t]
  \caption{\label{alg:qdectech}
    The core subroutine for the decoder in Proposition~\ref{prop:qdectech}. Here we define $\epsilon,n,Q^1,Q^2,Q_Z'$ as in Proposition~\ref{prop:qdectech}.
    Furthermore, for $i\in[2]$, we let $k_i=\dim({Q^i_X}^\perp)$, $k_i'=\dim(Q^i_Z)$, so that ${Q^i_X}^\perp=\evl_{E_i}(\bF_q[X]^{[0,k_i)})$, $Q^i_Z=\evl_{E_i}(\bF_q[X]^{[0,k_i')})$ are Reed-Solomon codes with an evaluation set $E_i\subseteq[n]$ of size $|E_i|=n$.
    As a shorthand we denote $X=(X_1,X_2)$, $x=(x_1,x_2)$, and $E=E_1\times E_2$.
  }
  
  \SetKwInOut{Input}{Input}
  \SetKwInOut{Output}{Output}

  \SetKwFunction{FnDecQuantum}{DecQuantum}
  \SetKwProg{Fn}{Function}{:}{}
  \Input{$c_0\in{Q^1_X}^\perp\boxplus Q^2_Z$ that is close to $Q_Z':=(Q^1_Z\otimes Q^2_Z)+(Q^1_X\otimes Q^2_X)^\perp$}
  \Output{$c'\in Q_Z'$ that is close to $c_0$ (by proof of Proposition~\ref{prop:qdectech})}

  \Fn{\FnDecQuantum{$c_0$}}{
    Let $f(X)=\sum_{j=(j_1,j_2)}f_jX^j\in\bF_q[X]^{[0,n)^2}$ be the unique polynomial with $c_0=\evl_E(f)$ \\ \label{li:qf}
    \For{$j_2\in\{k_2,\dots,k_2'-1\}$}{ \label{li:qfor2}
      Compute some $r_2^{j_2}\in\bF_q^{E_1}$ with $|r_2^{j_2}|\leq\epsilon n/25$ such that
      \begin{equation*}
        \evl_{E_1}\left(\sum_{j_1\in[n]}f_{(j_1,j_2)}X_1^{j_1}\right)-r_2^{j_2} \in Q^1_Z.
      \end{equation*} \label{li:qr2}
    }
    \KwRet{
      \begin{equation*}
        c' = c_0 - \sum_{j_2=k_2}^{k_2+s-1}r_2^{j_2}\otimes\evl_{E_2}(X_2^{j_2})
      \end{equation*}
    } \label{li:qcp}
  }
\end{algorithm}

\begin{proof}[Proof of Proposition~\ref{prop:qdectech}]
  Define $\alpha=\alpha(\epsilon)$ as in Theorem~\ref{thm:dualtensordec}, and then define
  \begin{align*}
    \delta &= \delta(\epsilon,\delta') = \frac{\rho\epsilon\delta'}{50\cdot\alpha(\epsilon)}.
  \end{align*}
  Observing that $Q_Z'\subseteq{Q^1_X}^\perp\boxplus Q^2_Z$, we first run the $\poly(n,q)$-time decoder in Theorem~\ref{thm:dualtensordec} on the input $c$ with respect to the Reed-Solomon codes ${Q^1_X}^\perp,Q^2_Z$, to obtain some $c_0\in{Q^1_X}^\perp\boxplus Q^2_Z $ with
  \begin{equation*}
    |c_0-c| \leq \alpha\cdot|c-\tilde{c}| \leq \alpha\delta n^2 = \frac{\rho\epsilon\delta'}{50} \cdot n^2.
  \end{equation*}
  We then run Algorithm~\ref{alg:qdectech} on input $c_0$. We will show that Algorithm~\ref{alg:qdectech} successfully outputs some $c'\in Q_Z'$ satisfying $|c-c'|\leq\delta' n^2$, and can be implemented in $\poly(n,q)$ time.

  Define all variables as in Algorithm~\ref{alg:qdectech}, including the integers $k_1\leq k_1'$, $k_2\leq k_2'$ and the evaluation set $E=E_1\times E_2$. Let $b_0=c_0-\tilde{c}\in{Q^1_X}^\perp\boxplus Q^2_Z$, so that
  \begin{equation*}
    |b_0| \leq |c_0-c|+|c-\tilde{c}| \leq \alpha\delta n^2+\delta n^2\leq 2\alpha\delta n^2 \leq \frac{\rho\epsilon\delta'}{25} \cdot n^2,
  \end{equation*}
  where we have used the fact that every $\alpha$-decoder by definition has $\alpha\geq 1$. Then by Theorem~\ref{thm:pe2RS}, we may express
  \begin{equation*}
    b_0 = \evl_E(h_1(X)+h_2(X))
  \end{equation*}
  for some $h_1(X)\in\bF_q[X]^{[0,n)\times[0,k_2')}$ and $h_2(X)\in\bF_q[X]^{[0,k_1)\times[0,n)}$ such that the sets
  \begin{align*}
    H_1 &= \{x_1\in E_1:h_1(x_1,X_2)\neq 0\in\bF_q[X_2]\} \\
    H_2 &= \{x_2\in E_2:h_2(X_1,x_2)\neq 0\in\bF_q[X_1]\}
  \end{align*}
  have size $|H_1|,|H_2|\leq\epsilon\delta' n/25$.
  
  Define $f(X)=\sum_jf_jX^j$ as in line~\ref{li:qf} of Algorithm~\ref{alg:qdectech}. By definition $f(X)\in\bF_q[X]^{([0,k_1)\times[0,n))\cup([0,n)\times[0,k_2'))}$, and furthermore $\evl_E(f-h_1-h_2)=c_0-b_0=\tilde{c}\in Q_Z'$, where by definition $Q_Z'=\evl_E(\bF_q[X]^S)$ for
  \begin{equation*}
    S := ([0,k_1)\times[0,n)) \cup ([0,n)\times[0,k_2)) \cup ([0,k_1')\times[0,k_2')).
  \end{equation*}
  Therefore letting $h_1(X)=\sum_{j=(j_1,j_2)}h_{1,j}X^j$, then for every $j_2\in\{k_2,\dots,k_2'-1\}$ we have
  \begin{align*}
    \left(\sum_{j_1\in[n]}f_{(j_1,j_2)}X_1^{j_1}\right) - \left(\sum_{j_1\in[n]}h_{1,(j_1,j_2)}X_1^{j_1}\right)
    &\in \bF_q[X_1]^{[0,k_1')}.
  \end{align*}
  Hence
  \begin{align*}
    \evl_{E_1}\left(\sum_{j_1\in[n]}f_{(j_1,j_2)}X_1^{j_1}\right) - \evl_{E_1}\left(\sum_{j_1\in[n]}h_{1,(j_1,j_2)}X_1^{j_1}\right)
    &\in \evl_{E_1}(\bF_q[X_1]^{[0,k_1')}) = Q^1_Z.
  \end{align*}
  For $x_1\in E_1$, by definition $\sum_{j_1\in[n]}h_{1,(j_1,j_2)}x_1^{j_1}$ is precisely the degree-$j_2$ coefficient of the polynomial $h_1(x_1,X_2)$. But by definition $h_1(x_1,X_2)=0$ for every $x_1\in E_1\setminus H_1$, and hence $\evl_{E_1}(\sum_{j_1\in[n]}h_{1,(j_1,j_2)}X_1^{j_1})$ is supported inside $H_1$, and therefore has weight at most $|H_1|\leq\epsilon\delta' n/25\leq\epsilon n/25$. Therefore $\evl_{E_1}(\sum_{j_1\in[n]}h_{1,(j_1,j_2)}X_1^{j_1})$ is a valid choice of $r_2^{j_2}$ in line~\ref{li:qr2}. Furthermore, because $Q^2_Z$ is a Reed-Solomon code of distance $\geq\epsilon n/2$, this choice of $r_2^{j_2}$ is unique.

  Thus we have shown that for every $j_2\in\{k_2,\dots,k_2'-1\}$, line~\ref{li:qr2} computes
  \begin{equation*}
    r_2^{j_2} = \evl_{E_1}\left(\sum_{j_1\in[n]}h_{1,(j_1,j_2)}X_1^{j_1}\right),
  \end{equation*}
  which is supported inside $H_1$. Thus the value $c'$ returned by line~\ref{li:qcp} by definition agrees with $c_0$, and therefore also with $\tilde{c}=c_0-b_0$, inside $(E_1\setminus H_1)\times(E_2\setminus H_2)$ so
  \begin{align*}
    |c-c'|
    &\leq |c-\tilde{c}|+|\tilde{c}-c'| \\
    &\leq \delta n^2+n|H_1|+n|H_2| \\
    &\leq \delta' n^2.
  \end{align*}
  Furthermore, by definition
  \begin{align*}
    c'
    &= \evl_E\left(f(X)-h_1(X)-h_2(X) + \sum_{j_1=0}^{n-1}\sum_{j_2=0}^{k_2-1}h_{1,j=(j_1,j_2)}X^j + \sum_{j_1=0}^{k_1-1}\sum_{j_2=0}^{n-1}h_{2,j=(j_1,j_2)}X^j\right) \\
    &= \tilde{c} + \evl_E\left(\sum_{j_1=0}^{n-1}\sum_{j_2=0}^{k_2-1}h_{1,j=(j_1,j_2)}X^j + \sum_{j_1=0}^{k_1-1}\sum_{j_2=0}^{n-1}h_{2,j=(j_1,j_2)}X^j\right) \\
    &\in Q_Z' + ({Q^1_X}^\perp\boxplus{Q^2_X}^\perp) \\
    &= Q_Z',
  \end{align*}
  as desired.

  It remains to be shown that Algorithm~\ref{alg:qdectech} runs in time $\poly(n,q)$. For this purpose, we observe that line~\ref{li:qf} simply must solve a system of $\poly(n)$ linear equations over $\bF_q$, which can be performed by Guassian elimination in time $\poly(n,q)$. Meanwhile, line~\ref{li:qr2} can simply run the Welch-Berlekamp decoder for the Reed-Solomon code $Q^1_Z$ (see e.g.~\cite{guruswami_essential_2022}), which also runs in time $\poly(n,q)$. Thus the entire running time is $\poly(n,q)$.
\end{proof}

\section{Error Correction of Quantum Code States}
\label{sec:ec}
In this section, we provide some more details on how to perform error correction on quantum codes, and in particular on subsystem product codes, at the level of quantum states and operators. Specifically, we first describe the precise operators (i.e.~parity checks) to meausure, and the classical processing to perform on the measurement outcomes, in order to correct errors on the code state; these details are mostly folklore, and we provide them for the reader less familiar with quantum error correction. We also describe how this error correction procedure can function even in the presence of measurement errors for subsystem products, which is of particular interest given that some subsystem codes can fail to support such fault-tolerance in the presence of measurement errors.

Specifically, in Section~\ref{sec:ecCSS} below, we describe the general error correction procedure for non-subsystem CSS codes, and then in Section~\ref{sec:ecsubsystem} we describe error correction for subsystem codes. Readers familiar with these topics may want to skip to Section~\ref{sec:ftdec}, where we describe how to leverage local testability to perform error correction on on certain subsystem codes, including subsystem product codes, in the presence of syndrome errors.

Recall that for a finite field $\bF_q$ of characteristic $p$ and an element $a\in\bF_q$, the $q$-ary Pauli $X_q^a$ and $Z_q^a$ matrices are unitary operators in $\bC^{\bF_q\times\bF_q}\cong\bC^{q\times q}$ that act on $\ket{y}\in\bC^q$ for $y\in\bF_q$ by
\begin{align*}
  X_q^a\ket{y} &= \ket{y+a} \\
  Z_q^a\ket{y} &= e^{2\pi i\tr_{\bF_q/\bF_p}(ay)/p}\ket{y}.
\end{align*}
From here on we assume $q$ is fixed, and omit the $q$ subscript. For $a\in\bF_q^n$ and $P\in\{X,Z\}$, we denote $P^a=\bigotimes_{i\in[n]}P^{a_j}$.

\subsection{Non-Subsystem CSS Codes}
\label{sec:ecCSS}
For a non-subsystem CSS code $Q=(Q_X,Q_Z)$, error correction may be performed by measuring the operators (called \emph{$X$-stabilizers}) $X^c$ for an $\bF_p$-generating set\footnote{That is, a set of elements whose $\bF_p$-span equals $Q_X^\perp$.} of $c\in Q_X^\perp$, as well as the operators (called \emph{$Q$-stabilizers}) $Z^c$ for an $\bF_p$-generating set of $c\in Q_Z^\perp$. Such generating sets of $Q_X^\perp$ and $Q_Z^\perp$ are specifically given by scalar multiples of rows of the parity-check matrices $H_X$ and $H_Z$ respectively. If the code state was corrupted by a Pauli error of the form $X^{e_X}Z^{e_Z}$ with $e_X,e_Z\in\bF_q^n$, these measurement outcomes provide the decoder with the \emph{syndromes} $s_X=H_Xe_Z$ and $s_Z=H_Ze_X$. Note that it suffices to consider Pauli errors, as non-Pauli errors can be decomposed into superpositions of such Pauli errors. Given the syndromes $s_X,s_Z$, the decoder then finds low-weight errors $e_Z',e_X'$ satisfying $s_X=H_Xe_Z'$ and $s_Z=H_Ze_X'$, and applies $X^{-e_X'}Z^{-e_Z'}$ to revert the effect of the errors.

In a non-subsystem CSS code, the orthogonality condition $Q_X^\perp\subseteq Q_Z$ implies that all the stabilizers commute, so they can be measured in any order, and the measurements can be repeated, which can allow for successful error correction even if some of the measurements are noisy, i.e.~return the incorrect outcome.

\subsection{Subsystem CSS Codes}
\label{sec:ecsubsystem}
For a subsystem CSS code $Q=(Q_X,Q_Z)$ as defined in Definition~\ref{def:subsystem}, error correction will again be performed by measuring the operators $X^c$ for a generating set of $c\in Q_X^\perp$, as well as the operators $Z^c$ for a generating set of $c\in Q_Z^\perp$. However, these operators are now called \emph{$X$-gauge operators} and \emph{$Z$-gauge operators} respectively, and the $X$-gauge operator may not commute with the $Z$-gauge operators.

Hence to perform error correction, we first measure $X$-gauge operators to obtain an $X$-syndrome $s_X=H_X(e_Z+g_Z)$, where $e_Z\in\bF_q^n$ is the low-weight $Z$-error that occured, and $g_Z\in Q_Z^\perp$ is an arbitrary $Z$-gauge operator; recall that because gauge qudits are not protected by the code's distance, we must always assume arbitrary gauge operators may be applied to the code state. We then find some low-weight $e_Z'\in\bF_q^n$ such that $s_X-H_Xe_Z'\in\im(H_X|_{Q_Z^\perp})$, and we apply $Z^{-e_Z'}$ to revert the effect of the error. We subsequently do the same procedure, but for $X$ errors. That is, we measure $Z$-gauge operators to obtain a $Z$-syndrome $s_Z=H_Z(e_X+g_X)$, where $e_X\in\bF_q^n$ is a low-weight $X$-error and $g_X\in Q_X^\perp$ is an arbitrary $X$-gauge operator. We then find some low-weight $e_X'\in\bF_q^n$ such that $s_Z-H_Ze_X'\in\im(H_Z|_{Q_X^\perp})$, and we apply $X^{-e_X'}$ to revert the effect of the error.

Note that because the $X$ and $Z$ gauge operators do not commute, the $X$ and $Z$ error correction procedures here may induce the action of arbitrary $X$ and $Z$ gauge operators respectively. Therefore even in the absence of any errors, the syndromes may be different in every round. As such, if measurement errors occur, we may not be able to simply repeat the measurements to determine the true syndrome values. Nevertheless, we show below how to successfully perform error correction in the presence of measurement errors of subsystem codes such constructed from locally testable classical codes; subsystem products (Definition~\ref{def:subhomprod}) provide an example of such subsystem codes.

\subsection{Fault-Tolerant Error Correction for Subsystem Codes via Local Testability}
\label{sec:ftdec}
In Lemma~\ref{lem:ftdec} below, we show that if a subsystem code $Q=(Q_X,Q_Z)$ has locally testable $Q_X,Q_Z$, then there exists a decoder that can receive a noisy syndrome, and still compute a close approximation to the true error on the code state, where the accuracy of this approximation is proportional to the level of syndrome noise divided by the soundness of $Q_X,Q_Z$. Here we aim to prove a general statement that applies to arbitrary subsystem codes with good soundness, so we focus on the information-theoretic decoding problem, and do not consider algorithmic efficiency.

\begin{lemma}
  \label{lem:ftdec}
  Let $Q=(Q_X,Q_Z)$ be a $[[n,k,d]]_q$ subsystem code such that $Q_X,Q_Z$ are classical locally testable codes of soundness $\rho$.
  Then there exists an $X$-decoder takes as input a noisy syndrome $s_Z\in\bF_q^{m_Z}$ of the form $s_Z=H_Z(e_X+g_X)+v_X$ for an $X$-error $e_X\in\bF_q^n$ of weight $|e_X|<d/4$, a gauge operator $g_X\in Q_X^\perp$, and a syndrome error $v_X\in\bF_q^{m_Z}$ of weight $|v_X|<\rho d m_Z/4n$, and returns (a representative of) the coset $e_X+f+Q_X^\perp$ for some $f\in\bF_q^n$ of weight $|f|\leq|v|\cdot n/\rho m_Z$.

  Similarly, there exists an analogous $Z$-decoder (with $X$ and $Z$ swapped above).


\end{lemma}
\begin{proof}
  The desired $X$-decoder simply first computes the nearest $s_Z'\in\im(H_Z)$ to $s_Z$, then computes any $e_X'\in\bF_q^n$ of weight $|e_X'|<d/2$ such that $s_Z'-H_Ze_X'\in\im(H_Z|_{Q_X^\perp})$, and returns $e_X'$.

  To see that this algorithm succeeds, first observe that by definition $|s_Z'-H_Z(e_X+g_X)|\leq|v_X|$, so by the definition of local testability, there exists $f\in\bF_q^n$ of weight $|f|\leq|v|\cdot n/\rho m_Z<d/4$ such that $s_Z'-H_Z(e_X+g_X)=H_Z(-f)$. Thus $e_X+f$ satisfies $s_Z'-H_Z(e_X+f)=-H_Zg_X$ and $|e_X+f|<|e_X|+|f|<d/2$, so $e_X+f$ is a valid choice of $e_X'$. Every other valid choice $e_X''$ of $e_X'$ by definition has $e_X''-(e_X+f)\in \ker(H_Z)+Q_X^\perp=Q_Z+Q_X^\perp$ and $|e_X''-(e_X+f)|<d$, so because $Q$ has distance $d$, we must have $e_X''-(e_X+f)\in Q_X^\perp$. Hence the algorithm returns a representative of the coset $e_X+f+Q_X^\perp$, as desired.
\end{proof}

\begin{remark}
  While we did not consider locality of the codes $Q_X,Q_Z$ in Lemma~\ref{lem:ftdec} above, in practice measuring gauge operators of larger locality will require larger syndrome extraction circuits, which may lead to higher syndrome noise rates.
\end{remark}


\begin{remark}
  Decoding in the presence of syndrome errors without repeated syndrome measurements is sometimes called ``single-shot decoding.'' Hence Lemma~\ref{lem:ftdec} can be interpreted as saying that a subsystem code $Q=(Q_X,Q_Z)$ is (information-theoretically) single-shot decodable if $Q_X,Q_Z$ are locally testable.
\end{remark}

We can apply Lemma~\ref{lem:ftdec} to the subsystem product of an arbitrary pair of CSS codes as follows. For $i\in[2]$, assume that $Q^i=(Q^i_X,Q^i_Z)$ is a length-$n$ (non-subsystem) CSS code. Let $H^i_X\in\bF_q^{k^i_X\times n}$, $H^i_Z\in\bF_q^{k^i_Z\times n}$ be full-rank parity-check matrices for $Q^i_X,Q^i_Z$ respectively, and define ${H^i_X}'=G^i_XH^i_X,\;{H^i_Z}'=G^i_ZH^i_Z$ for some full-rank generator matrices $G^i_X\in\bF_q^{(n+k^i_X)\times k^i_X},\;G^i_Z\in\bF_q^{(n+k^i_Z)\times k^i_Z}$ of linear-distance codes $\im(G^i_X),\im(G^i_Z)$, respectively. Therefore ${H^i_X}',{H^i_Z}'$ are still parity-check matrices for $Q^i_X,Q^i_Z$ respectively, but now the syndrome spaces $\im({H^i_X}')=\im(G^i_X),\;\im({H^i_Z}')=\im(G^i_Z)$ are linear-distance codes, meaning that every error that does not equal a codeword has a linear-weight syndrome.

Then letting $Q=(Q_X,Q_Z)=Q^1\otimes Q^2$ be the subsystem product, the length-$N=n^2$ tensor codes
\begin{equation*}
  Q_X=Q^1_X\otimes Q^2_X, \hspace{1em} Q_Z=Q^1_Z\otimes Q^2_Z
\end{equation*}
with respective parity-check matrices
\begin{equation}
  \label{eq:subprodpc}
  H_X = \begin{pmatrix}{H^1_X}'\otimes I\\I\otimes{H^2_X}'\end{pmatrix}, \hspace{1em} H_Z = \begin{pmatrix}{H^1_Z}'\otimes I\\I\otimes{H^2_Z}'\end{pmatrix}
\end{equation}
will necessarily be locally testable with locality $w=O(\sqrt{N})$ and soundness $\rho=\Omega(1)$, as shown in \cite[Claim~6.2]{viderman_combination_2015}\footnote{Specifically, \cite{viderman_combination_2015} shows that the tensor product $C_1\otimes C_2$ of two length-$n$ classical codes $C_1,C_2$ has constant soundness in the following sense: the probability that a randomly chosen row or column of a given matrix $c\in\bF_q^{n\times n}$ lies within the respective code $C_1$ or $C_2$ is at least $\Omega(d(c,C_1\otimes C_2)/n^2)$. This result directly implies constant soundness in the sense of Definition~\ref{def:MDS} for our parity-check matrices in~(\ref{eq:subprodpc}), because by construction every nonzero syndrome of the factor parity-check matrices ${H^i_X}',{H^i_Z}'$ has weight $\Theta(n)$.}. Hence if $Q$ has linear distance, such as for our codes in Corollary~\ref{cor:subpeRS}, then Lemma~\ref{lem:ftdec} implies that $Q$ is (information-theoretically) decodable from syndromes experiencing corruptions of up to linear weight $\Theta(N)$.


We can also apply Lemma~\ref{lem:ftdec} to higher-order subsystem products. Specifically, for $i\in[t]$ let $Q^i=(Q^i_X,Q^i_Z))$ be a length-$n$ (non-subsystem) CSS code, and let $Q=(Q_X,Q_Z)=\bigotimes_{i\in[t]}Q^i$ be the subsystem product (see Definition~\ref{def:subhomprod}). Then $Q_X,Q_Z$ are classical $t$-dimensional tensor product codes of length $N=n^t$, which by Theorem~\ref{thm:tensorltc} are locally testable\footnote{While Theorem~\ref{thm:tensorltc}, proven by \cite{viderman_combination_2015}, is only stated for the tensor product of $t$ copies of the same code, it also applies to products of different codes, as explained in \cite[Remark~4.5]{viderman_combination_2015}.} of locality $w\leq tn$ and soundness $\rho\geq\delta^{3t}/(tn)^{O(1)}$, where $\delta$ denotes the minimum relative distance of the classical codes $Q^i_X,Q^i_Z$ for
$i\in[t]$. Hence if $Q$ and all $Q^i_X,Q^i_Z$ have constant relative distance and $t$ is taken to be a large constant as $n\rightarrow\infty$, as is the case for the construction described in Remark~\ref{remark:subpemanyrandom}, then Lemma~\ref{lem:ftdec} implies that $Q$ is (information-theoretically) decodable from syndromes experiencing corruptions of almost-linear weight $N^{1-\epsilon}$, where $\epsilon\rightarrow 0$ as $t\rightarrow\infty$.



\end{document}